\documentclass[11pt,a4paper]{article}
\usepackage{graphicx}
\usepackage{ifpdf}
\usepackage{amsthm,amsfonts,amssymb,amsmath,euscript,comment,bbm,bm,mathabx}
\usepackage{authblk}
\usepackage[utf8]{inputenc}
\usepackage[T1]{fontenc}
\usepackage{color}
\usepackage{textcomp}
\usepackage{slashed}
\usepackage[colorlinks=true,linkcolor=red,citecolor=blue]{hyperref}
\usepackage{mathrsfs}
\usepackage{dsfont}
\usepackage{geometry}
\usepackage{epsfig}
\usepackage{tikz}
\usetikzlibrary{arrows.meta}
\usepackage{float}
\usetikzlibrary{decorations.pathmorphing}
\usetikzlibrary{decorations.markings}
\usetikzlibrary{decorations.pathreplacing}
\usepackage{tikz-3dplot}
\usepackage{xcolor} 
\definecolor{darkgreen}{rgb}{0.0, 0.5, 0.0}
\numberwithin{equation}{section}
\renewcommand{\c}{\cdot}
\newcommand{\UU}{\mathcal{U}}

\def\pr{\partial}
\def\ub{{\underline{u}}}

\newcommand{\omb}{\underline{\omega}}
\newcommand{\etab}{\underline{\eta}}
\def\sld{\slashed{d}}

\def\LL{\mathcal{L}}
\def\Fl{F_l}
\def\Fr{F_r}

\def\phit{\widetilde{\phi}}
\def\Phit{\widetilde{\Phi}}
\def\AAt{\widetilde{A}}
\def\Eg{\mathfrak{E}}

\def\Ssl{\slashed{\S}}
\def\Sslh{\widehat{\Ssl}}
\def\trchbt{\widetilde{\trchb}}

\def\Tk{\mathfrak{T}}
\def\T{\mathbf{T}}
\def\Tr{\mathbf{Tr}}
\def\FF{\mathcal{F}}
\def\FFb{\underline{\FF}}
\def\bbb{{\underline{b}}}
\def\Bg{\mathfrak{B}}
\def\W{\mathbf{W}}

\def\Egt{\widetilde{\Eg}}
\def\Bgt{\widetilde{\Bg}}
\def\AA{\mathcal{A}}
\def\BB{\mathcal{B}}
\def\bom{{\boldsymbol{\om}}}
\def\ins{\slashed{\in}}

\def\bu{\bullet}
\def\etabf{\bm{\eta}}
\def\psit{\widetilde{\psi}}

\def\slg{\slashed{g}}

\def\sdivs{\slashed{\sdiv}}
\def\curls{\slashed{\curl}}

\def\ik{\mathfrak{i}}
\def\af{a^\frac{1}{2}}

\def\RRR{\mathbb{R}}
\def\SSS{\mathbb{S}}
\def\S{\mathbf{S}}
\def\Hb{\underline{H}}
\def\aa{{\underline{\a}}}
\def\sk{\mathfrak{s}}
\def\th{\theta}
\def\cb{\mathbf{c}}
\def\At{\widetilde{\AA}}
\def\ev{\vec{e}}

\def\slJ{\widehat{J}}

\def\NNN{\mathbb{N}}
\def\dual{{}^*}
\def\Et{\widetilde{E}}
\def\Bt{\widetilde{B}}
\def\go{\mathring{g}}
\def\ko{\mathring{k}}

\def\Mt{\widetilde{M}}
\def\Nt{\widetilde{N}}
\def\bF{{}^{(F)}\hspace{-1.5pt}\b}
\def\bbF{{}^{(F)}\hspace{-1.5pt}\bb}

\def\rhoF{{}^{(F)}\hspace{-1.5pt}\rho}

\def\siF{{}^{(F)}\hspace{-1.5pt}\si}
\def\nabs{\slashed{\nab}}
\def\Lt{\widetilde{L}}
\def\Kt{\widetilde{K}}
\def\dag{\dagger}

\def\afd{a^{-\frac{1}{2}}}

\def\trcht{\widetilde{\trch}}

\def\Rl{R_{l}}
\def\Rr{R_{r}}
\def\Rkb{{\underline{\Rk}}}

\def\slg{{\slashed{g}}}
\def\afd{a^{-\frac{1}{2}}}
\def\XX{\mathcal{X}}
\def\Xt{\widetilde{X}}
\def\YY{\mathcal{Y}}

\def\Xb{\ov{X}}

\def\slgc{\widecheck{\slg}}
\def\ins{\slashed{\in}}
\def\Rk{\mathfrak{R}}
\def\Ok{\mathfrak{O}}
\def\cuvss{{H_{u_0}^{(\de,\ub)}}}
\def\ucuvss{{\Hb_\de^{(u_0,u)}}}
\newcommand{\cuvs}{{H_u^{(\de,\ub)}}}
\newcommand{\ucuvs}{{\Hb_\ub^{(u_0,u)}}}

\newcommand{\M}{\mathcal{M}}
\newcommand{\D}{\mathbf{D}}

\newcommand{\dk}{{\mathfrak{d}}}

\DeclareMathOperator{\sRic}{Ric}
\newcommand{\g}{\mathbf{g}}
\newcommand{\F}{\mathbf{F}}
\newcommand{\R}{{\mathbf{R}}}
\newcommand{\K}{\mathbf{K}}
\def\xja{\langle x\rangle}
\def\kt{\widetilde{k}}
\def\ep{\varepsilon}
\def\eps{\epsilon}
\def\ins{\slashed{\in}}
\def\hot{\widehat{\otimes}}

\def\La{\Lambda}
\def\Lb{\underline{L}}

\def\la{\lambda}

\def\MM{\mathcal{M}}
\def\om{\omega}
\def\ze{\zeta}
\def\Om{\Omega}
\def\bb{{\underline{\b}}}
\def\xb{\mathbf{x}}
\def\cb{\mathbf{c}}
\def\Si{\Sigma}
\def\si{\sigma}
\def\bnab{\bm{\nab}}

\def\dkb{\slashed{\dk}}
\def\ga{\gamma}
\def\Ga{\Gamma}
\def\xib{\underline{\xi}}
\def\a{\alpha}
\def\b{\beta}
\def\hch{\widehat{\chi}}
\def\hchb{\widehat{\chib}}
\def\trch{\tr\chi}
\def\trchb{\tr\chib}
\def\trchc{\widecheck{\trch}}
\def\trchbc{\widecheck{\trchb}}
\def\ka{\kappa}

\def\de{\delta}
\def\De{\Delta}
\def\nab{\nabla}
\def\ov{\overline}

\def\ho{\mathring{h}}
\def\pio{\mathring{\pi}}

\def\les{\lesssim}

\def\slA{\slashed{A}}
\newcommand{\blue}{\textcolor{blue}}
\def\f{\mathbf{f}}

\DeclareMathOperator{\curl}{curl}

\DeclareMathOperator{\grad}{grad}
\DeclareMathOperator{\supp}{supp}

\DeclareMathOperator{\tr}{tr}
\DeclareMathOperator{\sdiv}{div}
\def\Qo{\mathring{Q}}
\def\bdiv{\mathbf{Div}}
\def\Ric{\mathbf{Ric}}

\def\ef{\mathfrak{e}}
\def\A{\mathbf{A}}
\def\Asl{\slashed{A}}
\def\hb{\ov{h}}
\def\Xo{\mathring{X}}

\def\hti{\widetilde{h}}
\def\pit{\widetilde{\pi}}
\def\slep{\slashed{\ep}}

\def\gt{\widetilde{g}}

\def\QQ{\mathcal{Q}}
\def\E{\mathbf{E}}
\def\P{\mathbf{P}}
\def\C{\mathbf{C}}
\def\J{\mathbf{J}}
\def\Q{\mathbf{Q}}
\def\V{\mathbf{V}}
\def\chib{{\underline{\chi}}}
\def\T{\mathbf{T}}
\def\Eo{{\mathring{E}}}
\def\Bo{{\mathring{B}}}
\def\Ub{\underline{U}}
\def\Psisl{\slashed{\Psi}}
\def\Psit{\widetilde{\Psi}}
\def\cb{\mathbf{c}}
\newcommand{\x}{\mathbf{x}}
\def\y{\mathbf{y}}
\def\NN{\mathbb{N}}
\def\Nb{\underline{N}}
\newtheorem{thm}{Theorem}[section]
\newtheorem{prop}[thm]{Proposition}
\newtheorem{lem}[thm]{Lemma}
\newtheorem{cor}[thm]{Corollary}
\newtheorem{rk}[thm]{Remark}
\newtheorem{df}[thm]{Definition}

\newtheorem{proposition}[thm]{Proposition}

\newcommand{\red}{\textcolor{red}}
\newcommand{\green}{\textcolor{darkgreen}}
\allowdisplaybreaks[3]
\begin{document}
\title{Cauchy data for multiple collapsing boson stars}
 \author[1]{Elena Giorgi\footnote{elena.giorgi@columbia.edu}}
 \author[2]{Dawei Shen\footnote{ds4350@columbia.edu}}
  \author[3]{Jingbo Wan\footnote{jingbo.wan@sorbonne-universite.fr}}
\affil[1,2]{\small Department of Mathematics, Columbia University, New York, USA}
\affil[3]{\small Laboratoire Jacques-Louis Lions, Sorbonne Universit\'e, Paris, France}
\maketitle
\begin{abstract}
We construct Cauchy initial data for the Einstein-Maxwell-Klein-Gordon (EMKG) system, which evolves in finite time into spacetimes containing multiple trapped surfaces. From a physical perspective, this corresponds to preparing multiple well-separated boson stars, each of which collapses to form a spacelike black hole region. In particular, this extends the result of \cite{ShenWan2} on the formation of multiple trapped surfaces in vacuum to the EMKG system. 
\end{abstract}
{\hypersetup{linkcolor=black}\tableofcontents}
\section{Introduction}
The Einstein-Maxwell-Klein-Gordon (EMKG) system is a fully nonlinear and strongly coupled field theory that models the interaction of gravity, electromagnetism, and a charged scalar field. Bosonic matter bound by gravity is expected to form compact objects called \emph{boson stars}, which are modeled by solutions to the EMKG system, see already Section \ref{sec:physical} for more details. The EMKG system provides a natural framework for studying self-gravitating (charged) bosonic matter and, in particular, the dynamics and collapse of several interacting boson-star configurations.
Constructing asymptotically flat initial data describing multiple such objects is a fundamental problem in mathematical general relativity.

This paper constructs Cauchy initial data for the EMKG system, which evolves to produce \emph{multiple causally independent trapped surfaces}. Initial data for the EMKG system comprise of
\[
(\Si,g,k,E,B,\psi,\phi,A,\Phi),
\]
where $(\Si, g)$ is a Riemannian 3-manifold, $k$ is a symmetric 2-tensor, $E,B$ are electromagnetic vectorfields, $\psi, \phi$ are complex scalar functions as Cauchy data for the Klein-Gordon field, $A$ is the magnetic potential $1$--form, and $\Phi$ is the electric potential, a real scalar function on $\Sigma$ satisfying the EMKG constraint equations
\begin{equation}\label{EMKGconstraint:intro}
\begin{split}
R(g)+(\tr_g k)^2-|k|_g^2&=2\left(|E|_g^2+|B|_g^2+|D_0\psi|_g^2+|D\psi|_g^2\right), \\
\nab^j\!\left(k_{ij}-g_{ij}\tr_g k\right)&=2\Re\left((D_i\psi)^\dag D_0\psi\right)+2(E\times B)_i,\\
\sdiv_gE&=-2\ef\Im\left(\psi^\dag D_0\psi\right),\\
\sdiv_gB&=0,
\end{split}
\end{equation}
where $(E\times B)_i:={\in_{i}}^{jk}E_j B_k$ with $\in$ the volume element of $(\Si,g)$, $\ef$ is a coupling constant and
\[
D_0\psi:=\phi+\ik\ef \Phi \psi,\qquad\quad D_i\psi:=\pr_i\psi+\ik\ef A_i\psi.
\]
A simple version of our main theorem is stated as follows. A precise version of this theorem is stated later as Theorem~\ref{mainCauchy}.
\begin{thm}[Informal version]\label{maintheoremintro}
Fix $N\in \mathbb{N}$. On $\Sigma=\mathbb{R}^3$, there exist smooth asymptotically flat initial data that satisfy the EMKG constraints \eqref{EMKGconstraint:intro} and describe $N$ well-separated boson stars, whose evolution leads to the formation of $N$ causally independent black holes in finite time.
\end{thm}
The configuration in Theorem~\ref{maintheoremintro} is intrinsically an $N$--body problem: no symmetry reduction is available, and the construction must instead rely on well-separated, localized configurations that remain effectively decoupled at the level of the constraint equations. This localization is a fundamental geometric feature of the initial data we build. More precisely, consider charged Brill-Lindquist data $(g_{BL},0, E_{\mathrm{BL}}, 0)$ with $N$ well-separated poles $\{\cb_I\}_{I=1,\dots,N}$ and small individual masses. Then Theorem \ref{maintheoremintro} constructs an EMKG initial data set $(\Si,g,k,E,B,\psi,\phi,A,\Phi)$ solving \eqref{EMKGconstraint:intro}, free of trapped surfaces on $\Sigma$, such that for each center $\cb_I$ the following regions are present:
\begin{itemize}
\item \textbf{Euclidean inner core}, i.e. $(g,k,E,B)=(e,0,0,0)$ (where $e$ is the Euclidean metric) on the ball $\BB_{1-2\delta_I}(\cb_I)$ of radius $1-2\delta_I$ centered at $\cb_I$ for a fixed $0<\de_I\ll 1$;
\item \textbf{short-pulse annulus} on $\BB_1(\cb_I)\setminus\overline{\BB_{1-2\delta_I}(\cb_I)}$, which is obtained from the evolution of short-pulse characteristic initial data of length $\de_I$ and size $a_I$ producing a trapped surface in the future domain of dependence $D^+(\BB_1(\cb_I))$;
\item \textbf{barrier annulus} on $\BB_\frac{3}{2}(\cb_I)\setminus\ov{\BB_1(\cb_I)}$, where the data is $\varepsilon_0$--close to Euclidean with $\varepsilon_0\sim a_I^{-1}$.
\item \textbf{star shell and near-star annulus}, extending outward, we obtain a \emph{star shell} $\BB_2(\cb_I)\setminus\ov{\BB_\frac{3}{2}(\cb_I)}$ on which $(\psi,\phi)$ may be nonzero, and a \emph{near-star} annulus $\BB_4(\cb_I)\setminus\ov{\BB_2(\cb_I)}$ on which the Klein-Gordon fields vanish; therefore, the solution is electrovacuum. We call $\BB_2(\cb_I)$ the \emph{star region} centered at $\cb_I$.
\item \textbf{exterior Brill-Lindquist data}, i.e. $(g,k,E,B)=(g_{BL},0,k_{BL},0)$ on $\BB_{32}^c(\cb_I)$.
\end{itemize}
\begin{figure}[H]
\centering
\tdplotsetmaincoords{70}{120} 
\begin{tikzpicture}[scale=0.4, tdplot_main_coords]
\fill[blue!20, opacity=0.3] (-13,-12,0) -- (12,-12,0) -- (12,12,0) -- (-13,12,0) -- cycle; 
\node at (2,-9,0) {\scalebox{1}{$g_{\text{\tiny BL}},E_{\text{\tiny BL}}$}};
\newcommand{\AnnulusUnit}[4]{%
    \begin{scope}[shift={(#1,#2,0)}, scale=#3]
        \begin{scope}
            \fill[green!50!black!50, opacity=0.6] (0,0,0) ellipse (6.3 and 6.3);
            \clip (0,0,0) ellipse (4 and 4);
            \fill[white] (0,0,0) ellipse (4 and 4);
        \end{scope}    
\begin{scope}
    \fill[blue!40, opacity=0.8] (0,0,0) ellipse (4 and 4);
    \clip (0,0,0) ellipse (2 and 2);
    \fill[white] (0,0,0) ellipse (2 and 2); 
\end{scope}
\begin{scope}
    \fill[red!90, opacity=0.8] (0,0,0) ellipse (2 and 2);
    \clip (0,0,0) ellipse (1.5 and 1.5);
    \fill[white] (0,0,0) ellipse (1.5 and 1.5); 
\end{scope}
\fill[black] (0,0,0) circle (2pt);
\node[left] at (0,0,0) {\scalebox{0.7}{$\cb_{#4}$}};
\draw (0,0,0) ellipse (2 and 2); 
\draw (0,0,0) ellipse (1.5 and 1.5); 
\draw (0,0,0) ellipse (4 and 4);
\draw (0,0,0) ellipse (8 and 8);
\draw[->, thin, rounded corners=4pt] (0,-0.5,8) to[out=-90,in=90] (0,1,0);
\node[above] at (0.1,-0.7,8) {\scalebox{0.65}{Euclidean}};
\draw[->, thin, rounded corners=4pt] (0,2.8,11) to[out=-90,in=90] (0,1.9,0);
\node[above] at (0,2.8,11) {\scalebox{0.65}{\red{Short-pulse annulus}}};
\draw[->, thin, rounded corners=4pt] (0,-4,5) to[out=-90,in=90] (0,-3,0);
\node[above] at (0,-4,5) {\scalebox{0.65}{\blue{Barrier annulus}}};
\draw[->, thin, rounded corners=4pt] (0,-6.7,2.5) to[out=-90,in=90] (0,-5.3,0);
\node[above] at (0,-8.5,2.3) {\scalebox{0.65}{\green{Star shell}}};
\draw[->, thin, rounded corners=4pt] (0,-9,0.7) to[out=-90,in=90] (0,-7,0);
\node[above] at (0,-11,0.3) {\scalebox{0.65}{Near-star annulus}};
\draw[dashed] (0,0,-0.25) ellipse (2*11/12 and 2*11/12); 
\foreach \angle in {-40,110}
{\draw[dashed] (0,0,-3) -- ({2*cos(\angle)},{2*sin(\angle)},0);}
\end{scope}
}
\newcommand{\AnnulusUnitplain}[4]{%
    \begin{scope}[shift={(#1,#2,0)}, scale=#3]
    \begin{scope}
        \fill[green!50!black!50, opacity=0.6] (0,0,0) ellipse (6.5 and 6.5);
        \clip (0,0,0) ellipse (4 and 4);
        \fill[white] (0,0,0) ellipse (4 and 4);
    \end{scope}
\begin{scope}
    \fill[blue!40, opacity=0.8] (0,0,0) ellipse (4 and 4);
    \clip (0,0,0) ellipse (2 and 2);
    \fill[white] (0,0,0) ellipse (2 and 2); 
\end{scope}
\begin{scope}
    \fill[red!90, opacity=0.8] (0,0,0) ellipse (2 and 2);
    \clip (0,0,0) ellipse (1.5 and 1.5);
    \fill[white] (0,0,0) ellipse (1.5 and 1.5); 
\end{scope}
\fill[black] (0,0,0) circle (2pt);
\node[left] at (0,0.8,0) {\scalebox{0.5}{$\cb_{#4}$}};
\draw (0,0,0) ellipse (2 and 2); 
\draw (0,0,0) ellipse (1.5 and 1.5); 
\draw (0,0,0) ellipse (4 and 4);
\draw (0,0,0) ellipse (8 and 8);
\draw[dashed] (0,0,-0.25) ellipse (2*11/12 and 2*11/12); 
\foreach \angle in {-40,110}
{\draw[dashed] (0,0,-3) -- ({2*cos(\angle)},{2*sin(\angle)},0);}
\end{scope}
\begin{scope}[shift={(#1,#2)}, scale=#3]

    \draw[dashed] (0,6.5,0) -- (0,6.5,7);       
    \draw[dashed] (0,-6.5,0) -- (0,-6.5,7);   

    \draw[
        decorate,
        decoration={brace, amplitude=6pt, mirror}
    ] 
        (0,6.5,7) -- (0,-6.5,7)
        node[midway, yshift=0.55cm]{\scalebox{0.7}{Star region}};
\end{scope}
}
\AnnulusUnit{5}{5.4}{0.73}{1}
\AnnulusUnitplain{-8}{6}{0.45}{2}
\AnnulusUnitplain{-7}{-8}{0.36}{3}
\draw[->] (9,-8,0) -- (9,-8,9) node[anchor=south]{\scalebox{0.5}{$t$}};
\end{tikzpicture}
\caption{Geometric illustration of $(\Sigma,g,k,E,B,\psi,\phi,A,\Phi)$ from Theorem \ref{maintheoremintro} for $N=3$. From the center outward: Euclidean disk; short-pulse annulus (red); barrier annulus (blue); star shell (green); near-star annulus; Brill-Lindquist exterior. We refer the union of the Euclidean disk, short-pulse annulus, barrier annulus and star shell to be the star region.}
\label{IntroID}
\end{figure}
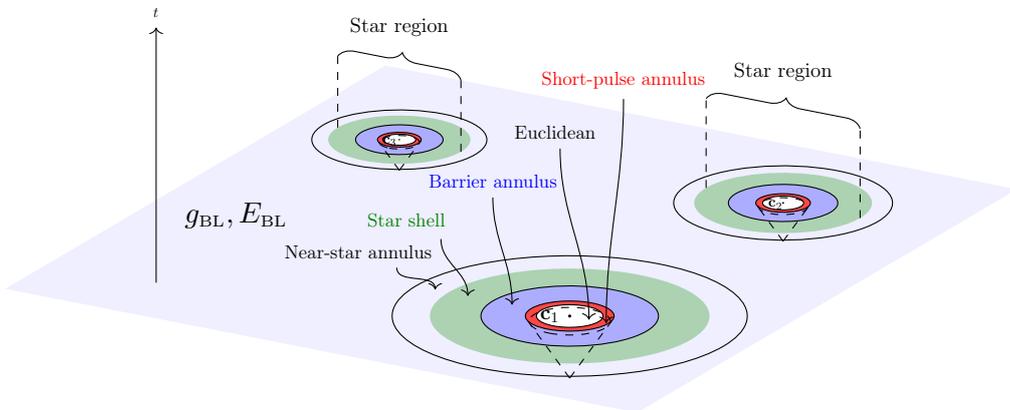

Our construction glues together three ingredients: (i) a short-pulse region that triggers trapped surface formation, (ii) a transition (\emph{barrier}) annulus where the geometry is a small perturbation of Euclidean space, and (iii) a charged Brill-Lindquist exterior carrying the desired ADM mass and electric charge near each center. The Klein-Gordon field is localized to the star region, so that the gluing reduces to an electrovacuum problem. 
Each part of the construction will be explored in more detail in Section \ref{sec:outline-proof}, where we provide a concise overview of the proof. 

\medskip 
Some remarks on Theorem \ref{maintheoremintro} are in order.

Since the short-pulse regions around different centers $\{\BB_1(\cb_I)\}_{I=1,\ldots,N}$ are mutually disjoint, the corresponding $N$ trapped surfaces $\{S_I\}_{I=1}^N$ are \emph{strongly causally independent}, in the sense that for $I\ne J$ and any spacelike fillings $M_I, M_J$ with $\pr M_I=S_I$ and $\pr M_J=S_J$, we have 
\begin{equation*}
    J^+(M_I) \cap J^-(M_J)=\emptyset,\qquad J^+(M_J)\cap J^-(M_I)=\emptyset.
\end{equation*}
By \cite{AnHan}, multiple MOTS must therefore form within a finite time.

The future evolution of the initial data constructed in Theorem~\ref{maintheoremintro} remains unknown. Understanding this evolution lies at the heart of the relativistic $N$--body problem and is closely related to the Final State Conjecture.  

Because the short-pulse method for the EMKG system requires a smallness condition on the coupling constant (see \cite{ShenWan} for the case of the Einstein-Maxwell-Charge Scalar Field (EMCSF) system), our construction only produces initial configurations in which the $N$ charged bodies are at rest and satisfy $|Q_I|\ll M_I,\ I=1,\dots,N.$ In this regime, the electromagnetic repulsion is dominated by the gravitational attraction, and it is therefore natural to conjecture that the future development of such data collapses to a single black hole spacetime formed by the eventual merger of the $N$ bodies.
This picture is very different from the Majumdar-Papapetrou (MP) family, which is a static solution to the Einstein-Maxwell equations. The MP spacetimes arise as the maximal Cauchy developments of the charged Brill-Lindquist initial data in the extremal regime $|Q_I|=M_I$. In this special case, the electromagnetic repulsion exactly balances the gravitational attraction, producing a static configuration of $N$ extremal black holes in equilibrium.  Such data represent a highly non-generic situation.

The recent breakthrough of Kehle-Unger \cite{KU} shows that an extremal Reissner-Nordstr\"om horizon can arise dynamically from regular asymptotically flat Cauchy data through a nonperturbative characteristic gluing scheme, providing, in particular, a definitive disproof of the classical third law of black hole thermodynamics. This suggests that a similar characteristic construction might be used to prepare individual extremal configurations in a multi-body setting. It is conceivable that one could combine such a characteristic gluing procedure, which creates the desired extremal behavior near each body, with a spacelike gluing scheme that attaches the resulting configuration to a Majumdar-Papapetrou exterior. This type of two-step construction could, in principle, produce Cauchy data whose future development approaches the extremal $N$--black hole equilibrium. We leave these possibilities for future investigation.
\subsection{Previous and related results}\label{sec:review}
A development closely related to the present work is the construction in \cite{ShenWan2} by the second and third named authors of smooth, asymptotically flat vacuum Cauchy data, whose future evolution produces multiple causally disjoint trapped surfaces. This result combines the short-pulse method with a multi-localized, obstruction-free gluing scheme, enabling several spatially separated collapse events to occur within a single initial data set. It extends the earlier construction of Li-Yu \cite{LY}, who produced a single trapped surface by matching it to a suitably chosen Kerr exterior, and provides the first rigorous realization of a vacuum $N$--body collapse scenario in which multiple black hole regions arise dynamically from regular spacelike initial data. For further background on trapped surface formation, stability analysis, and gluing methods for the constraint equations, we refer to Section 1.2 of \cite{ShenWan2}.

The EMKG system is the relativistic model for charged boson stars; see already the physical discussion in Section \ref{sec:physical}. On the analytical side, this system and its subsystems have been studied systematically through a series of nonlinear stability results. Following Christodoulou-Klainerman \cite{Ch-Kl}, Zipser \cite{zipser} proved the global stability of Minkowski spacetime for the Einstein-Maxwell system. In \cite{lo09}, Loizelet established stability for the Einstein-scalar-field-Maxwell system in $(n+1)$ dimensions for $n\ge 3$ by using the Lindblad-Rodnianski framework \cite{lr2}. Speck \cite{speck} obtained a related stability result for a coupled Einstein-electromagnetic system in $3+1$ dimensions. For the Einstein-Klein-Gordon system, stability was shown by LeFloch-Ma \cite{lefloch,leflochMa}, Wang \cite{wang}, and Ionescu-Pausader \cite{ionescu}. In Section \ref{secstability}, we establish a stability statement for the EMKG system in a finite region, sufficient for the construction presented in this paper.

The $N$--body problem in general relativity has been approached through different geometric and analytic mechanisms. In the vacuum setting, the Brill-Lindquist type gluing schemes of Anderson-Chru\'sciel-Pacard and Chru\'sciel-Corvino-Isenberg \cite{ACP,CCI} produce time-symmetric initial data with multiple localized ends, while the geometric construction of Chru\'sciel-Mazzeo \cite{CM} yields vacuum initial data sets with several apparent horizons, whose evolution has disconnected event-horizon intersections on the initial hypersurface. Hintz \cite{Hin23a,Hin24a,Hin24b,HintzIDS} developed a microlocal framework that permits the insertion of small black holes into a background spacetime, either directly at the level of initial data or along prescribed timelike geodesics. 

Static multi-black hole configurations arise through different balancing mechanisms. For the cosmic constant $\La=0$, the Majumdar-Papapetrou solutions demonstrate that electromagnetic repulsion in the Einstein-Maxwell system can balance gravitational attraction. For $\La>0$, Dias-Gibbons-Santos-Way \cite{DGSW23} constructed a fully vacuum static $2$--black hole solution stabilized by cosmic expansion, and rotating de Sitter binaries were subsequently obtained in \cite{DSW24}. Another construction in the de Sitter setting appears in Hintz's work \cite{Hin21}, where dynamical many-black-hole spacetimes are produced by gluing Kerr-de Sitter or Schwarzschild-de Sitter metrics near the future conformal boundary. These works demonstrate that multi-black hole and $N$--body configurations may arise through microlocal insertion, vacuum or EM gluing, electromagnetic balance, or cosmological expansion.

Further geometric results also relate to the structure of multi-black-hole initial data.  In the time-symmetric, nonnegative scalar curvature setting, the Riemannian Penrose inequality provides global constraints linking ADM mass to the area of the apparent horizon: Huisken-Ilmanen \cite{HI} established the connected-horizon case, and Bray \cite{Bray} proved the full multi-horizon version.  For charged configurations, the Majumdar-Papapetrou family furnishes static multi-black-hole solutions, and the corresponding rigidity in the Einstein-Maxwell system was shown by Chru\'sciel-Reall-Tod \cite{CRT}.  These geometric inputs provide complementary information about possible horizon geometries in multi-body configurations.

In contrast to incorporating preexisting black holes into a prescribed geometric setting or relying on static balance mechanisms, the present work follows the dynamical viewpoint initiated in \cite{ShenWan2} and constructs a $N$--body configuration directly at the level of \emph{Cauchy data} for the Einstein-Maxwell-Klein-Gordon system. Our initial data contain several collapsing charged boson stars, but no trapped surfaces or black hole regions; the resulting black holes form only in future development. This provides a matter-coupled and genuinely dynamical analog of the multi-black hole scenarios described above.
\subsection{Physical context and motivation}\label{sec:physical}
Here, we discuss the physical interpretation of our results, which clarifies the modeling choices and the role of the parameters that appear in the main statement.

A \emph{charged boson star} is a hypothetical compact object in General Relativity, formed from bosonic particles bound primarily by gravity. Unlike black holes, charged boson stars are completely regular---they possess neither a singularity nor a horizon. Charged boson stars are modeled by solutions to the EMKG system, derived from the action
\begin{align*}
\mathcal{S}= \int \big(\R+\mathcal{L}_{m}\big)\sqrt{-\g}\, d^4x,
\end{align*}
where $\R$ is the scalar curvature of spacetime, and the Lagrangian is defined by
\[
\LL_m:=-\D_\mu\psi\,(\D^\mu\psi)^\dag-V(|\psi|^2)-\frac{1}{4}\F_{\mu\nu}\F^{\mu\nu},
\]
with $\D_\mu:=\bnab_\mu+i\ef \A_\mu$ and $\ef$ as coupling constants. The potential $V(|\psi|^2)$ depends only on the field magnitude, consistent with the $U(1)$ symmetry of the Lagrangian, which yields, via Noether’s theorem, a conserved current representing the total number of bosonic particles. Varying the action gives the EMKG system
\begin{align}
\begin{split}\label{EMKG-phy}
\Ric_{\mu\nu}-\frac12\R\g_{\mu\nu}&=\T^{EM}_{\mu\nu}+\T^{KG}_{\mu\nu},\\
\bnab^\mu \F_{\mu\nu}&=-2\ef\Im\big(\psi (\D_\nu\psi)^\dag\big),\\
\g^{\mu\nu}\D_\mu\D_\nu\psi&=V'(|\psi|^2)\psi,
\end{split}
\end{align}
where the energy-momentum tensors are given by
\begin{align}
\begin{split}\label{Tdfeqintro}
\T^{EM}_{\mu\nu}&=2\left(\F_{\mu\a}{\F_{\nu}}^\a-\frac{1}{4}\g_{\mu\nu}\F_{\a\b} \F^{\a\b}\right),\\
\T^{KG}_{\mu\nu}&=2\left(\Re\left(\D_\mu\psi(\D_\nu\psi)^\dag\right)-\frac{1}{2} \g_{\mu\nu}\left( \D_\a \psi(\D^\a\psi)^\dag + V(|\psi|^2)\right)\right).
\end{split}
\end{align}
If the Maxwell field vanishes, the system becomes
\begin{align*}
\begin{split}
\Ric_{\mu\nu}-\frac12 \R\g_{\mu\nu} &= \T_{\mu\nu}^{KG},\\
\g^{\mu\nu}\bnab_\mu\bnab_\nu\psi &= V'(|\psi|^2)\psi.
\end{split}
\end{align*}
Depending on the scalar potential $V(|\psi|^2)$ and gauge coupling $\ef$, one obtains different families of boson stars:
\begin{itemize}
\item \textit{Free or ``mini'' boson stars} when $V(|\psi|^2)=m^2|\psi|^2$,
\item \textit{Self-interacting boson stars} when $$V(|\psi|^2)=m^2|\psi|^2+\frac{\lambda}{2}|\psi|^4$$ with $\la$ a dimensionless coupling constant, 
\item \textit{Charged boson stars} when the field couples to electromagnetism via $\D_\mu=\bnab_\mu+i\ef\A_\mu$, and the total electric charge coincides with the Noether charge associated with the global $U(1)$ symmetry. Stable configurations are known only for sufficiently small $\ef$ and suitable potentials.
\item \textit{Rotating and solitonic variants} obtained by including angular dependence or nonlinear potentials such as 
\[
V(|\psi|^2)=m^2|\psi|^2\!\left(1-\frac{|\psi|^2}{\psi_0^2}\right),
\]
or polynomial or periodic potentials, such as the \emph{axion} model
\[
V(|\psi|^2)=(m_a f_a)^2\!\left(1-\cos\frac{|\psi|^2}{f_a}\right),
\]
where $m_a$ and $f_a$ are fixed parameters.
\end{itemize}
In this work, we assume that $V$ is smooth near the vacuum and satisfies
\begin{equation}\label{assumptionV:intro}
V(0)=0, \qquad\quad |V^{(k)}(0)|\le m^2\ll 1\qquad \forall \; k\geq 1,
\end{equation}
so that $m$ acts as an effective small mass and the nonlinear terms remain perturbative on the finite time scales we consider. This covers the free and weakly self-interacting regimes, as well as the weak-coupling regimes of solitonic, polynomial, or axion-type models, and is compatible with the short-pulse framework. The smallness condition of $V$ introduced in \eqref{assumptionV:intro} is only used in the construction of short-pulse annuli and barrier annuli; see Sections \ref{sectrapped} and \ref{secstability}. The gluing arguments of the initial data are independent of \eqref{assumptionV:intro}.

Numerical and analytical evidence indicate that boson star perturbations may disperse, migrate to nearby equilibria, or collapse \cite{Liebling,LA23}. We address the last scenario and produce initial data for \emph{multiple} boson stars arranged such that each collapses to a trapped surface, with causal independence across the centers.
\subsection{Outline of the proof}\label{sec:outline-proof}
We sketch the six ingredients of the construction and indicate where the full statements appear.

\paragraph{Step 1: Short-pulse trapped surface mechanism.} 
In a double null foliation $(u,\ub)$, we adapt the short-pulse method to EMKG (under the Weyl gauge) with a small coupling constant $0<\ef\ll1$. Short-pulse data of size $a$ and width $0<\delta\leq a^{-2}$ on $H_{u_0}$, together with Minkowskian data on $\Hb_0$, produce a unique solution in a finite $u$--range whose terminal sphere $S_{-\frac{\delta a}{4},\delta}$ is trapped\footnote{Here, $H_u$ and $\Hb_\ub$ denote the outgoing and incoming null hypersurfaces, and $S_{u,\ub}=H_u\cap\Hb_\ub$.}, as indicated in Figure \ref{3Dshortpulseconeintro}. The $V$--dependent terms obey a smallness condition \eqref{assumptionV:intro} and fit into the familiar short-pulse hierarchy, extending the results of the second and third named authors \cite{ShenWan} for the EMCSF system to the EMKG system.
This step is obtained in Theorem \ref{thmtrapped}.
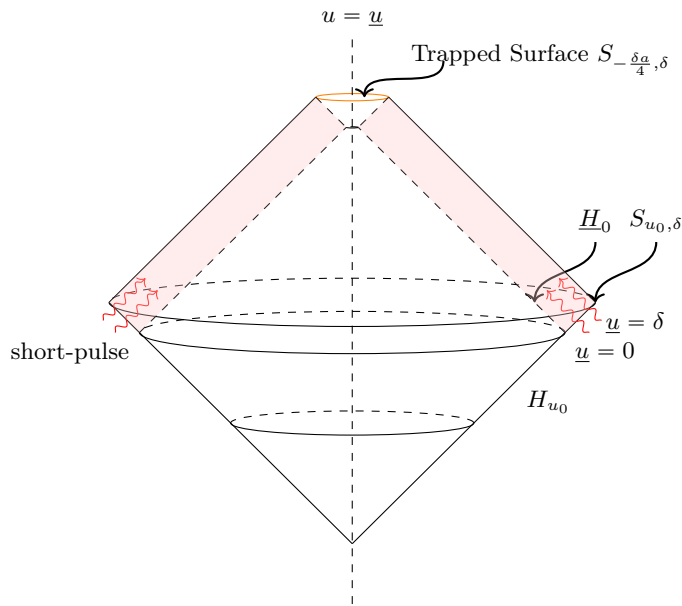
\begin{figure}[H]
\centering
\begin{tikzpicture}[scale=0.8, decorate]
    \draw[orange] (0,3.4) ellipse (0.6 and 0.06);
    \draw[->, thick, rounded corners=8pt] (1.5,4) to[out=-90, in=90] (0.2,3.4);
    \node[above] at (3.1,3.6) {\footnotesize Trapped Surface $S_{-\frac{\de a}{4},\de}$};
    \draw[dashed] (0.1,2.9) arc[start angle=0,end angle=180,x radius=0.1,y radius=0.01];
    \draw (-0.1,2.9) arc[start angle=180,end angle=360,x radius=0.1,y radius=0.01];
    \draw[dashed] (4,0) arc[start angle=0,end angle=180,x radius=4,y radius=0.4];
    \draw (-4,0) arc[start angle=180,end angle=360,x radius=4,y radius=0.4];
\begin{scope}
    \fill[white] 
        (3,0) arc[start angle=0,end angle=180,x radius=3,y radius=0.3] --
        (-3,0) arc[start angle=180,end angle=360,x radius=3,y radius=0.3] -- cycle;
\end{scope}
    \draw[dashed] (3.5,-0.5) arc[start angle=0,end angle=180,x radius=3.5,y radius=0.35];
    \draw (-3.5,-0.5) arc[start angle=180,end angle=360,x radius=3.5,y radius=0.35];
    \draw[dashed] (2,-2) arc[start angle=0,end angle=180,x radius=2,y radius=0.2];
    \draw (-2,-2) arc[start angle=180,end angle=360,x radius=2,y radius=0.2];
    \draw[dashed] (0,-5) -- (0,4.4) node[above]{\footnotesize $u=\ub$};
    \draw (0,-4) -- (4,0);
    \draw (-4,0) -- (0,-4);
    \draw (4,0) -- (0.6,3.4);
    \draw (-4,0) -- (-0.6,3.4);
    \draw[dashed] (3.5,-0.5) -- (0.1,2.9);
    \draw[dashed] (0.6,3.4) -- (0.1,2.9);
    \draw[dashed] (-0.6,3.4) -- (-0.1,2.9);
    \draw[dashed] (-3.5,-0.5) -- (-0.1,2.9);
    \draw[->, thick, rounded corners=8pt] (4,1) to[out=-90, in=90] (3,0);
    \node[above] at (4,1) {\footnotesize $\Hb_0$};
    \draw[->, thick, rounded corners=8pt] (5,1) to[out=-90, in=90] (4,0);
    \node[above] at (5,1) {\footnotesize $S_{u_0,\de}$};
    \node[below right] at (4,0) {\footnotesize $\ub=\de$};
    \node[below right] at (3.5,-0.5) {\footnotesize $\ub=0$};
    \node[below right] at (2.7,-1.3) {\footnotesize $H_{u_0}$};
\draw[->, decorate, decoration={snake, amplitude=0.5mm, segment length=2mm}, thin, red]  (4.1,-0.3) -- (3.4,0.4);
\draw[->, decorate, decoration={snake, amplitude=0.5mm, segment length=2mm}, thin, red]  (3.9,-0.5) -- (3.2,0.2);
\draw[->, decorate, decoration={snake, amplitude=0.5mm, segment length=2mm}, thin, red]  (-4.1,-0.3) -- (-3.4,0.4);
\draw[->, decorate, decoration={snake, amplitude=0.5mm, segment length=2mm}, thin, red]  (-3.9,-0.5) -- (-3.2,0.2);
\node[below left] at (-3.5,-0.5)  {\footnotesize short-pulse};
\fill[red!20,opacity=0.35](3.5,-0.5)--(0.1,2.9)--(0.6,3.4)--(4,0)--cycle;
\fill[red!20,opacity=0.35](-3.5,-0.5)--(-0.1,2.9)--(-0.6,3.4)--(-4,0)--cycle;
\end{tikzpicture}
\caption{Short-pulse cone. The red region is called the short-pulse region and the orange circle denotes the trapped surface $S_{-\frac{\de a}{4},\de}$.}
\label{3Dshortpulseconeintro}
\end{figure}
\paragraph{Step 2: Transition region and barrier annulus.}
We propagate the solution beyond the short-pulse region into a finite \emph{transition} zone where the geometry relaxes to an $\ep_0$--perturbation of Minkowski with $\ep_0= a^{-1}$. A constant-time slice $\Si$ through this region contains a short-pulse annulus bordered by a well-prepared \emph{barrier annulus}, as illustrated in Figure~\ref{fig:shortpulse+stabintro}, on which the data are sufficiently small and regular to serve as the setup for the gluing construction later.
This step is obtained in Theorems \ref{mainstability} and \ref{interiorsolution} as a finite-region stability result for the EMKG system. 
\begin{figure}[H]
\centering
\begin{tikzpicture}[scale=2.1, decorate]
  \coordinate (A) at (0,-2);
  \coordinate (B) at (0,0);
  \coordinate (C) at (0,0.8);
  \coordinate (D) at (0.1,0.7);
  \coordinate (E) at (0.2,0.8);
  \coordinate (F) at (1.4,-0.6);
  \coordinate (G) at (1.5,-0.5);
  \coordinate (H) at (2,0);
  \coordinate (I) at (2.5,0.5);
  \coordinate (J) at (1,0);
  \coordinate (K) at (0.5,0.5);
  \coordinate (L) at (1.5,1.5);
  \coordinate (M) at (2.5,0);
  \coordinate (N) at (0.8,0);
  \draw (A) -- (C);
  \draw (A) -- (I);
  \draw (F) -- (C);
  \draw (D) -- (E);
  \draw (G) -- (E);
  \draw (K) -- (L);
  \draw (I) -- (L);
  \draw (B) -- (M);
\fill[red!20,opacity=0.35](G) -- (E) -- (D) -- (F) -- cycle;
\fill[blue!30,opacity=0.35](G) -- (I) -- (L) -- (K) -- cycle;
  \node[right] at (1.35,-0.7) {\footnotesize $\ub = 0$};
  \node[right] at (G) {\footnotesize $\ub = \de$};
  \node[above] at (2.3,0.9) {\footnotesize $\ub = 1$};
  \node[above] at (0.75,0.9) {\footnotesize $u = -\frac{1}{2}$};
  \node[below right] at (M) {\footnotesize $\Si$};
  \node at (0.5, -0.8) {\footnotesize Minkowski};
  \node[right] at (0.8, -1.3) {\footnotesize $H_{-\frac{5}{4}+\de}$};
  \node[below right] at (1.6,-0.8) {\footnotesize \red{short-pulse}};
\filldraw[orange] (E) circle (0.5pt);
    \draw[->, thick, rounded corners=8pt] (-0.3,1.1) to[out=-90, in=90] (E);
    \node[above] at (-0.5,1.1) {\footnotesize {\color{orange}Trapped surface}};
    \draw[->, thick, rounded corners=8pt] (0.3,1.5) to[out=-90, in=90] (0.9,0);
    \node[above] at (0.2,1.5) {\footnotesize \red{Short-pulse annulus}};
    \draw[->, thick, rounded corners=8pt] (1.9,1.5) to[out=-90, in=90] (1.5,0);
    \node[above] at (2,1.5) {\footnotesize \blue{Barrier annulus}};
    \draw[->, decorate, decoration={snake, amplitude=0.5mm, segment length=2mm}, thin, red]  (1.7,-0.8) -- (1.3,-0.4);
    \draw[draw=blue, thick] (J) -- (H);
    \draw[draw=red, thick] (N) -- (J);
\end{tikzpicture}
    \caption{The red region marks the short-pulse domain, and the blue region indicates the transition zone. The hypersurface $\Si$ represents the target constant-time slice.}\label{fig:shortpulse+stabintro}
\end{figure}
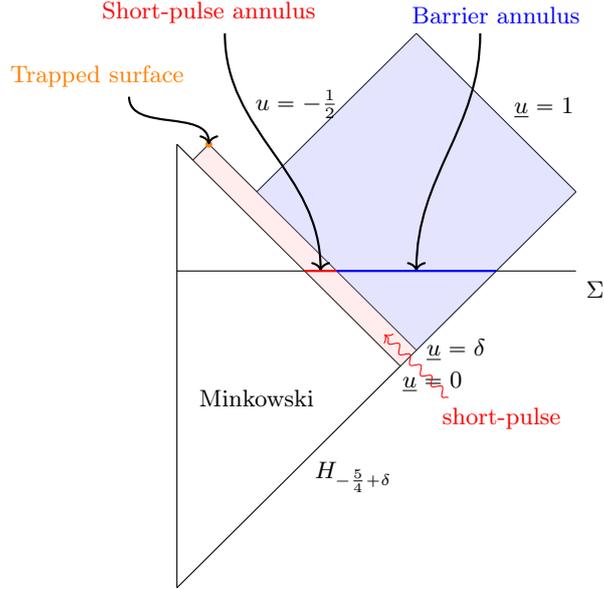
\paragraph{Step 3: Charged Brill-Lindquist exterior.} 
The constant-time slice $\Si$ obtained from Step 2 evolves into trapped surfaces within a finite time. To produce \emph{multiple} such trapped regions, the exterior geometry must be arranged so that, around each designated center, the data approximate that of a charged isolated mass. This role is fulfilled by the charged Brill-Lindquist family \cite{BL}, which provides multi-centered electrovacuum solutions to the constraint equations~\eqref{EMKGconstraint:intro}.
For $N$ poles with sufficiently small individual masses and large mutual separations, the charged Brill-Lindquist data $(g_{BL},0,E_{\mathrm{BL}},0)$ furnish an exterior region enclosing the barrier annuli constructed in Step 2. Near each center $\cb_I$, the averaged ADM mass, center of mass, and electric charge are close to $(M_I, M_I\cb_I, Q_I)$, while the remaining ADM components are small, and the magnetic charge vanishes. On suitably large annuli, the data remain $O(M_I)$--close to Euclidean in relevant Sobolev norms. In particular, each Brill-Lindquist end contributes the correct physical charges while already being nearly Euclidean in a sufficiently large annulus suitable for the subsequent gluing with the inner \emph{transition region}.
This step follows from the estimates in Section \ref{secBL-EM}.

\paragraph{Step 4: Star shell and near-star annulus.}
Extending outward from the barrier annulus of Step 2, we introduce a \emph{star shell} and a \emph{near-star annulus} that lie strictly between the barrier annulus and the Brill-Lindquist annulus of Step 3. These regions localize the Klein-Gordon field: the star shell supports the scalar field, while on the near-star annulus the Klein-Gordon field vanishes, leaving a purely Einstein-Maxwell regime. This arrangement ensures smooth compatibility for the subsequent gluing to the exterior Brill-Lindquist. The data remain $O(a^{-1})$--close to Euclidean with controlled charges and continuity across interfaces. The region that includes the Euclidean disk, short-pulse annulus, barrier annulus, and star shell is called the \emph{star region}. This step is obtained in Theorem \ref{constructioninnerannulus}. 

\paragraph{Step 5: Electrovacuum gluing.} By Step 4, the Klein-Gordon field has been localized to the star region, so that the subsequent gluing reduces to the pure Einstein-Maxwell case. In the electrovacuum regime, Maxwell fields are glued linearly, subject only to matching electric and magnetic charges (the latter is zero here). The Einstein part then glues via obstruction--free spacelike gluing, as in \cite{MOT}. Under the smallness and compatibility of ADM data across the common annulus, we obtain a solution that equals the inner data on the inside and Brill-Lindquist data on the outside. This step is obtained in Theorem \ref{MOT1.7}.

\paragraph{Step 6: Assembly of the Cauchy data and conclusion.} 
In Section \ref{secconstruction}, we combine all the previous constructions to obtain the final Cauchy initial data set that satisfies the EMKG constraints and leads to the formation of multiple causally independent trapped surfaces.
\begin{enumerate}
\item \emph{Outer region.} Around each center $\cb_I$, the charged Brill-Lindquist solution provides a large outer annulus on which the averaged ADM charges satisfy
\[
\E \simeq M_I,\qquad\quad\C \simeq M_I\cb_I,\qquad\quad\Q_E\simeq Q_I,
\]
with all other ADM components small and $\Q_B=0$. Moreover, $(g_{BL},E_{\mathrm{BL}})$ is $O(M_I)$--close to Euclidean space (Step 3).
\item \emph{Inner region.} Inside this, the near-star annulus from Step 4 provides a purely electrovacuum region where the Klein-Gordon fields vanish. The data there are $O(a_I^{-1})$--close to Euclidean, with vanishing magnetic charge $\Q_B=0$. This ensures the compatibility of the inner and outer regions for gluing.
\item \emph{Matching the electromagnetic charge.} The electric flux through the outer Brill-Lindquist annuli depends smoothly on the parameters ${Q_I}$. Using the implicit function theorem, we can adjust these parameters so that the electric fluxes agree across each gluing interface. The magnetic charges already match and vanish identically.
\item \emph{Applying the Einstein-Maxwell gluing theorem.} Since the electromagnetic charges coincide and all remaining ADM quantities are small and mutually compatible, following Step 5, we produce a smooth Einstein-Maxwell initial data set that matches the inner (transition) data on the inside and the Brill-Lindquist data on the outside. This gluing preserves asymptotic flatness, the controlled ADM charges of each component, and the regularity of the resulting initial data set.
\item \emph{Conclusion:} Finite propagation from the short-pulse annuli constructed in Step 1 implies that each interior region evolves to form a trapped surface in finite time. The disjointness of the star regions guaranties that these $N$ trapped surfaces are strongly causally independent, as in the discussion after Theorem \ref{maintheoremintro}. On the initial slice, a mean curvature comparison argument confirms the absence of trapped surfaces. 
\end{enumerate}
This completes the proof of Theorem \ref{maintheoremintro}.
\subsection{Structure of the paper}
The remainder of this paper is organized as follows. Section \ref{secpre} fixes notation, recalls the charged Brill-Lindquist metric \cite{BL}, computes local ADM charges, and reviews trapped surfaces/MOTS. Section \ref{sectrapped} establishes the formation of trapped surfaces for EMKG via the short-pulse method. Section \ref{secstability} constructs the transition region and the barrier annulus via stability analysis. Section \ref{secgluing} proves the electrovacuum gluing theorem adapted from \cite{MOT} and implements the star shell/near-star construction. Section \ref{secconstruction} assembles the Cauchy data and proves the main theorem. Appendix \ref{doublenullfoliation} summarizes the double null framework and the main spacetime equations.
\paragraph{Acknowledgments} E.G. acknowledges the support of NSF Grants DMS-2306143, DMS-2336118 and of a grant of the Sloan Foundation. J.W. is supported by ERC-2023 AdG 101141855 BlaHSt.
\section{Preliminaries}\label{secpre}
In this section, we collect some basic definitions, fix notation and parameters, and summarize the background necessary for our main results.
\subsection{Notation}\label{sec:notation}
Throughout the paper, $(\MM, \g)$ denotes a smooth Lorentzian manifold $\MM$ of dimension $3+1$ with a Lorentzian metric $\g$ of signature $(-,+,+,+)$.
\begin{itemize}
\item Capital Latin indices $A,B,C,\ldots$ run from $1$ to $2$; lowercase Latin indices $i,j,k,\ldots$ run from $1$ to $3$; Greek indices $\a,\b,\ga,\ldots$ run from $0$ to $3$.
\item $\Re$ and $\Im$ denote the real and imaginary parts; $\psi^\dag$ denotes the complex conjugate of $\psi$.
\item For any matrix $X$, $X^\top$ denotes its transpose.
\item $(\g,\bnab,\R,\D)$ denotes the spacetime metric, Levi-Civita connection, curvature tensor, and gauge covariant derivative on $\M$ (associated with the electromagnetic potential $\A$).
\item $(g,\nab,R,D)$ denotes the induced metric, Levi-Civita connection, curvature tensor, and gauge covariant derivative on the initial hypersurface $\Si_0$ (associated with the electromagnetic potential $A,\Phi$).
\item $(\slg,\nabs)$ denotes the induced metric and Levi-Civita connection on the spheres $S_{u,\ub}$ (cf.~\eqref{eq:def-S-u-ub}); $\nabs$ is also used for the horizontal covariant derivative on $\M$.
\item $e$ is the Euclidean metric on $\RRR^3$ and $\etabf$ is the Minkowski metric on $\RRR^{1,3}$.
\item $\BB_r(\xb)$ denotes the open ball of radius $r$ and centered at $\xb$ in $\RRR^3$, and $\pr\BB_r(\bf{x})$ is its boundary sphere.
\item $\AA_r :=\BB_{2r}\setminus\ov{\BB_r}$ is the annulus between radii $r$ and $2r$; $\At_r :=\BB_{4r}\setminus\ov{\BB_{r/2}}$ is the enlarged annulus that strictly contains $\AA_r$.
\end{itemize}
\subsubsection{Key parameters}
The following parameters will be frequently used throughout this paper:
\begin{itemize}
    \item $s\geq 3$ denotes the regularity of the Cauchy data, which is fixed at the beginning.
    \item $m^2$ denotes the size of the potential function $V$ in the Klein-Gordon equation.
    \item $N\in\NNN$ denotes the number of poles of the Brill-Lindquist metric \eqref{eq:gBL-EM}.
    \item $\{M_I\}_{I=1}^N\in\RRR_+^N$ denotes the mass parameters of \eqref{eq:gBL-EM}. The total mass $M$ is defined by $M:=\sum_{I=1}^N M_I$.
    \item $\{Q_I\}_{I=1}^N$ denotes the charge parameters of \eqref{eq:gBL-EM}.
    \item $\{\cb_I\}_{I=1}^N\in\RRR^{3N}$ denotes the positions of the poles of \eqref{eq:gBL-EM}.
    \item $d_I:=\min_{J\ne I}|\cb_I-\cb_J|$ denotes the minimal Euclidean distance from $\cb_I$ to the other poles.
    \item $\{a_I\}_{I=1}^N\in\RRR_+^N$ denotes the size of the short-pulse ansatz near $\cb_I$.
    \item $\{\de_I\}_{I=1}^N\in\RRR_+^N$ denotes the length of the short-pulse posed on the characteristic initial data near $\cb_I$.
\end{itemize}
Let $I\in\{1,2,\ldots,N\}$ and we consider the transition region near the pole $\cb_I$, the following parameters are used in Section \ref{secstability}:
\begin{itemize}
    \item $\ep_0:=a_I^{-1}$ measures the difference from the transition region to Minkowski spacetime.
    \item $\ep:=m^{-1}\ep_0$ denotes the bootstrap bound assumed in Section \ref{secstability}.
\end{itemize}
\subsubsection{Smallness constants}
For any quantities $A$ and $B$:
\begin{itemize}
    \item We write $A\les B$ to indicate that there exists a universal constant $C(s)>0$ that depends only on the regularity parameter $s$ such that $A\leq C(s) B$.
    \item We write $A=O(B)$ to indicate that $|A|\les |B|$.
    \item We write $A\simeq B$ to indicate that $|A|\les |B|\les |A|$.
    \item We write $A\ll B$ to indicate that $CA<B$, where $C$ is the largest universal constant among all the constants involved in the paper by $\les$.
\end{itemize}
Throughout this paper, the constants are chosen such that, for all $I\in\{1,2,\dots,N\}$:
\begin{align*}
0<\de_I\leq a_I^{-2},\qquad\quad a_I^{-1},d_I^{-1}\ll m^2\ll M_I<M\ll 1.
\end{align*}
\subsection{Einstein-Maxwell-Klein-Gordon constraint equations}
We introduce here the main equations we solve in this paper, i.e., the constraint equations for an initial data set for the Einstein-Maxwell-Klein-Gordon (EMKG) system \eqref{EMKG-phy}. We first define an initial data set for the equations and then derive the form of the constraint equations obtained by projecting the EMKG system \eqref{EMKG-phy} onto $\Si$ and applying the Gauss-Codazzi identities.
\begin{df}\label{def:initial-data}
An \emph{initial data set} for the Einstein-Maxwell-Klein-Gordon (EMKG) system consists of a $9$--tuple $(\Si,g,k,E,B,\psi,\phi,A,\Phi)$, where
\begin{itemize}
\item $(\Si,g)$ is a $3$--dimensional Riemannian manifold,
\item $k$ is a symmetric $2$--tensor on $\Si$ (identified as the second fundamental form),
\item $E,B$ are vectorfields on $\Si$ (identified as the electric and magnetic fields),
\item $\psi,\phi:\Si\to\mathbb{C}$ are complex scalar functions (identified as the Klein-Gordon field and its time-derivative),
\item $A$ is a $1$--form on $\Si$ (identified as the spatial gauge potential),
\item $\Phi:\Si\rightarrow \mathbb{R}$ is a scalar function (identified as the time gauge potential).
\end{itemize}
When $\psi=\phi=A=\Phi=0$, we simply write $(\Si,g,k,E,B)$ and refer to it as an \emph{initial data set} for the Einstein-Maxwell (EM) system.
\end{df}
\begin{prop}[EMKG constraint equations]\label{def:EMKG-constraints}
Let $(\Si,g,k,E,B,\psi,\phi,A,\Phi)$ be an initial data set for the EMKG system. Then, the associated EMKG constraint equations on $\Si$ are given by:
\begin{equation}\label{EMKGconstraint}
\begin{split}
R(g)+(\tr_g k)^2-|k|_g^2&=2\left(|E|_g^2+|B|_g^2+|D_0\psi|_g^2+|D\psi|_g^2\right),\\
\nab^j\!\left(k_{ij}-g_{ij}\tr_g k\right)&=2\Re\left((D_i\psi)^\dag D_0\psi\right)+2(E\times B)_i,\\
\sdiv_gE&=-2\ef\Im\left(\psi^\dag D_0\psi\right),\\
\sdiv_gB&=0,
\end{split}
\end{equation}
where $R(g)$ denotes the scalar curvature of $g$, $(E\times B)_i:={\in_{i}}^{jk}E_j B_k$ with $\in$ being the induced volume element of $(\Si,g)$ and
\[
D_0\psi:=\phi+\ik\ef\Phi\psi,\qquad\quad D_i\psi:=\pr_i\psi+\ik\ef A_i\psi.
\]
\end{prop}
\begin{proof}
 Taking the $00$--component of the Einstein equations (where $\mu=0$ corresponds to the direction of the unit normal to $\Si$) and using the Gauss identity yields
\[
R(g)+(\tr_g k)^2-|k|_g^2 = 2\T_{00},
\]
which gives the first (Hamiltonian) constraint once $\T=\T^{EM}+\T^{KG}$ is expanded. From \eqref{Tdfeqintro} one computes
\[
\T^{EM}_{00}=|E|_g^2+|B|_g^2,\qquad\quad\T^{KG}_{00}=|D_0\psi|_g^2+|D\psi|_g^2,
\]
hence, the right-hand side of the first equation in \eqref{EMKGconstraint}. Next, taking the $0i$--component of the Einstein equations and using the Codazzi identity gives
\[
\nab^j\!\left(k_{ij}-g_{ij}\tr_g k\right)=\T_{0i}.
\]
The matter terms decompose as
\[
\T^{KG}_{0i}=2\,\Re\big((D_i\psi)^\dag D_0\psi\big), \qquad
\T^{EM}_{0i}=2(E\times B)_i,
\]
which yields the second (momentum) constraint. Finally, projecting the Maxwell equation $\bnab^\mu \F_{\mu\nu}=-2\ef\Im\big(\psi(\D_\nu\psi)^\dag\big)$ onto $\nu=0$ gives
\[
\sdiv_g E=-2\ef\Im(\psi^\dag D_0\psi),
\]
while the $\nu=i$ projection (together with the Bianchi identity for $\F$) yields $\sdiv_g B=0$. This concludes the proof of Proposition \ref{def:EMKG-constraints}.
\end{proof}
When the data for the Klein-Gordon field vanish, i.e., $\psi=\phi=0$, one can deduce the form of the constraint equations for EM.
\begin{cor}[EM constraint equations]\label{EM-constraints}
Let $(\Si,g,k,E,B)$ be an initial data set for the EM system. Then, the associated EM constraint equations on $\Si$ are given by:
\begin{equation}\label{EMconstraint}
\begin{split}
R(g)+(\tr_g k)^2-|k|_g^2&=2(|E|_g^2+|B|_g^2), \\
\nab^j(k_{ij}-g_{ij}\tr_g k)&=2(E\times B)_i,\\
\sdiv_g E&=0,\\
\sdiv_g B&=0,
\end{split}
\end{equation}
\end{cor}
\begin{rk}
The system \eqref{EMconstraint} is under-determined: the constraints restrict $(g,k,E,B)$ but do not fix them uniquely. For example, $\sdiv_g E=0$ determines only the divergence of $E$, so writing
\[
E=-\nab f + W,\qquad\quad\nab^i W_i=0,
\]
with $f$ a $g$--harmonic function, one can see that the divergence-free part $W$ remains free. The same under-determination applies to $B$ and to $(g,k)$ through the Hamiltonian and momentum constraints.
\end{rk}
\subsection{Charges of Cauchy data}\label{sseccharges}
Here, we define an important set of charges associated with initial data sets. They are divided into two types: 10 of them correspond to a localized version of the well-known ADM charges for initial data for the Einstein equations, while the remaining ones are identified with the electric and magnetic charges of the initial data set.
\begin{df}\label{chargesEMKG}
Let $\Si \subseteq \mathbb{R}^3$ be equipped with a coordinate system $(x^1,x^2,x^3)$, and let $(g,k,E,B)$ be the geometric and electromagnetic fields of an initial data set as in Definition \ref{def:initial-data}. We define the following charges\footnote{These definitions depend only on $(g,k,E,B)$ and are independent of the Klein-Gordon and gauge fields $\psi,\phi,A,\Phi$.} measured on the coordinate spheres $\pr\BB_r$:
\begin{align*}
    \E[(g,k);\pr\BB_r]&:=\frac{1}{2}\int_{\pr\BB_r}(\pr_ig_{ij}-\pr_jg_{ii})\nu^jdS,\\
    \P_i[(g,k);\pr\BB_r]&:=\int_{\pr\BB_r}(k_{ij}-\de_{ij}\tr_{e} k)\nu^jdS,\\
    \C_l[(g,k);\pr\BB_r]&:=\frac{1}{2}\int_{\pr\BB_r}\left(x_l\pr_ig_{ij}-x_l\pr_jg_{ii}-\de_{il}(g-e)_{ij}+\de_{jl}(g-e)_{ii}\right)\nu^jdS,\\
    \J_l[(g,k);\pr\BB_r]&:=\int_{\pr\BB_r}(k_{ij}-\de_{ij}\tr_{e} k)Y_l^i\nu^jdS,
    \end{align*}
    where $ e $ is the Euclidean metric and we denote
\begin{align*}
    \pr_i:=\pr_{x^i},\qquad\quad \nu:=\frac{x^j}{|\mathbf{x}|}\pr_j,\qquad\quad Y_i:={\in_{ij}}^kx^j\pr_k.
\end{align*}
We put these together to form a $10$--vector:
\[
\Q_{ADM}[(g,k);\pr\BB_r]:=(\E,\P_1,\P_2,\P_3,\C_1,\C_2,\C_3,\J_1,\J_2,\J_3)[(g,k);\pr\BB_r].\]
We also define the following electromagnetic charges of $(g,k,E,B)$ measured on the sphere $\pr\BB_r$:
\begin{align*}
\Q_E[(g,E);\pr\BB_r]:=\int_{\pr\BB_r}\sqrt{\deg g} E_j \nu^jdS,\qquad \Q_B[(g,B);\pr\BB_r]:=\int_{\pr\BB_r}\sqrt{\deg g} B_j \nu^jdS.
\end{align*}
\end{df}
Since we will work with annuli, it's convenient to introduce the following averaged charges. 
\begin{df}\label{ADMannulus}
Let $\eta \in C^{\infty}_{c}(0, \infty)$ be a cutoff function such that 
\begin{equation*}
\supp\eta\subseteq (1, 2),\qquad \int_1^2\eta(r)dr=1,
\end{equation*}
and we denote for any $r>0$
\begin{align}\label{scaleeta}
\eta_{r}(r'):=r^{-1}\eta(r^{-1}r').
\end{align}
We then define
\begin{align*}
\Q_{ADM}[(g,k);\AA_r]&=\int_r^{2r}\eta_r(r')\Q_{ADM}[(g,k);\pr\BB_{r'}] dr',\\
\Q_{E,B}[(g,E,B);\AA_r]&=\int_r^{2r}\eta_r(r')\Q_{E,B}[(g,E,B);\pr\BB_{r'}]dr'.
\end{align*}
where $\AA_r$ denotes the annulus between the spheres $\pr\BB_{2r}$ and $\pr\BB_r$.
\end{df}
Recall that we denote $\At_r:=\BB_{4r}\setminus\ov{\BB_{\frac{r}{2}}}$. We introduce the notion of a 10--parameter family of admissible annular initial data for EM constraints with prescribed electromagnetic charges in a smaller annulus and Sobolev regularity.
\begin{df}\label{MOT1.2}
Let $\QQ\subseteq\RRR^{10}$ be a bounded open set and fix $Q_E,Q_B\in\RRR$ and $s>\frac{3}{2}$. A family $\{(g_Q,k_Q,E_Q,B_Q)\}_{Q\in\QQ}$ is called an admissible 10--parameter
family of annular initial data on $\At_r$ with fixed electromagnetic charges $(Q_E,Q_B)$ on $
\AA_r\subseteq\At_r$ and Sobolev regularity $s$ if for every $Q\in\QQ$:
\begin{enumerate}
\item $(g_Q,k_Q,E_Q,B_Q)\in (H^1\cap C^0)\times L^2\times L^2\times L^2(\At_r)$
      and solves the EM constraint equations \eqref{EMconstraint} in $\At_r$;
\item $(g_Q,k_Q,E_Q,B_Q)\in \XX^s(\At_r) :=
      H^s\times H^{s-1}\times H^{s-1}\times H^{s-1}(\At_r)$;
\item the map $Q\mapsto(g_Q,k_Q,E_Q,B_Q)$ is Lipschitz from $\QQ$ into $\XX^s(\At_r)$;
\item the localized ADM charges on $\AA_r$ satisfy
      \[
      \Q_{ADM}[(g_Q,k_Q);\AA_r]=Q;
      \]
\item the electromagnetic charges on $\AA_r$ are fixed:
      \[
      \Q_E[(g_Q,E_Q);\AA_r]=Q_E,\qquad \Q_B[(g_Q,B_Q);\AA_r]=Q_B.
      \]
\end{enumerate}
In other words, on $\AA_r$, the ADM charges $\Q_{ADM}=(\E,\P,\C,\J)$ vary with $Q\in\QQ$, while the electromagnetic charges $(Q_E,Q_B)$ remain fixed.
\end{df}
As in Lemma 4.2 of \cite{MOT}, for the Kerr family, one can obtain such an admissible set on a given annulus by taking the exterior regions of Kerr-Newman initial data with fixed electromagnetic charges $(Q_E,Q_B=0)$ and varying geometric parameters. In fact, an admissible family with arbitrary $Q_E$ and $Q_B$ can be obtained by exploiting the $U(1)$ symmetry of Maxwell equations. 
\subsection{Charged Brill-Lindquist initial data}\label{secBL-EM}
A convenient way to describe an initial data set representing several non-rotating, localized solutions (such as black holes) is provided by the Brill-Lindquist construction. The Brill-Lindquist manifold $(\Si,g,k)$ is a time-symmetric initial data set for the vacuum Einstein equations, obtained by specifying a conformally flat 3-metric and a vanishing second fundamental form,
\begin{equation*}
      g = U^4  e, \qquad k = 0,
\end{equation*}
where $U$ is a positive function on $\Sigma$ satisfying the vacuum Hamiltonian constraint from $R(g)=0$, giving
\begin{equation*}
   \Delta_{ e } U = 0,
\end{equation*}
with $\Delta_{ e }$ the flat Laplacian in $\mathbb{R}^3$. 

To represent $N$ localized black holes, one considers a conformal factor of the form
\begin{equation*}
    U(\mathbf{x}) = 1 + \sum_{I=1}^{N} \frac{M_I}{2|\mathbf{x} - \mathbf{c}_I|},
\end{equation*}
where each term corresponds to an asymptotically flat end centered at $\mathbf{c}_I \in \mathbb{R}^3$, associated with a black hole of mass parameter $M_I > 0$. The manifold $\Sigma$ is then obtained by removing the points $\{\mathbf{c}_I\}$ from $\mathbb{R}^3$ and endowing the remainder with the metric $g = U^4  e $. Each puncture $\mathbf{c}_I$ corresponds to an additional asymptotically flat region, connected through a minimal surface (or ``throat'') to the exterior region.

The Brill-Lindquist data thus represent a time-symmetric configuration of $N$ non-spinning black holes momentarily at rest. The geometry possesses $(N+1)$ asymptotically flat ends and satisfies the vacuum Einstein constraint equations exactly. In the limit where the punctures are widely separated, the ADM mass of the manifold is approximately the sum of the individual parameters $M_I$, and the initial data describe a system of well-separated Schwarzschild black holes. See \cite{BL} for more details.

In this section, we introduce the \emph{charged Brill-Lindquist manifold}, whose construction also appeared in \cite{BL}, which generalizes the vacuum construction to the Einstein-Maxwell system. 
\begin{prop}\label{prop:EM-reduction}\label{lem:harmonic-chi-psi}\label{prop:BL-EM}
On $\Si=\RRR^3\setminus\{\cb_1,\dots,\cb_N\}$, adopt the conformally flat, time-symmetric ansatz
\begin{equation*}
g=U^4  e ,\qquad U:=(U_+U_-)^{1/2},\qquad k=0,\qquad B=0,
\end{equation*}
with the electric field
\begin{equation*}
E:=-\nab_g(\ln U_+-\ln U_-)=-U^{-4}\nab_e(\ln U_+-\ln U_-),
\end{equation*}
where $U_+,U_->0$ and $U_+,U_-\to 1$ as $|\x|\to\infty$. Assume that $U_+$ and $U_-$ are Euclidean harmonic functions, i.e.
\begin{equation}\label{EM-constraint}
 \Delta_{ e } U_+=0,\qquad  \Delta_{ e } U_-=0.
\end{equation}
Then, such $(g, k, E, B)$ are solutions to the EM constraint equations \eqref{EMconstraint}. In particular, for $U_+, U_-$ defined on $\RRR^3\setminus\{\cb_1,\dots,\cb_N\}$ by
\begin{equation}\label{eq:chi-psi-explicit}
U_+(\x)=1+\sum_{I=1}^{N}\frac{M_I+Q_I}{2|\x-\cb_I|},\qquad\quad U_-(\x)=1+\sum_{I=1}^{N}\frac{M_I-Q_I}{2|\x-\cb_I|},
\end{equation}
where $M_I$ and $|Q_I|$ are fixed constants satisfying $M_I\geq |Q_I|$ for all $I\in\{1,2,\cdots,N\}$, 
the corresponding initial data, explicitly given by
\begin{align}\label{eq:gBL-EM}
\begin{split}
g_{BL}&= U^4 e=\left(1+\sum_{I=1}^{N}\frac{M_I+Q_I}{2|\x-\cb_I|}\right)^2\left(1+\sum_{I=1}^{N}\frac{M_I-Q_I}{2|\x-\cb_I|}\right)^2 e,\qquad k=0,\\
E_{BL}&= -\nabla_g\ln\left(\frac{1+\sum_{I=1}^N\frac{M_I+Q_I}{2|\x-\cb_I|}}{1+\sum_{I=1}^N\frac{M_I-Q_I}{2|\x-\cb_I|}}\right),\qquad\qquad\qquad\qquad\qquad\quad\;\;\; B=0,
\end{split}
\end{align}
is a solution to the EM constraint equations \eqref{EMconstraint}. Such $(\Si,g, k, E, B)$ is called a \emph{charged Brill-Lindquist manifold}.
\end{prop}
\begin{proof}
By construction, the second and fourth identities of \eqref{EMconstraint} follow directly from $k=B=0$. Next, for $g=U^4e$, we obtain
\[
R(g)=-8U^{-5}\De_{ e } U.
\]
From the form of $U$ as $U=(U_+U_-)^{1/2}$, we infer
\begin{align*}
\nab_e U&=\frac{U}{2}\left(\nab_e\ln U_++\nab_e\ln U_-\right), \\
\De_e U&=\frac{U}{4}|\nab_e\ln U_++\nab_e\ln U_-|^2+\frac{U}{2}(\De_e\ln U_++\De_e\ln U_-)\\
&=-\frac{U}{4}|\nab_e\ln U_+|^2-\frac{U}{4}|\nab_e\ln U_-|^2+\frac{U}{2}\nab_e\ln U_+\c\nab_e\ln U_-\\
&=-\frac{U}{4}|\nab_e(\ln U_+-\ln U_-)|^2.
\end{align*}
Hence, we deduce
\[
R(g)=-8U^{-5}\De_e U=2U^{-4}\,|\nab_e(\ln U_+-\ln U_-)|^2=2|E|_g^2,
\]
which implies the first identity in \eqref{EMconstraint}. Finally, we have from \eqref{EM-constraint},
\begin{align*}
\sdiv_gE&=U^{-6}\sdiv_e(U^6E)\\
&=-U^{-6}\sdiv_e\left(U^2\nab_e(\ln U_+-\ln U_-)\right),\\
&=-U^{-4}\De_e(\ln U_+-\ln U_-)-2U^{-5}\nab_e U\c\nab_e(\ln U_+-\ln U_-),\\
&=-U^{-4}\left(\De_e\ln U_++|\nab_e\ln U_+|^2\right)+U^{-4}\left(\De_e\ln U_-+|\nab_e\ln U_-|^2\right)\\
&=-U^{-4}\left(\frac{\De_eU_+}{U_+}\right)+U^{-4}\left(\frac{\De_eU_-}{U_-}\right)=0,
\end{align*}
which implies the third identity in \eqref{EMconstraint}, proving that the above ansatz is a solution to \eqref{EMconstraint}. Observe that $U_+$ and $U_-$ given by \eqref{eq:chi-psi-explicit} are solutions to \eqref{EM-constraint} as they are linear combinations of the fundamental solution of the Laplace equation $|\x-\cb_I|^{-1}$ on $\RRR^3\setminus\{\cb_I\}$.
\end{proof}
\begin{rk}
Letting $|\x|\to\infty$, the conformal factors have the asymptotics
\begin{align*}
U_+(\x)\approx 1+\sum_{I=1}^N\frac{M_I+Q_I}{2|\x|},\qquad\quad U_-(\x)\approx 1+\sum_{I=1}^N\frac{M_I-Q_I}{2|\x|}.
\end{align*}
Hence, we have
\begin{align}
\begin{split}\label{eq:asymp-g-E}
g_{BL}\approx\left(1+\sum_{I=1}^N\frac{2M_I}{|\x|}\right)e, \qquad E_{BL}\approx\sum_{I=1}^N\frac{Q_I}{|\x|^2}\frac{\x}{|\x|}.
\end{split}
\end{align}
Thus $M_I$ and $Q_I$ are the natural ADM mass and electric charge parameters of the $I$--th pole.
\end{rk}
\begin{rk}\label{rk:EM-special}
The Brill-Lindquist metric defined in Proposition \ref{prop:BL-EM} includes the following special cases: $Q_I=0$ for all $I$ is the vacuum Brill-Lindquist data; $Q_I=\pm M_I$ for all $I$ is the Majumdar-Papapetrou initial data; $N=1$ is the time-symmetric slice of Reissner-Nordstr\"om in isotropic coordinates.
\end{rk}
\subsubsection{Local masses, centers, and electric charges}\label{subsec:local-charges-EM}
We now compute the ADM, electric and magnetic charges defined in  Definition \ref{chargesEMKG} for the charged Brill-Lindquist manifold measured on spheres centered at the poles.
\begin{prop}\label{prop:BL-local-EM}\label{cor:BL-annulus-EM}
Let $g_{BL},E_{BL}$ be defined in \eqref{eq:gBL-EM}.  Fix $I\in\{1,\dots,N\}$, $R\in[32,64]$ and set
\begin{equation*}
\y_I:= \x-\cb_I,\qquad\pr\BB_R(\cb_I) := \{|\y_I|=R\},\qquad\nu_I:=\frac{\y_I}{R}.
\end{equation*}
Assume the following small-mass and large-separation conditions
\begin{equation}\label{dfMd-EM}
M:=\sum_{I=1}^{N}M_I\ll 1,
\qquad 
d_I := \min_{J\ne I}|\cb_I-\cb_J|\gg M_I^{-1}.
\end{equation}
Then, we have for $l=1,2,3$,
\begin{align}
\begin{split}\label{EM-EE-local}
\E[(g_{BL},k=0);\pr\BB_R(\cb_I)]&=8\pi M_I+O(M_IM+Md_I^{-1}),\\
\P_l[(g_{BL},k=0);\pr\BB_R(\cb_I)]&=0,\\
\C_l[(g_{BL},k=0);\pr\BB_R(\cb_I)]&=(\cb_I)_l\left(8\pi M_I+O(M_IM+Md_I^{-1})\right)+O(M_IM+Md_I^{-1}),\\
\J_l[(g_{BL},k=0);\pr\BB_R(\cb_I)] &=0,\\
\Q_E[(g_{BL},E_{BL});\pr\BB_R(\cb_I)]&= 4\pi Q_I+O(Q_IM+Md_I^{-1}),\\
\Q_B[(g_{BL},B=0);\pr\BB_R(\cb_I)]&=0.
\end{split}
\end{align}
As a consequence, the same identities hold on $\AA_{32}(\cb_I)=\BB_{64}(\cb_I)\setminus\ov{\BB_{32}(\cb_I)}$. 
\end{prop}
\begin{proof}
The following expansions hold on $\pr\BB_R(\cb_I)$ for $J\ne I$:
\begin{align}
\begin{split}\label{chi-psi-exp}
U_+ &= 1+\frac{M_I+Q_I}{2R}+O(Md_I^{-1}),\qquad U_-=1+\frac{M_I-Q_I}{2R}+O(Md_I^{-1}), \\
U^2&=U_+U_-=1+\frac{M_I}{R}+O(M_IM)+O(Md_I^{-1}),\\
\nu_I(U^4)&=-\frac{2M_I}{R^2}+O(M_IM)+O(M d_I^{-1}).
\end{split}
\end{align}
These yield the first four identities in \eqref{EM-EE-local} via the standard surface integral formulas for $g=U^4 e$. See Proposition 5.7 and Corollary 5.8 in \cite{ShenWan2}. Next, we have from Definition \ref{chargesEMKG},
\[
\Q_E[(g_{BL},E_{BL});\pr\BB_R(\cb_I)]= \int_{\pr\BB_R(\cb_I)}\sqrt{\det g_{BL}}\, (E_{BL})_j(\nu_I)^j dS=\int_{\pr\BB_R(\cb_I)} U^6 E_{BL}\cdot\nu_I dS.
\]
Recalling from \eqref{eq:gBL-EM} and \eqref{eq:asymp-g-E} that we have
\[
E_{BL}\cdot \nu_I=-g\left(\nab_g\ln\left(\frac{U_+}{U_-}\right),\nu_I\right)=-\nu_I\ln\!\left(\frac{U_+}{U_-}\right)=\frac{Q_I}{R^2}+O\left(\frac{M}{d_I^2}\right).
\]
Moreover, we have
$$
U^6=(U_+U_-)^3=1+\frac{3M_I}{R}+O\left(\frac{M}{d_I}\right).
$$
Hence, we obtain
\begin{align*}
\Q_E\big[(g_{BL},E_{BL});\pr\BB_R(\cb_I)\big]
&=\int_{\pr\BB_R(\cb_I)}\left(1+\frac{3M_I}{R}+O\left(\frac{M}{d_I}\right)\right)\left(\frac{Q_I}{R^2}+O\left(\frac{M}{d_I^2}\right)\right)dS\\
&=4\pi Q_I+O(Q_IM)+O(Md_I^{-1}),
\end{align*}
which yields the fifth identity in \eqref{EM-EE-local}. Moreover, the last identity in \eqref{EM-EE-local} follows directly from the fact that $B=0$. Finally, integrating the fluxes over $\pr\BB_R(\cb_I)$ against the annular weight $\eta_{32}(r')$ one obtains the same formulas on the annulus $\AA_{32}(\mathbf{c}_I)$ as in Definition \ref{ADMannulus}. This concludes the proof of Proposition \ref{prop:BL-local-EM}.
\end{proof}
We now bound the size of $g_{BL}-e$ and $E_{BL}$ in the annulus $\AA_{32}(\cb_I)$.
\begin{prop}\label{prop:Sobolev-EM}
Under the hypotheses of Proposition \ref{prop:BL-local-EM}, we have, for any $I\in\{1,\dots,N\}$ and $s\in\NN$,
\begin{align}
\begin{split}\label{eq:Sobolev-g-EM}
\|g_{BL}-e\|_{H^{s}(\AA_{32}(\cb_I))}\les M_I, \qquad \|E_{BL}\|_{H^{s}(\AA_{32}(\cb_I))}\les |Q_I|+d_I^{-1}.
\end{split}
\end{align}
\end{prop}
\begin{proof}
We have from \eqref{chi-psi-exp}
\[
(g_{BL}-e)_{ij}=(U^4-1)\de_{ij}=\left(\frac{2M_I}{r}+O\left(\frac{M}{d_I}\right)\right)\de_{ij}.
\]
Hence, for any multi-index $\a$ with $0\leq|\a|\leq s$,
\[
\pr^\a(g_{BL}-e)= O\left(\frac{M_I}{r^{1+|\a|}}\right)+O\left(\frac{M}{d_I}\right).
\]
Recalling that $r\in[32,64]$ on $\AA_{32}(\cb_I)$, we deduce 
\[
\pr^\a(g_{BL}-e)=O(M_I)+O(Md_I^{-1}),
\]
which implies \eqref{eq:Sobolev-g-EM}. Next, we have from \eqref{chi-psi-exp}
\[
E_{BL} = -\frac{Q_I}{r^2}\frac{\y_I}{r} + O\left(\frac{M}{d_I^2}\right).
\]
Noticing that $r\in[32,64]$, we infer $\pr^\a E_{BL}=O(|Q_I|)+O(Md_I^{-2})$.
Consequently,
\[
\|E_{BL}\|_{H^s(\AA_{32}(\cb_I))}\les |Q_I|+Md_I^{-2}\les |Q_I|+d_I^{-1},
\]
where we used the fact that $d_I\gg M_I^{-1}$. This concludes the proof of Proposition \ref{prop:Sobolev-EM}.
\end{proof}
\subsection{Trapped surfaces and MOTS}
We introduce the following definitions of trapped surfaces.
\begin{df}\label{dftrapped2+2}
Let $(\M,\g)$ be a spacetime endowed with a double null $(u,\ub)$--foliation defined in Appendix \ref{doublenullfoliation}. Then, with respect to the given double null $(u,\ub)$--foliation, a leaf $S_{u,\ub}\subseteq\M$ is called 
\begin{itemize}
    \item a trapped surface if the following hold on $S_{u,\ub}$:
    \begin{equation*}
        \trch<0,\qquad\quad \trchb<0;
    \end{equation*}
    \item a marginally outer trapped surface (MOTS) if the following hold on $S_{u,\ub}$:
    \begin{equation*}
        \trch=0,\qquad\quad \trchb<0.
    \end{equation*}
\end{itemize}
\end{df}
\begin{df}\label{dftrappedsurface}
Let $g$ be a Riemannian metric on $\Si$ and $k$ be a symmetric covariant $2$--tensor on $\Si$. Let $S\subset\Si$ be a compact, embedded smooth $2$--surface. Let $\th$ be the second fundamental form of $S$. We will say that $S$ is
\begin{itemize}
\item a trapped surface if the following hold on $S$:
\begin{equation}\label{dftrapped}
    \tr_\slg(-\th-k)<0,\qquad\quad\tr_\slg(\th-k)<0;
\end{equation}
\item a marginally outer trapped surface (MOTS) if the following hold on $S$:
\begin{equation*}
    \tr_\slg(-\th-k)<0,\qquad\quad\tr_\slg(\th-k)=0.
\end{equation*}
\end{itemize}
\end{df}
\begin{rk}
Definition \ref{dftrappedsurface} is consistent with Definition \ref{dftrapped2+2}. More precisely, we have the following identities:
\begin{align*}
    \chi=\th-k,\qquad\quad \chib=-\th-k,
\end{align*}
see (7.5.2b) in \cite{Ch-Kl}. 
\end{rk}
\section{Short-pulse cone and formation of trapped surfaces}\label{sectrapped}
We recall that the Einstein-Maxwell-Klein-Gordon system takes the following form:
\begin{align}
\begin{split}\label{EMKG-SP}
\Ric_{\mu\nu}-\frac{1}{2}\R\g_{\mu\nu}&={\T}^{EM}_{\mu\nu}+{\T}^{KG}_{\mu\nu},\\
{\bnab}^{\mu}\F_{\mu\nu}&=-2\ef\Im\left(\psi(\D_{\nu}\psi)^\dag\right),\\
{\g}^{\mu\nu}\D_{\mu}\D_{\nu}\psi&=V'(|\psi|^2)\psi,
\end{split}
\end{align}
where $V(x)$ is a real potential function satisfying
\begin{equation}\label{assumptionV}
    V(0)=0,\qquad\quad |V^{(k)}(0)|\leq m^2\ll 1, \qquad \forall\; k\geq 1.
\end{equation}
In Appendix \ref{doublenullfoliation}, we introduce the standard notations used in this section, based on the double null foliation of spacetime. We now state the following theorem.
\begin{thm}\label{thmtrapped}
Consider the Einstein-Maxwell-Klein-Gordon system \eqref{EMKG-SP} under the gauge choice $U=0$ defined in \eqref{dfUUbAsl}. Let $0<\de\leq a^{-2}$, $a\gg 1$, and $0<\ef\ll 1$ be fixed constants, and let $u_0\in [-1,-2]$. Consider the initial data that satisfy the following assumptions:
\begin{itemize}
\item along $H_{u_0}^{(0,\de)}$, we have
    $$
    \sum_{i+j\leq s+15}\afd\left\|(\de\nabs_4)^j\nabs^i\left(\hch,\bF,\psi\right)\right\|_{L^\infty(S_{u_0,\ub})}\leq 1;
    $$
\item along $\ub=0$, the data is Minkowskian;
\item along $u=u_0$, the following lower bound conditions hold
\begin{align*}
    \int_0^\de\left(|\hch|^2+|\bF|^2+|\Psi_4|^2\right)(u_0,\ub,x^{1},x^{2})d\ub &\geq\de a,\qquad\forall\; (x^{1},x^{2}), \\
    \left|\int_{H_{u_0}^{(0,\de)}}\Om\Im(\psi\Psi_4^\dag)\right|&\geq\de a.
\end{align*}
\end{itemize}
Then, the respective EMKG system \eqref{EMKG-SP} admits a unique smooth solution in the colored region of Figure \ref{scalingfigure}. Moreover, we have:
\begin{itemize}
        \item The $2$--sphere $S_{-\frac{\de a}{4},\de}$ is a trapped surface.
        \item The Hawking mass of the sphere $S_{u_0,\de}$ satisfies
        \begin{align*}
            m_H(S_{u_0,\de})\simeq\de a.
        \end{align*}
        \item The electric charge of the sphere $S_{u_0,\de}$ satisfies
        \begin{align*}
            |Q(S_{u_0,\de})|\simeq\de^2\ef a.
        \end{align*}
        \item The following estimates hold in the colored region of Figure \ref{scalingfigure}:
    \begin{align}
    \begin{split}\label{deaest}
    |\a|&\les\de^{-1}\af,\qquad\qquad\qquad\qquad |\b,\hch,\bF,\Psi_4,\Asl|\les\af,\\
    |\rho,\si,\rhoF,\siF,\Psisl,\Ub|&\les\de a,\qquad\qquad\quad\;\;|\bb,\bbF|\les\de^2 a^\frac{3}{2},\qquad\quad|\aa|\les\de^3a^2,\\
    |\om,\trch,\log\Om|&\les 1,\quad\;\;\;|\eta,\etab,\ze,\hchb,\bbb,\Psi_3,\psi|\les\de\af,\quad\;|\trchbt,\omb|\les\de^2a.
    \end{split}
\end{align}
Moreover, analog estimates also hold for their $(\de\nabs_4,\nabs_3,\nabs)^{s+4}$--derivatives.
\end{itemize}
\begin{figure}[H]
\centering
\begin{tikzpicture}[scale=0.75]
\fill[red!20,opacity=0.35] (1,1)--(2,2)--(4,0)--(3,-1)--cycle;
\draw[very thick] (1,1)--(2,2)--(4,0)--(3,-1)--cycle;
\node[scale=1, black] at (4,-0.8) {$H_{u_0}$};
\node[scale=1, black] at (1.1,1.9) {$H_{-\frac{\de a}{4}}$};
\node[scale=1, black] at (3.2,1.4) {$\Hb_\de$};
\node[scale=1, black] at (1.7,-0.3) {$\Hb_0$};
\draw[dashed] (0,-4)--(0,4);
\draw[dashed] (0,-4)--(4,0)--(0,4);
\draw[dashed] (0,0)--(2,2);
\draw[dashed] (0,2)--(3,-1);=
\draw[->] (3.3,-0.6)--(3.6,-0.3) node[midway, sloped, above, scale=1] {$e_4$};
\draw[->] (2.4,-0.3)--(1.7,0.4) node[midway, sloped, above, scale=1] {$e_3$};
\filldraw[red] (2,2) circle (2pt);
\filldraw[red] (4,0) circle (2pt);
\draw[thin] (2.1,2.5) node[above, scale=1] {$S_{-\frac{\de a}{4},\de}$} -- (2,2);
\draw[thin] (4.1,0.5) node[above, scale=1] {$S_{u_0,\de}$} -- (4,0);
\end{tikzpicture}
\caption{{\footnotesize The initial conditions in Theorem \ref{thmtrapped} lead to trapped surface ($S_{-\frac{\de a}{4},\de}$) formation in the future of the $H_{u_0}$ and $\Hb_0$.}}\label{scalingfigure}
\end{figure}
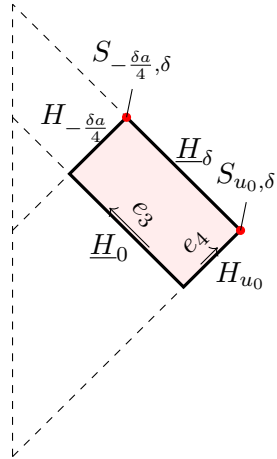
\end{thm}
\begin{rk}
According to the lack of $|u|$--decay for $V(|\psi|^2)$, it will be difficult to prove the formation of trapped surfaces from the  past null infinity $u=u_\infty$. More precisely, in Proposition \ref{Ricciexpression}, the terms involving $V(|\psi|^2)$ have worse $|u|$--decay than the other terms. This blocks us from integrating the $\nabs_3$--transport equations from $u_\infty$ to $u$. However, the formation of trapped surfaces in a finite region can be deduced similarly to \cite{ShenWan}.
\end{rk}
\begin{proof}[Proof of Theorem \ref{thmtrapped}]
The proof in the particular case $V=0$ follows from Theorem 10.5 in \cite{ShenWan}. Notice from Proposition \ref{Ricciexpression}, \eqref{assumptionV} and \eqref{deaest} that the Ricci tensor $\Ric_{\mu\nu}$ has the same regularity and estimate as the case $V=0$. It remains to treat the Klein-Gordon equation for $\psi$. To this end, we have from Proposition \ref{waveequation}
\begin{align*}
\nabs_3(\Psi_4)+\frac{1}{2}\trchb\Psi_4&=\sdivs\Psisl-V'(|\psi|^2)\psi+N_3[\Psi_4],\\
\nabs_4\Psisl+\frac{1}{2}\trch\Psisl&=\nabs(\Psi_4)+N_4[\Psisl],\\
\nabs_3\Psisl+\trchb \Psisl&=\nabs(\Psi_3)+N_3[\Psisl],\\
\nabs_4(\Psit_3)+\frac{1}{2}\trch\Psit_3&=\sdivs\Psisl-V'(|\psi|^2)\psi+N_4[\Psit_3],
\end{align*}
with $\Psit_3$ defined by \eqref{dfPsit3} and $N_3[\Psi_4]$, $N_4[\Psisl]$, $N_3[\Psisl]$ and $N_4[\Psit_3]$ denote the nonlinear terms in the schematic form $\A\c\Psi$ and $\Ga\c\Psi$, which are already estimated in \cite{ShenWan}. Moreover, we have from Proposition \ref{waveequation}
\begin{align*}
N_3[\Psi_4]&=2\omb\Psi_4-i\ef \Ub\Psi_4+2\eta\c\Psisl+i\ef \Asl\c\Psisl-\frac{1}{2}\trch\Psi_3+i\ef\rhoF\psi,\\
N_4[\Psit_3]&=2\om\Psit_3-i\ef U\Psit_3+2\etab\c\Psisl+i\ef\Asl\c\Psisl-\frac{1}{2}\trchbt\Psi_4-\frac{\trch}{2\Om|u|}\psi-i\ef\rhoF\psi.
\end{align*}
Hence, we have from \eqref{assumptionV} and \eqref{deaest} that
\begin{align*}
N_3[\Psi_4]=O(\de\af+\ef\de a^\frac{3}{2}),\qquad N_4[\Psit_3]=O(\de\af+\ef\de a^\frac{3}{2}),\qquad V'(|\psi|^2)\psi=O(m^2\de\af).
\end{align*}
Thus, we have for $m^2\ll 1$ that $V'(|\psi|^2)\psi$ has better decay and regularity than the nonlinear terms $N_3[\Psi_4]$ and $N_4[\Psit_3]$ . Hence, the $L^2$--flux estimate for the Bianchi pairs $((\Psi_4,0),\Psisl)$ and $(\Psisl,(\Psit_3,0))$ can be deduced as in Section 6.1 of \cite{ShenWan}. This concludes the proof of Theorem \ref{thmtrapped}.
\end{proof}
\section{Transition region and barrier annulus}\label{secstability}
In this section, we first construct in Theorem \ref{mainstability} a transition region based on the short-pulse region obtained in Section \ref{sectrapped}. Then, in Theorem \ref{interiorsolution}, we construct a spacelike initial data set by taking a constant-time slice in the region obtained in Theorem \ref{mainstability}.
\subsection{Transition region and fundamental norms}
Throughout this section, we always assume\footnote{Recall that $u_0\in[-2,-1]$ is a constant introduced in Theorem \ref{thmtrapped}.} that $u\in[u_0,-\frac{1}{2}]$ and $\ub\in[\de,1]$. We define
\begin{align*}
\V:=\V(u,\ub):=\big\{(u',\ub')\in\left[u_0,u\right]\times[\de,\ub]\big\},\qquad \V_*:=\V\left(-\frac{1}{2},1\right).
\end{align*}
The main goal of this section is to prove the following theorem.
\begin{thm}\label{mainstability}
Consider the Einstein-Maxwell-Klein-Gordon system \eqref{EMKG-SP} under the gauge choice $U=0$ and let $a_0>0$ be a sufficiently large constant. For $a>a_0$, $0<\de\leq a^{-2}$ and $u_0\in[-2,-1]$, we denote
\begin{equation}\label{dfep0}
    \ep_0:=a^{-1}\geq\de a.
\end{equation}
Assume that initial data are prescribed on $H_{u_0}\cup\Hb_{0}$ satisfying the following conditions:
\begin{itemize}
\item $H_{u_0}^{(0,1)}$ is endowed with an outgoing geodesic foliation.
\item The following assumption holds along $H_{u_0}^{(0,\de)}$:
\begin{equation}\label{hchassumption}
\sum_{i+j\leq s+15}\afd\left\|(\de\nabs_4)^j(|u_0|\nabs)^i\left(\hch,\bF,\Psi_4\right)\right\|_{L^\infty(S_{u_0,\ub})}\leq 1.
\end{equation}
\item The following conditions hold along $H_{u_0}^{(\de,1)}$:
\begin{align}\label{vanishingconditions}
    \hch=0,\qquad \a=0,\qquad \om=0,\qquad \bF=0,\qquad \Psi_4=0,\qquad \psi=0.
\end{align}
\item Minkowskian initial data along $\Hb_0$.
\end{itemize}
Then \eqref{EMKG-SP} admits a unique solution in $\left\{(u,\ub)\in\left[u_0,-\frac{1}{2}\right]\times[0,1]\right\}$ and the following estimates hold in $\V_*$:
\begin{align*}
\left\|\nabs^{\leq s+2}(\a,\b,\rho,\si,\bb,\aa)\right\|_{L^2(S_{u,\ub})}&\les\ep_0,\\
\left\|\nabs^{\leq s+3}(\trcht,\trchbt,\hch,\hchb,\om,\omb,\eta,\etab,\ze,\log\Om,\bbb,\slgc)\right\|_{L^2(S_{u,\ub})}&\les\ep_0,\\
\left\|\nabs^{\leq s+3}\left(\bF,\rhoF,\siF,\bbF,\Psi_4,\Psisl,\Psi_3,\psi,U,\Ub,\Asl\right)\right\|_{L^2(S_{u,\ub})}&\les\ep_0,
\end{align*}
where $\slg$ denotes the induced metric of $\g$ on $S_{u,\ub}$ and
\begin{align*}
\trcht:=\trch-\frac{2}{r},\qquad\quad\trchbt:=\trchb+\frac{2}{r},\qquad\quad r:=\ub-u.
\end{align*}
\end{thm}
The following schematic notations $\Ga$, $\Rl$, $\Rr$, $\Fl$ and $\Fr$ will be used frequently throughout Section \ref{secstability}.
\begin{df}\label{gamma}
We introduce the following schematic notations:
\begin{align*}
\Rl&:=\left\{\a,\,\b,\,\rho,\,\si,\,\bb\right\},\qquad\qquad\qquad\Rr:=\left\{\b,\,\rho,\,\si,\,\bb,\,\aa\right\},\\
\Fl&:=\left\{\bF,\,\rhoF,\,\siF,\,\Psi_4,\,\Psisl\right\},\qquad\Fr:=\left\{\rhoF,\,\siF,\,\bbF,\,\Psisl,\,\Psi_3,\,\Asl,\,\Ub\right\}.
\end{align*}
We also denote:
\begin{align*}
\Ga&:=\left\{\trcht,\,\trchbt,\,\hch,\,\hchb,\,\eta,\,\etab,\,\ze,\,\om,\,\omb,\,\log\Om\right\}\cup\Fl\cup\Fr,\\
\Ga^{(1)}&:=\nabs^{\leq 1}\Ga\cup\{\a,\,\b,\,\rho,\,\si,\,\bb,\,\aa\}.
\end{align*}
Finally, we define for $i\geq 1$
\begin{align*}
\Ga^{(i+1)}&:=\nabs^{\leq 1}\Ga^{(i)},\qquad\quad\,\Rl^{(i+1)}:=\nabs^{\leq 1}\Rl^{(i)},\qquad\quad\Rr^{(i+1)}:=\nabs^{\leq 1}\Rr^{(i)},\\
\Fl^{(i+1)}&:=\nabs^{\leq 1}\Fl^{(i)},\qquad\quad\Fr^{(i+1)}:=\nabs^{\leq 1}\Fr^{(i)}.
\end{align*}
\end{df}
We now introduce the following fundamental norms:
\begin{enumerate}
\item We define the following $L^2$--flux norms of the matter fields:
\begin{align*}
\Tk:=\sup_{\V_*}\sum_{i=0}^{s+4}\left(\FF_i(u,\ub)+\FFb_i(u,\ub)\right),
\end{align*}
where, for all $i\in\NN$, we denote
\begin{align*}
\FF_i(u,\ub):=\big\|\Fl^{(i)}\big\|_{L^2(\cuvs)},\qquad\quad \FFb_i(u,\ub):=\big\|\Fr^{(i)}\big\|_{L^2(\ucuvs)}.
\end{align*}
\item We define the following $L^2$--flux norms of the curvature components:
\begin{align*}
\Rk:=\sup_{\V_*}\sum_{i=0}^{s+3}\left(\Rk_i(u,\ub)+\Rkb_i(u,\ub)\right).
\end{align*}
where, for all $i\in\NN$, we denote
\begin{align*}
\Rk_i(u,\ub):=\big\|\Rl^{(i)}\big\|_{L^2(\cuvs)},\qquad\quad \Rkb_i(u,\ub):=\big\|\Rr^{(i)}\big\|_{L^2(\ucuvs)}.
\end{align*}
\item We define the following $L^2(S_{u,\ub})$ norms of the geometric quantities:
\begin{align*}
\Ok:=\sup_{\V_*}\sum_{i=0}^{s+3}\Ok_i(u,\ub),\quad \mbox{ where }\; \Ok_i(u,\ub):=\big\|\Ga^{(i)}\big\|_{L^2(S_{u,\ub})}\quad \forall\; i\in \NN.
\end{align*}
\item We define the following norms on $H_{u_0}^{(\de,1)}\cup\Hb_\de^{(u_0,-\frac{1}{2})}$:
\begin{align*}
\Tk_{(0)}&:=\sup_{\ub\in[\de,1]}\sum_{i=0}^{s+4}\FF_i(u_0,\ub)+\sup_{u\in[u_0,-\frac{1}{2}]}\sum_{i=0}^{s+4}\FFb_i(u,\de),\\
\Ok_{(0)}&:=\sup_{\ub\in[\de,1]}\sum_{i=0}^{s+3}\Ok_i(u_0,\ub)+\sup_{u\in[u_0,-\frac{1}{2}]}\sum_{i=0}^{s+3}\Ok_i(u,\de),\\
\Rk_{(0)}&:=\sup_{\ub\in[\de,1]}\sum_{i=0}^{s+3}\Rk_i(u_0,\ub)+\sup_{u\in[u_0,-\frac{1}{2}]}\sum_{i=0}^{s+3}\Rkb_i(u,\de).
\end{align*}
\end{enumerate}
\subsection{Main intermediate results}
Throughout Section \ref{secstability}, we define
\begin{align*}
    \ep:=\ep_0 m^{-\frac{1}{2}}.
\end{align*}
which measures the size of bootstrap bounds. Note that we have $\ep^3\ll m^2\ep^2\ll\ep_0^2\ll\ep^2$.

\begin{thm}\label{MM0}
Under the assumptions of Theorem \ref{mainstability}. Assume, in addition, that
\begin{equation*}
\Ok\leq\ep,\qquad\quad\Rk\leq\ep,\qquad\quad\Tk\leq\ep.
\end{equation*}
Then, we have
\begin{equation*}
\Ok_{(0)}\les\ep_0,\qquad\quad\Rk_{(0)}\les\ep_0,\qquad\quad\Tk_{(0)}\les\ep_0.
\end{equation*}
\end{thm}
Theorem \ref{MM0} is proved in Section \ref{secII}. The estimates for $\a$ and $\hch$ on $\Hb_\de$ follow from their transport equations in the direction $e_3$. The smallness of all other quantities on $\Hb_\de$ follow directly from Theorem \ref{thmtrapped}. On the other hand, all quantities on $H_{u_0}$ can be estimated by integrating their equations along $H_{u_0}$ and combining them with \eqref{vanishingconditions}.
\begin{thm}\label{M1}
Under the assumptions of Theorem \ref{mainstability} and assuming, in addition, that
\begin{align*}
\Ok_{(0)}\les\ep_0,\qquad\Rk_{(0)}\les\ep_0,\qquad\Tk_{(0)}\les\ep_0,\qquad \Ok\leq\ep,\qquad\Rk\leq\ep,\qquad \Tk\leq\ep.
\end{align*}
Then, we have
\begin{equation*}
\Tk\les\ep_0.
\end{equation*}
\end{thm}
Theorem \ref{M1} is proved in Section \ref{secF}. The proof is based on the energy estimates for the Maxwell equations and wave equations. The method used here can also be considered as the $r^p$--weighted estimates used in \cite{holzegel,Shen22,Shen24} in the particular case $p=0$.
\begin{thm}\label{MM1}
Under the assumptions of Theorem \ref{mainstability}. Assume, in addition, that
\begin{align*}
\Ok_{(0)}\les\ep_0,\qquad\Rk_{(0)}\les\ep_0,\qquad\Tk_{(0)}\les\ep_0,\qquad \Ok\leq\ep,\qquad\Rk\leq\ep,\qquad \Tk\les\ep_0.
\end{align*}
Then, we have
\begin{equation*}
\Rk\les\ep_0.
\end{equation*}
\end{thm}
Note that $\Tk\les\ep_0$ implies the control of all components of the energy-momentum tensor. Then, the proof of Theorem \ref{MM1} is largely analogous to Section 4.6 of \cite{ShenWan2}, which applies the energy estimates for the Bianchi equations.
\begin{thm}\label{MM2}
Under the assumptions of Theorem \ref{mainstability}. Assume, in addition, that
\begin{align*}
\Ok_{(0)}\les\ep_0,\qquad\Rk_{(0)}\les\ep_0,\qquad\Tk_{(0)}\les\ep_0,\qquad \Ok\leq\ep,\qquad\Rk\les\ep_0,\qquad \Tk\les\ep_0.
\end{align*}
Then, we have
\begin{equation*}
\Ok\les\ep_0.
\end{equation*}
\end{thm}
The proof of Theorem \ref{MM2} is largely analogous to Section 4.7 of \cite{ShenWan2}, which is conducted by integrating the null structure equations, the Bianchi equations, the Maxwell equations, and the wave equations along the outgoing and incoming null cones.
\subsection{Bootstrap assumptions and first consequences}\label{secbootboot}
In the rest of Section \ref{secstability}, we always make the following bootstrap assumptions:
\begin{align}\label{BB}
\Ok\leq\ep,\qquad\quad\Rk\leq\ep,\qquad\quad\Tk\leq\ep.
\end{align}
The following consequences of \eqref{BB} will be used frequently throughout this paper.
\begin{lem}\label{decayGa}
We have the following estimates:
\begin{align*}
\|\Ga^{(s+3)}\|_{L^2(S_{u,\ub})}\les\ep.
\end{align*}
Moreover, $\Ga$ can be written as $\Ga=\Ga_l+\Ga_r$, where $\Ga_l$ and $\Ga_r$ satisfy
\begin{equation*}
    \|\Ga_l^{(s+4)}\|_{\cuvss}\les\ep,\qquad\quad\|\Ga_r^{(s+4)}\|_{\ucuvss}\les\ep.
\end{equation*}
\end{lem}
\begin{proof}
It follows directly from \eqref{BB} and Definition \ref{gamma}.
\end{proof}
\begin{lem}\label{evolutionlemma}
Let $X$ and $F$ satisfy the outgoing transport equation
\begin{equation*}
\nabs_4X+\la_0\trch\,X=F,
\end{equation*}
where $\la_0\geq 0$. We have the following transport estimate:
\begin{equation*}
\|X\|_{L^2(S_{u,\ub})}\les\|X\|_{L^2(S_{u,\de})}+\int_\de^\ub\|F\|_{L^2(S_{u,\ub'})}d\ub'.
\end{equation*}
\end{lem}
\begin{proof}
See Lemma 4.1.5 in \cite{kn} and Lemma 6.1 in \cite{Shen24}.
\end{proof}
\begin{prop}\label{commu}
We have the following schematic commutator formulas:
\begin{align*}
[\Om\nabs_4,r\nabs]&=\Ga\c r\nabs+\Ga^{(1)},\\
[\Om\nabs_3,r\nabs]&=\Ga\c r\nabs+\Ga^{(1)}.
\end{align*}
\end{prop}
\begin{proof}
It follows directly from Lemma \ref{comm} and Definition \ref{gamma}.
\end{proof}
\subsection{Estimates for characteristic initial data}\label{secII}
In this section, we prove Theorem \ref{MM0}.
\begin{prop}\label{Hbde}
We have the following estimates on $\Hb_\de^{(u_0,-\frac{1}{2})}$:
    \begin{align*}
    \left\|\nabs^{\leq s+3}\left(\b,\rho,\si,\bb,\aa,\trcht,\trchbt,\hch,\hchb,\om,\omb,\eta,\etab,\ze,\log\Om,\bbb,\slgc\right)\right\|_{L^2(S_{u,\ub})}&\les\ep_0,\\
    \left\|\nabs^{\leq s+3}\left(\bF,\rhoF,\siF,\bbF,\Psisl,\Psi_3,\Psi_4,\psi,\Ub,\Asl\right)\right\|_{L^2(S_{u,\ub})}&\les\ep_0,\\
    \left\|\nabs^{\leq s+2}\a\right\|_{L^2(S_{u,\ub})}&\les\ep_0.
    \end{align*}
\end{prop}
\begin{proof}
As a consequence of Theorem \ref{thmtrapped} and \eqref{dfep0}, we have
    \begin{align}
    \begin{split}\label{consequenceoftrapped}
        \left\|\nabs^{\leq s+4}\left(\rho,\si,\bb,\aa,\trcht,\trchbt,\hchb,\omb,\eta,\etab,\ze,\log\Om,\bbb,\slgc\right)\right\|_{L^2(S_{u,\ub})}&\les\ep_0,\\
        \left\|\nabs^{\leq s+4}\left(\rhoF,\siF,\bbF,\Psisl,\Psi_3,\Ub\right)\right\|_{L^2(S_{u,\ub})}&\les\ep_0.
    \end{split}
    \end{align}
    It remains to estimate $\hch$, $\om$, $\b$, $\bF$, $\Psi_4$, $\Asl$ and $\a$ on $\Hb_\de^{(u_0,-\frac{1}{2})}$. We have from their $\nabs_3$--equations and \eqref{consequenceoftrapped}
    \begin{align*}
        \left\|\nabs_3\nabs^{\leq s+3}(\hch,\om,\b,\bF,\Asl)\right\|_{L^2(S_{u,\ub})}\les\ep_0.
    \end{align*}
    Recall that we have\footnote{The estimate for $\b$ follows from \eqref{codazzi} and the initial assumptions \eqref{hchassumption} and \eqref{vanishingconditions}.}
    \begin{align*}
        \hch=0,\quad \bF=0,\quad \om=0,\quad \Psi_4=0,\quad \psi=0,\quad |\nabs^{\leq s+3}\b|\les\ep_0\quad\mbox{ on }\; S_{u_0,\de}.
    \end{align*}
    Thus, we obtain by integrating them along $\Hb_\de$:
    \begin{align}\label{esthchomb}
    \left\|\nabs^{\leq s+3}(\hch,\om,\b,\bF,\om,\Psi_4,\psi)\right\|_{L^2(S_{u,\ub})}\les\ep_0\qquad \mbox{ on }\; \Hb_{\de}^{(u_0,-\frac{1}{2})}.
    \end{align}
    Finally, we have from Proposition \ref{bianchiequations} and \eqref{esthchomb}
    \begin{align*}
        \left\|\nabs_3\nabs^{\leq s+2}\a\right\|_{L^2(S_{u,\ub})}\les\ep_0.
    \end{align*}
    Integrating it along $\Hb_\de$ and applying \eqref{vanishingconditions}, we infer
    \begin{align*}
        \left\|\nabs^{\leq s+2}\a\right\|_{L^2(S_{u,\ub})}\les\ep_0.
    \end{align*}
    This concludes the proof of Proposition \ref{Hbde}.
\end{proof}
\begin{prop}\label{Hu0}
    We have the following estimates on $H_{u_0}^{(\de,1)}$:
    \begin{align}
    \begin{split}\label{Hu0eq}
        \left|\dk^{\leq s+3}\left(\b,\rho,\si,\bb,\trcht,\trchbt,\hchb,\omb,\ze\right)\right|&\les\ep_0,\\
         \left|\dk^{\leq s+3}\left(\rhoF,\siF,\bbF,\Psisl,\Psi_3,\Ub,\Asl\right)\right|&\les\ep_0,
    \end{split}
    \end{align}
    where we denote $\dk:=\{\nabs_4,\nabs_3,\nabs\}$.
\end{prop}
\begin{proof}
Recall that we have on $H_{u_0}^{(\de,1)}$
\begin{align*}
\hch=0,\qquad\bF=0,\qquad\a=0,\qquad\Psi_4=0,\qquad\psi=0,\qquad\om=0,\qquad\etab=-\ze.
\end{align*}
Proceeding as in Chapter 2 of \cite{Chr}, we obtain \eqref{Hu0eq} as stated.
\end{proof}
Combining Propositions \ref{Hbde} and \ref{Hu0}, this concludes the proof of Theorem \ref{MM0}.
\subsection{General Bianchi pairs}\label{secBianchi}
\begin{df}\label{tensorfields}
For tensor fields defined on a $2$--sphere $S$, we denote by $\sk_0:=\sk_0(S)$ the set of pairs of scalar functions, $\sk_1:=\sk_1(S)$ the set of $1$--forms and $\sk_2:=\sk_2(S)$ the set of symmetric traceless $2$--tensors.
\end{df}
The following lemma provides the general structure of Bianchi pairs.
\begin{lem}\label{keypoint}
Let $a_{(1)}$, $a_{(2)}$ be real numbers. Then, we have the following properties:
\begin{enumerate}
\item If $\psi_{(1)},h_{(1)}\in\sk_1$ and $\psi_{(2)},h_{(2)}\in\sk_0$ satisfying
\begin{align}
\begin{split}\label{bianchi1}
\nabs_3(\psi_{(1)})+a_{(1)}\trchb\,\psi_{(1)}&=-\sld_1^*(\psi_{(2)})+h_{(1)},\\
\nabs_4(\psi_{(2)})+a_{(2)}\trch\,\psi_{(2)}&=\sld_1(\psi_{(1)})+h_{(2)}.
\end{split}
\end{align}
Then, the pair $(\psi_{(1)},\psi_{(2)})$ satisfies
\begin{align}
\begin{split}\label{div}
&\bdiv(|\psi_{(1)}|^2e_3)+\bdiv(|\psi_{(2)}|^2e_4)\\
+&(2a_{(1)}-1)\trchb|\psi_{(1)}|^2+(2a_{(2)}-1)\trch|\psi_{(2)}|^2\\
=&2\sdivs(\psi_{(1)}\cdot\psi_{(2)})+2\psi_{(1)}\cdot h_{(1)}+2\psi_{(2)}\cdot h_{(2)}+\Ga\left(|\psi_{(1)}|^2+|\psi_{(2)}|^2\right).
\end{split}
\end{align}
\item If $\psi_{(1)},h_{(1)}\in\sk_0$ and $\psi_{(2)},h_{(2)}\in\sk_1$ satisfying
\begin{align}
\begin{split}\label{bianchi2}
\nabs_3(\psi_{(1)})+a_{(1)}\trchb\,\psi_{(1)}&=\sld_1(\psi_{(2)})+h_{(1)},\\
\nabs_4(\psi_{(2)})+a_{(2)}\trch\,\psi_{(2)}&=-\sld_1^*(\psi_{(1)})+h_{(2)}.
\end{split}
\end{align}
Then, the pair $(\psi_{(1)},\psi_{(2)})$ satisfies
\begin{align*}
\begin{split}
&\bdiv(|\psi_{(1)}|^2e_3)+\bdiv(|\psi_{(2)}|^2e_4)\\
+&(2a_{(1)}-1)\trchb|\psi_{(1)}|^2+(2a_{(2)}-1)\trch|\psi_{(2)}|^2\\
=&2\sdivs(\psi_{(1)}\cdot\psi_{(2)})+2\psi_{(1)}\cdot h_{(1)}+2\psi_{(2)}\cdot h_{(2)}+\Ga\left(|\psi_{(1)}|^2+|\psi_{(2)}|^2\right).
\end{split}
\end{align*}
\end{enumerate}
\end{lem}
\begin{proof}
It follows by taking $p=0$ and $k=1$ in Lemma 4.2 of \cite{Shen22}.
\end{proof}
\begin{rk}\label{Bianchitype}
Note that the Bianchi equations can be written as systems of equations of type \eqref{bianchi1} and \eqref{bianchi2}. In particular
    \begin{itemize}
        \item the Bianchi pair $(\bF,(\rhoF,\siF))$ satisfies \eqref{bianchi1} with $k=1$,
        \item the Bianchi pair $((\rhoF,-\siF),\bbF))$ satisfies \eqref{bianchi2} with $k=1$,
        \item the Bianchi pair $((\Psi_4,0),\Psisl)$ satisfies \eqref{bianchi2} with $k=1$,
        \item the Bianchi pair $(\Psisl,(\Psit_3,0))$ satisfies \eqref{bianchi1} with $k=1$.
    \end{itemize}
\end{rk}
\begin{prop}\label{keyintegral}
Let $(\psi_{(1)},\psi_{(2)})$ be a Bianchi pair that satisfies \eqref{bianchi1} or \eqref{bianchi2}. Then, we have the following properties:
\begin{itemize}
\item In the case $1-2a_{(1)}\geq 0$ and $2a_{(2)}-1\geq 0$, we have
\begin{align*}
&\quad\;\|\psi_{(1)}\|_{L^2(\cuvs)}^2+\|\psi_{(2)}\|^2_{L^2(\ucuvs)}\\
&\les\|\psi_{(1)}\|_{L^2(H_{u_0}^{(\de,\ub)})}^2+\|\psi_{(2)}\|^2_{L^2(\Hb_{\de}^{(u_0,u)})}\\
&+\int_{\V}|\psi_{(1)}\c h_{(1)}|+|\psi_{(2)}\c h_{(2)}|+|\Ga|\left(|\psi_{(1)}|^2+|\psi_{(2)}|^2\right).
\end{align*}
\item In the case $1-2a_{(1)}\leq 0$ and $2a_{(2)}-1\geq 0$, we have
\begin{align*}
&\quad\;\|\psi_{(1)}\|_{L^2(\cuvs)}^2+\|\psi_{(2)}\|^2_{L^2(\ucuvs)}\\
&\les\|\psi_{(1)}\|_{L^2(H_{u_0}^{(\de,\ub)})}^2+\|\psi_{(2)}\|^2_{L^2(\Hb_{\de}^{(u_0,u)})}\\
&+\int_{\V}|\psi_{(1)}|^2+|\psi_{(1)}\c h_{(1)}|+|\psi_{(2)}\c h_{(2)}|+|\Ga|\left(|\psi_{(1)}|^2+|\psi_{(2)}|^2\right).
\end{align*}
\end{itemize}
\end{prop}
\begin{proof}
It follows by integrating \eqref{bianchi1} or \eqref{bianchi2} in $V$ and applying Stokes' theorem, see Proposition 5.7 of \cite{Shen24}.
\end{proof}
\begin{df}\label{dfdkbb}
We define the weighted angular derivatives $\dkb$ as follows:
\begin{align*}
\dkb U:=r\sld_2 U,\quad\forall\;U\in\sk_2, \qquad \dkb\xi:=r\sld_1\xi,\quad\,\,\, \forall\; \xi\in\sk_1,\qquad \dkb f:=r\sld_1^* f,\quad \,\,\forall\; f\in\sk_0.
\end{align*}
We denote for any tensor $h\in\sk_k$, $k=0,1,2$,
\begin{equation*}
    h^{(0)}:=h,\qquad\quad h^{(i)}:=(h,\dkb h,...,\dkb^i h).
\end{equation*}
\end{df}
The following theorem provides a unified treatment of all the nonlinear error terms.
\begin{thm}\label{Junkman}
We have the following estimates:
    \begin{align*}
    \int_{\V} |\Ga^{(s+3)}||\Ga^{(s+3)}|\les\ep^2,\qquad\quad\int_{\V}|\Ga^{(s+3)}||\Ga||\Ga^{(s+3)}|\les\ep^3.
    \end{align*}
\end{thm}
\begin{proof}
We have
\begin{align*}
\int_{\V}|\Ga^{(s+3)}||\Ga^{(s+3)}|&\les\int_{\V}|\Ga_l^{(s+3)}|^2+|\Ga_r^{(s+3)}|^2\\
&\les\int_{u_0}^u\int_{H_{u'}^{(\de,\ub)}}|\Ga_l^{(s+3)}|^2du'+\int_{\de}^\ub\int_{\Hb_{\ub'}^{(u_0,u)}}|\Ga_r^{(s+3)}|^2d\ub'\\
&\les\int_{u_0}^u \ep^2 du'+\int_\de^\ub\ep^2d\ub'\les\ep^2
\end{align*}
as stated. The second estimate follows from the fact that $|\Ga|\les\ep$. This concludes the proof of Theorem \ref{Junkman}.
\end{proof}
\subsection{Estimates for matter fields}\label{secF}
\begin{prop}\label{estPsi4Psisl}
We have the following estimate:
\begin{align*}
\|\Psi_4^{(s+3)}\|_{L^2(\cuvs)}+\|\Psisl^{(s+3)}\|_{L^2(\ucuvs)}\les\ep_0.
\end{align*}
\end{prop}
\begin{proof}
We have from Proposition \ref{waveequation}
\begin{align*}
\nabs_3(\Psi_4)+\frac{1}{2}\trchb\Psi_4-2\omb\Psi_4+i\ef \Ub\Psi_4&=\sdivs\Psisl+2\eta\c\Psisl+i\ef\Asl\c\Psisl-\frac{1}{2}\trch\Psi_3\\
&+i\ef\rhoF\psi-V'(|\psi|^2)\psi,\\
\nabs_4\Psisl+\frac{1}{2}\trch \Psisl+\hch\c\Psisl&=\nabs(\Psi_4)+\frac{\eta+\etab}{2}\Psi_4+i\ef\Asl\Psi_4-i\ef\bF\psi.
\end{align*}
Denoting
\begin{align}\label{Psit4}
    \Psit_4:=r^{-1}e_4(r\psi),
\end{align}
we obtain
\begin{align*}
\nabs_3(\widetilde{\Psi}_4)+\frac{1}{2}\trchb\widetilde{\Psi}_4-2\omb\widetilde{\Psi}_4+i\ef \Ub\widetilde{\Psi}_4&=\sdivs\Psisl+2\eta\c\Psisl+i\ef \Asl\c\Psisl+i\ef\rhoF\psi-V'(|\psi|^2)\psi\\
&-\frac{1}{2}\left(\trch-\frac{2}{\Om r}\right)\Psi_3-\frac{1}{\Om r}\left(\trchb+\frac{2}{\Om r}\right)\psi.
\end{align*}
Hence, we obtain
\begin{align*}
\nabs_3\Psit_4+\frac{1}{2}\trchb\Psit_4&=\sdivs\Psisl-V'(|\psi|^2)\psi+\Ga\c\Ga,\\
\nabs_4\Psisl+\frac{1}{2}\trch\Psisl&=\nabs(\Psi_4)+\Ga\c\Ga.
\end{align*}
Applying Proposition \ref{commu}, we have for $i\leq s+3$
\begin{align*}
\nabs_3(\dkb^i\Psit_4)+\frac{1}{2}\trchb(\dkb^i\Psit_4)&=\sdivs(\dkb^i\Psisl)+m^2\psi^{(s+3)}+\Ga\c\Ga^{(s+3)},\\
\nabs_4(\dkb^i\Psisl)+\frac{1}{2}\trch(\dkb^i\Psisl)&=\nabs(\dkb^i\Psi_4)+\Ga\c\Ga^{(s+3)}.
\end{align*}
Applying Proposition \ref{keyintegral} and Theorem \ref{Junkman}, we infer
\begin{align*}
&\quad\;\|\Psit_4^{(s+3)}\|_{L^2(\cuvs)}^2+\|\Psisl^{(s+3)}\|^2_{L^2(\ucuvs)}\\
&\les\|\Psit_4^{(s+3)}\|_{L^2(H_{u_0}^{(\de,\ub)})}^2+\|\Psisl^{(s+3)}\|^2_{L^2(\Hb_{\de}^{(u_0,u)})}+\int_{\V}m^2|\Ga^{(s+3)}||\Ga^{(s+3)}|+|\Ga^{(s+3)}||\Ga||\Ga^{(s+3)}|\\
&\les\ep_0^2+m^2\ep^2+\ep^3\les\ep_0^2.
\end{align*}
Combining with \eqref{Psit4}, we obtain
\begin{align*}
    |\psi|\les\ep_0,\qquad\quad \|\psi^{(s+3)}\|_{L^2(\cuvs)}\les\ep_0,\qquad\quad \|\Psi_4^{(s+3)}\|_{L^2(\cuvs)}\les\ep_0.
\end{align*}
This concludes the proof of Proposition \ref{estPsi4Psisl}.
\end{proof}
\begin{prop}\label{estPsislPsi3}
We have the following estimate:
\begin{align*}
\|\Psisl^{(s+3)}\|_{L^2(\cuvs)}+\|\Psi_3^{(s+3)}\|_{L^2(\ucuvs)}\les\ep_0.
\end{align*}
\end{prop}
\begin{proof}
We have from Proposition \ref{waveequation}
\begin{align*}
\nabs_3\Psisl+\frac{1}{2}\trchb\Psisl&=\nabs(\Psi_3)+\Ga\c\Ga,\\
\nabs_4(\Psi_3)+\frac{1}{2}\trch\Psi_3&=\sdivs\Psisl+\Ga\c\Ga-V'(|\psi|^2)\psi.
\end{align*}
Applying Proposition \ref{commu}, we have for $i\leq s+3$
\begin{align*}
\nabs_3(\dkb^i\Psisl)+\frac{1}{2}\trchb(\dkb^i\Psisl)&=\sdivs(\dkb^i\Psi_3)+\Ga\c\Ga^{(s+3)},\\
\nabs_4(\dkb^i\Psi_3)+\frac{1}{2}\trch(\dkb^i\Psi_3)&=\nabs(\dkb^i\Psisl)+m^2\psi^{(s+3)}+\Ga\c\Ga^{(s+3)}.
\end{align*}
Applying Proposition \ref{keyintegral} and Theorem \ref{Junkman}, we infer
\begin{align*}
&\quad\;\|\Psisl^{(s+3)}\|_{L^2(\cuvs)}^2+\|\Psi_3^{(s+3)}\|^2_{L^2(\ucuvs)}\\
&\les\|\Psisl^{(s+3)}\|_{L^2(H_{u_0}^{(\de,\ub)})}^2+\|\Psi_3^{(s+3)}\|^2_{L^2(\Hb_{\de}^{(u_0,u)})}+\int_{\V}m^2|\Ga^{(s+3)}||\Ga^{(s+3)}|+|\Ga^{(s+3)}||\Ga||\Ga^{(s+3)}|\\
&\les\ep_0^2+m^2\ep^2+\ep^3\les\ep_0^2.
\end{align*}
This concludes the proof of Proposition \ref{estPsislPsi3}.
\end{proof}
\begin{prop}\label{estbFrhoF}
We have the following estimates:
\begin{align}
\|\bF^{(s+3)}\|_{L^2(\cuvs)}+\|(\rhoF,\siF)^{(s+3)}\|_{L^2(\ucuvs)}&\les\ep_0,\label{bFrhoFsiF}\\
\|(\rhoF,\siF)^{(s+3)}\|_{L^2(\cuvs)}+\|\bbF^{(s+3)}\|_{L^2(\ucuvs)}&\les\ep_0. \nonumber
\end{align}
\end{prop}
\begin{proof}
We have from Proposition \ref{Maxwellequations}
\begin{align*}
\nabs_3\bF+\frac{1}{2}\trchb\bF&=-\sld_1^*(\rhoF,\siF)+\Ga\c\Ga,\\
\nabs_4(\rhoF,\siF)+\trch(\rhoF,\siF)&=\sld_1\bF+\Ga\c\Ga.
\end{align*}
Applying Proposition \ref{commu}, we have for $i\leq s+3$
\begin{align*}
\nabs_3(\dkb^i\bF)+\frac{1}{2}\trchb(\dkb^i\bF)&=\sdivs(\dkb^i(\rhoF,\siF))+\Ga\c\Ga^{(s+3)},\\
\nabs_4(\dkb^i(\rhoF,\siF))+\trch(\dkb^i(\rhoF,\siF))&=\nabs(\dkb^i\bF)+\Ga\c\Ga^{(s+3)}.
\end{align*}
Applying Proposition \ref{keyintegral} and Theorem \ref{Junkman}, we deduce
\begin{align*}
&\quad\;\|\bF^{(s+3)}\|_{L^2(\cuvs)}^2+\|(\rhoF,\siF)^{(s+3)}\|^2_{L^2(\ucuvs)}\\
&\les\|\bF^{(s+3)}\|_{L^2(\cuvss)}^2+\|(\rhoF,\siF)^{(s+3)}\|^2_{L^2(\ucuvss)}+\int_{\V}|\Ga^{(s+3)}||\Ga||\Ga^{(s+3)}|\\
&\les\ep_0^2+\ep^3\les\ep_0^2,
\end{align*}
which implies \eqref{bFrhoFsiF}. Next, we have from Propositions \ref{Maxwellequations} and \ref{commu} that for $i\leq s+3$
\begin{align*}
\nabs_3(\dkb^i(\rhoF,-\siF))+\trchb(\dkb^i(\rhoF,-\siF))&=-\sld_1(\dkb^i\bbF)+\Ga\c\Ga^{(s+3)},\\
\nabs_4(\dkb^i(\rhoF,-\siF))+\frac{1}{2}\trch(\dkb^i(\rhoF,-\siF))&=\sld_1^*(\dkb^i(\rhoF,-\siF))+\Ga\c\Ga^{(s+3)}.
\end{align*}
Applying Proposition \ref{keyintegral} and Theorem \ref{Junkman}, we infer
\begin{align*}
&\quad\;\|(\rhoF,\siF)^{(s+3)}\|_{L^2(\cuvs)}^2+\|\bbF^{(s+3)}\|^2_{L^2(\ucuvs)}\\
&\les\|(\rhoF,\siF)^{(s+3)}\|_{L^2(\cuvss)}^2+\|\bbF^{(s+3)}\|^2_{L^2(\ucuvss)}+\int_{\V}|(\rhoF,\siF)|^2+|\Ga^{(s+3)}||\Ga||\Ga^{(s+3)}|\\
&\les\ep_0^2+\int_{\de}^\ub\int_{\Hb_{\ub'}^{(u_0,u)}}|(\rhoF,\siF)|^2d\ub'+\ep^3\les\ep_0^2,
\end{align*}
where we used \eqref{bFrhoFsiF} at the last step. This concludes the proof of Proposition \ref{estbFrhoF}.
\end{proof}
\begin{prop}\label{estA}
    We have the following estimate:
    \begin{equation*}
        \|(\Ub,\Asl)^{(s+3)}\|_{L^2(\ucuvs)}\les\ep_0.
    \end{equation*}
\end{prop}
\begin{proof}
    We have from Proposition \ref{propdA=F}
    \begin{align*}
    \nabs_4\Ub&=-2\rhoF+\Ga\c\Ga,\\
    \nabs_4\Asl+\frac{1}{2}\trch\Asl&=-\bF+\Ga\c\Ga.
    \end{align*}
    Differentiating it by $\dkb^i$ and applying Proposition \ref{commu}, we infer for $i\leq s+3$
    \begin{align*}
    \nabs_4(\dkb^i\Ub)&=\rhoF^{(i)}+\Ga\c\Ga^{(s+3)},\\
    \nabs_4(\dkb^i\Asl)+\frac{1}{2}\trch(\dkb^i\Asl)&=\bF^{(i)}+\Ga\c\Ga^{(s+3)}.
    \end{align*}
Hence, we infer from Lemma \ref{evolutionlemma}, Proposition \ref{estbFrhoF} and Theorem \ref{Junkman}
\begin{align*}
\|(\Ub,\Asl)^{(s+3)}\|_{L^2(\ucuvs)}^2&\les\|(\Ub,\Asl)^{(s+3)}\|_{L^2(\ucuvss)}^2+\int_{\V}|(\bF,\rhoF)^{(s+3)}|^2+|\Ga\c\Ga^{(s+3)}|^2\\
&\les\ep_0^2+\int_{u_0}^u\int_{H_{u'}^{(\de,\ub)}}(\bF,\rhoF)^{(s+3)}|^2du'+\ep^4\les\ep_0^2.
\end{align*}
This concludes the proof of Proposition \ref{estA}.
\end{proof}
Combining Propositions \ref{estPsi4Psisl}--\ref{estA}, this concludes the proof of Theorem \ref{M1}.
\subsection{Proof of Theorem \ref{mainstability}}
We now use Theorems \ref{MM0}--\ref{MM2} to prove Theorem \ref{mainstability}.
\begin{df}\label{bootstrapstability}
For any $u_*\in[u_0,-\frac{1}{2}]$, let $\aleph(u_*)$ be the set of spacetimes $\V(u_*,1)$ associated with a double null foliation $(u,\ub)$ in which we have the following bounds:
\begin{align}
    \Tk\leq\ep,\qquad\quad\Rk\leq\ep,\qquad\quad\Ok\leq\ep.\label{BB2}
\end{align}
We also denote by $\mathcal{U}$ the set of values $u_*$ such that $\aleph(u_*)\ne\emptyset$.
\end{df}
The assumptions of Theorems \ref{mainstability} and \ref{MM0} imply that
\begin{equation*}
\Tk_{(0)}+\Rk_{(0)}+\Ok_{(0)}\les\ep_0.
\end{equation*}
Combining with the local existence theorem in \cite{luk}, we deduce that \eqref{BB2} holds if $u_*$ is sufficiently close to $u_0$. So, we have $\mathcal{U}\ne\emptyset$. Define $u_*$ as the supremum of the set $\mathcal{U}$. Applying Theorems \ref{M1}--\ref{MM2} in the region $\V(u_*,1)$ one by one, we obtain
\begin{equation}\label{TkRkOkep0}
    \Tk\les\ep_0,\qquad\quad\Rk\les\ep_0,\qquad\quad\Ok\les\ep_0.
\end{equation}
By a standard continuity argument, we have that $u_*=-\frac{1}{2}$ and \eqref{TkRkOkep0} holds in $\V(-\frac{1}{2},1)$. This concludes the proof of Theorem \ref{mainstability}.
\subsection{Short-pulse annulus and barrier annulus}
\begin{thm}\label{interiorsolution}
For any $s\in\mathbb{N}$, there exists a sufficiently large $a_0>0$. For any $a>a_0$ and $0<\de\leq a^{-2}$, there exists a spacelike initial data $(\Si,g,k,E,B,\psi,\phi,A,\Phi)$ solving \eqref{EMKGconstraint:intro}, endowed with a radial $r$--foliation for $r\in(0,\frac{3}{2})$, which satisfies the following properties: \begin{enumerate}
    \item We have
    \begin{align}
        \begin{split}\label{diffge}
        (g,k,E,B,\psi,\phi)&=(e,0,0,0,0,0)\quad\mbox{ in }\;\BB_{1-2\de},\\
        \|(g-e,k,E,B,\psi,\phi)\|_{\YY^s\big(\BB_\frac{3}{2}\setminus\ov{\BB_1}\big)}&\les a^{-1}.
        \end{split}
    \end{align}
    \item The charges of $(g,k,E,B)$ satisfy for $r\in(1,\frac{3}{2})$:
        \begin{align*}
        \Q_{ADM}[(g,k);\pr\BB_r],\Q_E[(g,E);\pr\BB_r]=O(a^{-1}),\qquad \Q_B[(g,B);\pr\BB_r]=0.
        \end{align*}
    \end{enumerate}
The spacelike initial data $(\Si,g,k,E,B)$ is denoted by $(\Si_{\de,a},g_{\de,a},k_{\de,a},E_{\de,a},B_{\de,a})$, or simply $\Si(\de,a)$. Moreover, the annulus $\BB_\frac{3}{2}\setminus\ov{\BB_1}$ is called the \emph{barrier annulus} , and $\BB_1\setminus\ov{\BB_{1-2\de}}$ is called the \emph{short-pulse annulus}.
\end{thm}
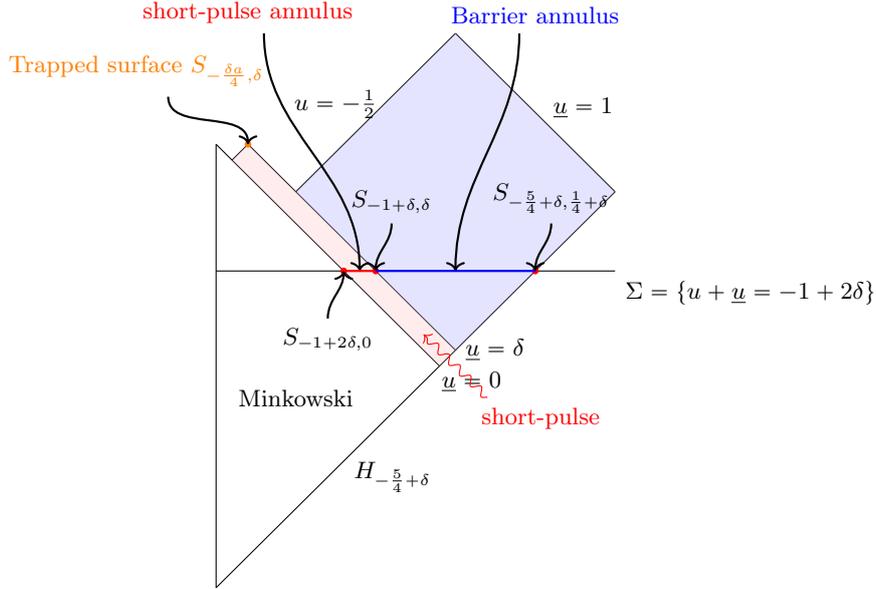
\begin{figure}[H]
\centering
\begin{tikzpicture}[scale=2.1, decorate]
  \coordinate (A) at (0,-2);
  \coordinate (B) at (0,0);
  \coordinate (C) at (0,0.8);
  \coordinate (D) at (0.1,0.7);
  \coordinate (E) at (0.2,0.8);
  \coordinate (F) at (1.4,-0.6);
  \coordinate (G) at (1.5,-0.5);
  \coordinate (H) at (2,0);
  \coordinate (I) at (2.5,0.5);
  \coordinate (J) at (1,0);
  \coordinate (K) at (0.5,0.5);
  \coordinate (L) at (1.5,1.5);
  \coordinate (M) at (2.5,0);
  \coordinate (N) at (0.8,0);
  \draw (A) -- (C);
  \draw (A) -- (I);
  \draw (F) -- (C);
  \draw (D) -- (E);
  \draw (G) -- (E);
  \draw (K) -- (L);
  \draw (I) -- (L);
  \draw (B) -- (M);
\fill[red!20,opacity=0.35](G)--(E)--(D)--(F)--cycle;
\fill[blue!30,opacity=0.35](G)--(I)--(L)--(K)--cycle;
  \node[right] at (1.35,-0.7) {\footnotesize $\ub = 0$};
  \node[right] at (G) {\footnotesize $\ub = \de$};
  \node[above] at (2.3,0.9) {\footnotesize $\ub = 1$};
  \node[above] at (0.75,0.9) {\footnotesize $u = -\frac{1}{2}$};
  \node[below right] at (M) {\footnotesize $\Si=\{u+\ub=-1+2\de\}$};
  \node at (0.5, -0.8) {\footnotesize Minkowski};
  \node[right] at (0.8, -1.3) {\footnotesize $H_{-\frac{5}{4}+\de}$};
  \node[below right] at (1.6,-0.8) {\footnotesize \red{short-pulse}};
\filldraw[red] (J) circle (0.5pt);
\filldraw[red] (H) circle (0.5pt);
\filldraw[red] (N) circle (0.5pt);
\filldraw[orange] (E) circle (0.5pt);
    \draw[->, thick, rounded corners=8pt] (1.1,0.3) to[out=-90, in=90] (J);
    \node[above] at (1.1,0.3) {\footnotesize $S_{-1+\de,\de}$};
    \draw[->, thick, rounded corners=8pt] (0.7,-0.3) to[out=90, in=-90] (N);
    \node[below] at (0.7,-0.3) {\footnotesize $S_{-1+2\de,0}$};
   \draw[->, thick, rounded corners=8pt] (2.1,0.3) to[out=-90, in=90] (H);
   \node[above] at (2.1,0.3) {\footnotesize $S_{-\frac{5}{4}+\de,\frac{1}{4}+\de}$};
    \draw[->, thick, rounded corners=8pt] (-0.3,1.1) to[out=-90, in=90] (E);
    \node[above] at (-0.5,1.1) {\footnotesize {\color{orange}Trapped surface $S_{-\frac{\de a}{4},\de}$}};
    \draw[->, thick, rounded corners=8pt] (0.3,1.5) to[out=-90, in=90] (0.9,0);
    \node[above] at (0.2,1.5) {\footnotesize \red{short-pulse annulus}};
    \draw[->, thick, rounded corners=8pt] (1.9,1.5) to[out=-90, in=90] (1.5,0);
    \node[above] at (2,1.5) {\footnotesize \blue{Barrier annulus}};
    \draw[->, decorate, decoration={snake, amplitude=0.5mm, segment length=2mm}, thin, red]  (1.7,-0.8) -- (1.3,-0.4);
    \draw[draw=blue, thick] (J)--(H);
    \draw[draw=red, thick] (N)--(J);
\end{tikzpicture}
\caption{constant-time slice $(\Si_{\de,a},g_{\de,a},k_{\de,a},E_{\de,a},B_{\de,a})$.}
\label{fig:shortpulse+stab}
\end{figure}
\begin{proof}
As an immediate consequence of Theorems \ref{thmtrapped} and \ref{mainstability}, the spacetime obtained in Theorem \ref{mainstability} satisfies the following properties:
    \begin{itemize}
    \item It coincides with the Minkowski spacetime for $\ub\leq 0$.
    \item In the region $\left\{(u,\ub)\in\left[u_0,-\frac{1}{2}\right]\times [0,\de]\right\}$, we have
    \begin{align}\label{trchctrchbcnotrapping}
        |\trchc,\trchbc|\les a^{-1}.
    \end{align}
    \item In the region $\left\{(u,\ub)\in\left[u_0,-\frac{1}{2}\right]\times [\de,1]\right\}$, we have
    \begin{align}\label{g-eta}
        \|(\g-\etabf,\psi,\F)\|_{C^s}\les a^{-1}.
    \end{align}
    \end{itemize}
    We now take
    \begin{align*}
    u_0=-\frac{5}{4}+\de,
    \end{align*}
    and we consider the following time-constant slice:
    \begin{align*}
        \Si:=\left\{u+\ub=-1+2\de\right\},
    \end{align*}
    see Figure \ref{fig:shortpulse+stab} for a geometric illustration. 
Taking $r:=\ub-u$, we obtain a radial foliation on $\Si$ for $r\in(0,\frac{3}{2})$. Notice that $\Si$ coincides with a constant-time slice of Minkowski spacetime for $r\leq 1-2\de$. Finally, we see from \eqref{g-eta} that \eqref{diffge} holds. This concludes the proof of Theorem \ref{interiorsolution}.
\end{proof}
\section{Main gluing arguments}\label{secgluing}
The goal of this section is to develop a gluing scheme for the full EMKG system adapted to our construction. By the localization of the Klein-Gordon field achieved in Section \ref{ssecinner}, the main gluing region is purely electrovacuum, and the EMKG constraints reduce exactly to the Einstein-Maxwell (EM) constraints. This reduction allows us to adapt the obstruction-free gluing method of~\cite{MOT}, originally designed for the vacuum case, to the electrovacuum setting. In Section \ref{secEMlinearglue}, we show that within this regime, the Maxwell equations glue linearly, with electromagnetic charges as the only compatibility conditions, while the Einstein sector glues obstruction-free as in~\cite{MOT}. In Section \ref{secobsfree}, we formulate and prove the resulting abstract EM gluing theorem. Its application to glue the inner barrier annulus of Theorem \ref{interiorsolution} with the outer charged Brill-Lindquist data \eqref{eq:gBL-EM} will be carried out later in Section \ref{secconstruction}. 
\subsection{Schematic EMKG constraint equations}
We begin by recalling the constraint equations for the EMKG system in schematic form. 
\begin{lem}\label{lem:trans}
We introduce (indices are raised or lowered by $e$)
\begin{align}
\begin{split}\label{dfhdfpi}
h_{ij}&:= g_{ij}-\de_{ij}-\de_{ij}\,\tr_e(g-e),\qquad\pi_{ij}:=k_{ij}-\de_{ij}\,\tr_e k,\\
\Eg^j&:=E^j\sqrt{\det g},\qquad\qquad\qquad\quad\;\;\,\,\Bg^j:=B^j\sqrt{\det g}.
\end{split}
\end{align}
Then, the inverse transformation formulas are given by
\begin{align*}
g_{ij}&=\de_{ij}+h_{ij}-\frac{1}{2}\de_{ij}\tr_e h,\qquad\quad k_{ij}=\pi_{ij}-\frac{1}{2}\de_{ij}\tr_e\pi,\\
E^j&=\frac{\Eg^j}{\sqrt{\det g}},\qquad\qquad\qquad\qquad\;\, B^j=\frac{\Bg^j}{\sqrt{\det g}}.
\end{align*}
\end{lem}
\begin{proof}
    Straightforward verification.
\end{proof}
\begin{proposition}\label{prop:schem-EMKG}
The EMKG constraint equations \eqref{EMKGconstraint} around the trivial data
\[
(g,k,E,B,\psi,\phi,A,\Phi)=(e,0,0,0,0,0,0),
\]
take the following schematic form:
\begin{align}
\begin{split}\label{schematicconstraint}
\pr_i\pr_j h^{ij}&=M(h,\pi,E,B,\psi,\phi,A,\Phi),\qquad\; \pr_i\pi^{ij}=N^j(h,\pi,E,B,\psi,\phi,A,\Phi),\\
\pr_i\Eg^i&=L(h,\psi,\phi,\Phi),\qquad\qquad\qquad\quad\;\,\pr_i\Bg^i=0,
\end{split}
\end{align}
where $M$, $N$ and $L$ take the schematic form
\begin{align*}
M(h,\pi,E,B,\psi,\phi,A,\Phi)&:=M_{EM}(h,\pi,E,B)+M_{KG}(h,\phi,\psi,A,\Phi),\\
N^j(h,\pi,E,B,\psi,\phi,A,\Phi)&:=N^j_{EM}(h,\pi,E,B)+N^j_{KG}(h,\phi,\psi,A,\Phi),\\
L(h,\psi,\phi,\Phi)&:=\psi\c\phi+\psi\c\psi\c\Phi,
\end{align*}
with\footnote{Here, $\cdot$ denotes the appropriate contraction of tensor fields.}
\begin{align*}
M_{EM}(h,\pi,E,B)&=h\c\pr^2 h+\pr h\c\pr h+\pi\c\pi+E\c E+B\c B,\\
M_{KG}(h,\phi,\psi,A,\Phi)&=\pr\psi\c\pr\psi+\pr\psi\c\psi\c A+\psi\c\psi\c A\c A+\phi\c\phi+\phi\c\psi\c\Phi+\psi\c\psi\c\Phi\c\Phi,\\
N_{EM}^j(h,\pi,E,B)&=(h\c\pr\pi)^j+(\pr h\c\pi)^j+(E\c B)^j,\\
N_{KG}^j(h,\psi,\phi,A,\Phi)&=(\pr\psi\c\phi)^j+(\pr\psi\c\psi\c\Phi)^j+(\phi\c\psi\c A)^j+(\psi\c\psi\c A\c\Phi)^j.
\end{align*}
\end{proposition}
In the remainder of this paper, we always denote $4$--tuple $X$ for EM data 
\begin{equation}\label{defX}
X:=(h,\pi,\Eg,\Bg),
\end{equation}
and $\vec{P}$ for the principal parts of EMKG/EM constraints and $\vec{N}$ for the nonlinear parts
\begin{align}\label{defvecPvecN}
    \vec{P}(X):=\begin{pmatrix}
     \pr_i\pr_j h^{ij}\\ \pr_i\pi^{ij}\\ \pr_i\Eg^i\\ \pr_i\Bg^i
    \end{pmatrix},\qquad\quad\vec{N}(X,\psi,\phi,A,\Phi):=\begin{pmatrix}
        M(h,\pi,E,B,\psi,\phi,A,\Phi)\\ N^j(h,\pi,E,B,\psi,\phi,A,\Phi)\\L(h,\psi,\phi,\Phi) \\ 0
    \end{pmatrix}.
\end{align}
Then, the schematic form of the EMKG constraint equations \eqref{schematicconstraint} becomes
\begin{align*}
    \vec{P}(X)=\vec{N}(X,\psi,\phi ,A,\Phi).
\end{align*}
\subsection{Analytic setup: Bogovskii/conic operators and conservation laws}
We introduce the analytic tools used in the gluing construction. Bogovskii-type and conic operators invert the divergence operators in~\eqref{schematicconstraint} while preserving annular support, and the accompanying flux identities control the variation of the ADM and electromagnetic charges under the EM constraints.
\subsubsection{Bogovskii-type and conic operators}
We now state the main tool for inverting the operator on the L.H.S. of \eqref{schematicconstraint} while preserving the annular support property.
\begin{lem}[Bogovskii-type operators on annuli]\label{lem:annulus-bog}
There exist linear operators
\begin{align*}
S:C_c^\infty(\AA_1)&\to C^\infty(\AA_1;\mathrm{Sym}^2\RRR^3),\\
T:C_c^\infty(\AA_1;\RRR^3)&\to C^\infty(\AA_1;\mathrm{Sym}^2\RRR^3),\\
P:C_c^\infty(\AA_1)&\to C^\infty(\AA_1;\RRR^3),
\end{align*}
such that:
\begin{enumerate}
\item Support property:
\[
\supp(Sf),\;\supp(T\f),\;\supp(Pf)\subseteq\AA_1.
\]
\item Right-invertibility property under compatibility conditions:
\begin{align*}
\pr_i\pr_j(Sf)^{ij}&=f &&\mbox{ if }\quad \int_{\AA_1} f\cdot (1,x_1,x_2,x_3)^\top dx=0,\\
\pr_i(T\f)^{ij}&=\f^j &&\mbox{ if }\quad \int_{\AA_1}\f\cdot (\ev_1,\ev_2,\ev_3,\vec{Y}_1,\vec{Y}_2,\vec{Y}_3)^\top dx=0,\\
\pr_i(Pf)^i&=f &&\mbox{ if }\quad \int_{\AA_1} f\,dx=0.
\end{align*}
\item For all $s\in\RRR$,
\[
\|Sf\|_{H^s}\les \|f\|_{H^{s-2}},\qquad\|T\f\|_{H^s}\les \|\f\|_{H^{s-1}},\qquad\|Pf\|_{H^s}\les \|f\|_{H^{s-1}}.
\]
\item The commutator estimates hold for all $j=1,2,3$
\[
\|[S,\pr_j]f\|_{H^s}\les\|f\|_{H^{s-2}},\qquad\|[T,\pr_j]\f\|_{H^s}\les\|\f\|_{H^{s-1}},\qquad\|[P,\pr_j]f\|_{H^s}\les\|f\|_{H^{s-1}}.
\]
\end{enumerate}
\end{lem}
\begin{proof}
Their existence, along with the stated support and mapping properties, follows from Lemma 2.2 in \cite{MOT} and the original construction of Bogovskii \cite{Bog}.
\end{proof}
We now define the $b$--Sobolev space.
\begin{df}\label{dfbsobolev}
For $s\in\NNN$, the $b$--Sobolev space $H^s_b(\RRR^3)$ is defined by the norm
\[
\|u\|^2_{H_b^s(\RRR^3)}:=\sum_{k\leq s}\|\xja^k\nab^k u\|_{L^2(\RRR^3)}^2.
\]
We extend the definition to $s\in\RRR$ by duality and interpolation. For $\de\in\RRR$, we set $H_b^{s,\ell}:=\xja^{-\ell}H_b^s$.
\end{df}
Given $\th\in(0,\pi)$, $\om\subseteq\SSS^2$ and $\bom\in\SSS^2$, we define the cones
\begin{align*}
C_\om=\left\{x\in\RRR^3\Big/\,\frac{x}{r}\in\om\right\},\qquad\quad C_\th(\bom)=\{x\in\RRR^3/\,\angle (x,\bom)<\th\}.
\end{align*}
\begin{lem}\label{lem:conic}
Let $\om\subseteq\SSS^2$ be a convex open subset and $\ka\in C^\infty(\SSS^2)$ with $\int_{\SSS^2}\ka=1$ and $\supp\ka\subseteq\om$. Then, there exist linear \emph{translation-invariant} operators
\begin{align*}
   S_c: C_c^\infty(C_\om)&\to C^\infty(C_\om;\mathrm{Sym}^2\RRR^3),\\
   T_c: C_c^\infty(C_\om;\RRR^3)&\to C^\infty(C_\om;\mathrm{Sym}^2\RRR^3),\\
   P_c: C_c^\infty(C_\om)&\to C^\infty(C_\om;\RRR^3),
\end{align*}
satisfying the following properties:
\begin{enumerate}
\item Support property:
\[\supp(S_cf),\ \supp(T_c\f),\ \supp(P_cf)\subseteq C_\om;\]
\item Right-invertibility property:
\[\pr_i\pr_j(S_cf)^{ij}=f,\qquad \pr_i(T_c\f)^{ij}=\f^j,\qquad \pr_i(P_cf)^i=f;\]
\item Mapping properties: for all $s\in\RRR$ and suitable weights,
\begin{align*}
\|S_cf\|_{H_b^{s,\de}}&\les \|f\|_{H_b^{s-2,\de+2}}, \qquad\qquad\de<-\frac{1}{2},\\
\|T_c\f\|_{H_b^{s,\de}}&\les \|\f\|_{H_b^{s-1,\de+1}},\qquad\qquad\,\de<\frac{1}{2},\\
\|P_cf\|_{H_b^{s,\de}}&\les \|f\|_{H_b^{s-1,\de+1}},\qquad\qquad\de<\frac{1}{2}.
\end{align*}
\end{enumerate}
\end{lem}
\begin{proof}
These are the conic analogs of the Bogovskii-type operators, obtained by convolution with homogeneous kernels adapted to $C_\om$. Their existence, translation-invariance, support preservation, and the stated $b$--Sobolev mapping properties follow from Lemma~2.5 in \cite{MOT}. See also the construction in \cite{OT} and Proposition 7 in \cite{MaoTao}.
\end{proof}
We introduce the shorthand notation
\begin{equation}\label{defvecS}
\vec{S}:=(S,T,P,P), \qquad \vec{S}_c:=(S_c,T_c,P_c,P_c),
\end{equation}
for the Bogovskii-type and conic-type operator families, respectively. Moreover, we define the following product spaces:
\begin{align*}
\XX^s&:=H^s\times H^{s-1}\times H^{s-1}\times H^{s-1},\\
\YY^s&:=H^s\times H^{s-1}\times H^{s-1}\times H^{s-1}\times H^{s-1}\times H^{s-1},\\
\XX_0^s&:=H_0^s\times H_0^{s-1}\times H_0^{s-1}\times H_0^{s-1},\\
\YY_0^s&:=H_0^s\times H_0^{s-1}\times H_0^{s-1}\times H_0^{s-1}\times H_0^{s-1}\times H_0^{s-1},\\
\XX_b^{s,\ell}&:=H_b^{s,\ell}\times H^{s-1,\ell+1}_b\times H^{s-1,\ell+1}_b\times H^{s-1,\ell+1}_b.
\end{align*}
Here $\XX$ corresponds to the EM initial data $X=(h,\pi,\Eg,\Bg)$, while $\YY$ corresponds to the EMKG initial data $(X,\psi,\phi)$.
\subsubsection{Conservation laws}
We record flux identities that later yield the variation of the ADM charges $\Q_{ADM}$ and, under the EM constraints, the conservation of the electromagnetic charges $\Q_E,\Q_B$. For $0<r_0<r_1$ and $\eta_r$ defined in \eqref{scaleeta}, we introduce
\begin{equation*}
    \chi_{r_0,r_1}(r)=\int_{-\infty}^r (\eta_{r_0}(r')-\eta_{r_1}(r'))dr'.
\end{equation*}
\begin{lem}\label{Lemconservationlaw}
Let $(g,k,E,B)$ and $(h,\pi,\Eg,\Bg)$ be defined by \eqref{dfhdfpi}. For any $0<r_0<r_1$, we have for $k,l=1,2,3$:
    \begin{align*}
        \int_{\BB_{r_1}\setminus\BB_{r_0}}\frac{1}{2}\pr_i\pr_j h^{ij}\begin{pmatrix}
            1 \\ x_l
        \end{pmatrix}dx&=\begin{pmatrix}
            \E \\ \C_l
        \end{pmatrix}[(g,k);\pr\BB_{r_1}]-\begin{pmatrix}
            \E \\ \C_l
        \end{pmatrix}[(g,k);\pr\BB_{r_0}],\\
         \int_{\BB_{r_1}\setminus\BB_{r_0}}\pr_i\pi^{ij}\begin{pmatrix}
            (\ev_k)^j  \\ (\vec{Y}_l)^j
        \end{pmatrix}dx&=\begin{pmatrix}
            \P_k \\ \J_l
        \end{pmatrix}[(g,k);\pr\BB_{r_1}]-\begin{pmatrix}
            \P_k \\ \J_l
        \end{pmatrix}[(g,k);\pr\BB_{r_0}],\\
        \int_{\BB_{r_1}\setminus\BB_{r_0}}\begin{pmatrix}
        \pr_i\Eg^i\\ \pr_i\Bg^i
        \end{pmatrix}dx&=\begin{pmatrix}
        \Q_E\\ \Q_B
        \end{pmatrix}[(g,E,B);\pr\BB_{r_1}]-\begin{pmatrix}
        \Q_E\\ \Q_B\end{pmatrix}[(g,E,B);\pr\BB_{r_0}].
    \end{align*}
    In terms of averaged charges, we have for $k,l=1,2,3$:
    \begin{align*}
        \int\chi_{r_0,r_1}(r)\frac{1}{2}\pr_i\pr_j h^{ij}\begin{pmatrix}
            1 \\ x_l
        \end{pmatrix}dx&=\begin{pmatrix}
            \E \\ \C_l
        \end{pmatrix}[(g,k);\AA_{r_1}]-\begin{pmatrix}
            \E \\ \C_l
        \end{pmatrix}[(g,k);\AA_{r_0}],\\
         \int\chi_{r_0,r_1}(r)\pr_i\pi^{ij}\begin{pmatrix}
            (\ev_k)^j \\ (\vec{Y}_l)^j 
        \end{pmatrix}dx&=\begin{pmatrix}
            \P_k \\ \J_l
        \end{pmatrix}[(g,k);\AA_{r_1}]-\begin{pmatrix}
            \P_k \\ \J_l
        \end{pmatrix}[(g,k);\AA_{r_0}],\\
        \int\chi_{r_0,r_1}(r)\begin{pmatrix}
        \pr_i\Eg^i\\ \pr_i\Bg^i
        \end{pmatrix}dx&=\begin{pmatrix}
            \Q_E\\ \Q_B
        \end{pmatrix}[(g,E,B);\AA_{r_1}]-\begin{pmatrix}
            \Q_E\\ \Q_B\end{pmatrix}[(g,E,B);\AA_{r_0}].
    \end{align*}
\end{lem}
\begin{proof}
It follows directly from \eqref{dfhdfpi}, Definition \ref{chargesEMKG} of $\Q_{ADM}$, $\Q_E$, $\Q_B$ and integration by parts. In particular, these charges are conserved under linearized EMKG constraint equations around Euclidean space.
\end{proof}
\subsection{Construction of star shell and near-star annulus}\label{ssecinner}
We reduce the EMKG constraints to the pure EM constraints in the gluing region by localizing the Klein-Gordon field inside the star region. The resulting electrovacuum annulus, constructed in Theorem~\ref{constructioninnerannulus}, will serve as the domain for the EM gluing.

To treat the electromagnetic potentials $(A,\Phi)$ in the EMKG constraints \eqref{EMKGconstraint}, we now recall the relevant $3$--dimensional elliptic estimates.
\begin{prop}\label{hodgerank1}
We have the following properties:
\begin{enumerate}
    \item Let $\xi$ be a $1$--form on $\Si$ satisfying $\sdiv_g\xi=D(\xi)$ and $\curl_g\xi=C(\xi)$.
Then, we have
\begin{align*}
\int_{\Si}|\nab\xi|^2=\int_\Si|C(\xi)|^2+|D(\xi)|^2-\sRic^{ij}\xi_i\xi_j.
\end{align*}
\item Let $\phi$ be a scalar function on $\Si$ satisfying $\De_g\phi=f$. Then, we have
\begin{align*}
    \int_{\Si}|\nab\phi|^2+|\nab^2\phi|^2\les\int_\Si r^2|f|^2.
\end{align*}
\end{enumerate}
\end{prop}
\begin{proof}
See Lemma 4.4.1 and Proposition 4.2.2 in \cite{Ch-Kl}.
\end{proof}
\begin{prop}\label{propAest}
Under the Coulomb gauge $ \sdiv_g A:=\nab_iA^i=0$,
and the asymptotically flat assumption $ \lim_{|\x|\to\infty}(A,\Phi)=0$, 
the initial data set for the EMKG system satisfies \eqref{EMKGconstraint} as well as the following determined system for $(A,\Phi)$:
\begin{align}\label{EMKGC}
\sdiv_gA=0,\qquad \curl_gA=B,\qquad \De_g\Phi=2\ef\Im\left(\psi^\dag\phi\right),\qquad\lim_{|\x|\to\infty}(A,\Phi)=0.
\end{align}
\end{prop}
\begin{proof}
It follows directly from \eqref{EMKGconstraint} and Proposition \ref{hodgerank1}.
\end{proof}

We are now ready to prove the following theorem, which constructs the star shell and the near-star annulus.
\begin{thm}\label{constructioninnerannulus}
Let $(g,k,E,B,\psi,\phi,A,\Phi)$ be the initial data set in Theorem \ref{interiorsolution} that solves \eqref{EMKGconstraint} on $\BB_\frac{3}{2}$. Then, there exists an initial data set $(\gt,\kt,\Et,\Bt,\psit,\phit,\AAt,\Phit)$ that solves \eqref{EMKGconstraint} on $\BB_4$ and satisfies the following properties:
    \begin{enumerate}
        \item We have
        \begin{align}
        (\gt,\kt,\Et,\Bt,\psit,\phit,\AAt,\Phit)&=(g,k,E,B,\psi,\phi,A,\Phi)\quad\mbox{ on }\;\BB_\frac{3}{2},\\
        \supp(\psit,\phit)&\subseteq\BB_2.
        \end{align}
        \item The following estimate hold:
        \begin{align}
        \begin{split}\label{innerannSobolev}
        \|(\gt-e,\kt,\Et,\Bt,\psit,\phit)\|_{\YY^s(\BB_4\setminus\ov{\BB_1})}\les a^{-1}.
        \end{split}
        \end{align}
        \item The charges of $(\gt,\kt,\Et,\Bt)$ satisfy for $r\in(1,4)$:
        \begin{align}
        \begin{split}\label{innerannQ}
        \Q_{ADM}[(\gt,\kt);\pr\BB_r],\Q_E[(\gt,\Et);\pr\BB_r]=O(a^{-1}),\qquad\Q_B[(\gt,\Bt);\pr\BB_r]=0.
        \end{split}
        \end{align}
    \end{enumerate}
    The annulus $\BB_2\setminus\BB_\frac{3}{2}$ is called a \emph{star shell} while the annulus $\AA_2=\BB_4\setminus\ov{\BB_2}$ is called a \emph{near-star annulus}, in which Klein-Gordon fields vanish. Moreover, the disk $\BB_2$ is called the \emph{star region}.
\end{thm}
\begin{proof}
As an consequence of Theorem \ref{interiorsolution}, we have
\[
\|(g-e,k,E,B,\psi,\phi)\|_{\YY^s(\BB_\frac{3}{2}\setminus\ov{\BB_1})}\les a^{-1}.
\]
We first \emph{freely extend} $(g,k,E,B,\psi,\phi)_{in}$ to $\BB_4$ so that
\[
\|(g-e,k,E,B,\psi,\phi)_{in}\|_{\YY_0^s(\BB_4\setminus\ov{\BB_1})}\les\ep,\qquad \supp(\psi,\phi)\subseteq\BB_2.
\]
We now fix $(\psit,\phit):=(\psi,\phi)_{in}$, so that $(A,\Phi)$ is determined elliptically from \eqref{EMKGC}, and we solve for the EM variables via the fixed-point problem:
\begin{equation}\label{fixedpoint}
\Xt = \vec{S}\left(\vec{N}(X_{in}+\Xt,\psit,\phit,A,\Phi)-\vec{P}(X_{in})\right)\qquad\text{ on }\;\BB_4\setminus\BB_\frac{3}{2},
\end{equation}
where $\vec{S}=(S,T,P,P)$ is the Bogovskii operator. By Lemma \ref{lem:annulus-bog}, $\vec{S}$ preserves support in the annulus, so the solution $\Xt$ is compactly supported in $\BB_4\setminus\ov{\BB_\frac{3}{2}}$. By Banach’s fixed point theorem, there exists a unique $\Xt\in\XX^s(\BB_4\setminus\BB_\frac{3}{2})$ solving \eqref{fixedpoint}. Letting
\[
\Xt_{in}:=X_{in}+\Xt
\]
and writing $(\gt,\kt,\Et,\Bt)$ for the variables corresponding to $\Xt_{in}=(\hti,\pit,\Egt,\Bgt)_{in}$, the estimate \eqref{innerannSobolev} follows from the EMKG constraint \eqref{EMKGconstraint} and the smallness of the data. The charge estimates \eqref{innerannQ} are consequences of Lemma \ref{Lemconservationlaw}. This completes the proof of Theorem \ref{constructioninnerannulus}.
\end{proof}
\subsection{Gluing up to linear obstructions for EM constraints}\label{secEMlinearglue}
Outside the inner annulus, the EMKG constraints reduce to the EM constraints. Under the renormalization of Lemma~\ref{lem:trans},
\[
\Eg^i=\sqrt{\det g}\,E^i,\qquad\quad\Bg^i=\sqrt{\det g}\,B^i,
\]
the Maxwell equations become the linear divergence equations in Euclidean space,
\[
\pr_i\Eg^i=0,\qquad\quad\pr_i\Bg^i=0,
\]
so that only the total electromagnetic charges remain as compatibility conditions for the gluing. In particular, as a corollary of Proposition~\ref{prop:schem-EMKG}, the EM constraint takes the following schematic form:
\begin{cor}\label{cor:schem-EM}
In the case $(\psi,\phi)=(0,0)$, the EM constraint equations \eqref{EMconstraint} take the schematic form
\begin{align}\label{EMschem}
\begin{split}
    \pr_i\pr_j h^{ij}&=M_{EM}(h,\pi,E,B),\qquad\quad\pr_i\pi^{ij}=N_{EM}^j(h,\pi,E,B),\\
\pr_i\Eg^i&=0,\qquad\qquad\qquad\qquad\quad\;\;\;\pr_i\Bg^i=0.
\end{split}
\end{align}
We also denote $\vec{N}_{EM}(X):=\big(M_{EM}(X),N^j_{EM}(X),0,0\big)^\top$.
\end{cor}
Thus, the Maxwell sector introduces no nonlinear obstructions, and the only obstructions to be removed come from the gravitational variables $(g,k)$, exactly as in the vacuum case of \cite{MOT}. In this regime, the gluing theorem below is a direct EM extension of Theorem 1.3 in \cite{MOT}. To this end, we first state the following lemma.
\begin{lem}\label{MOT3.1}
Let $u,v$ be Schwartz functions on $\RRR^3$. 
\begin{enumerate}
\item Let $F$ be a $C^\infty$ function with bounded derivatives. Then, we have for $s>0$,
\begin{equation*}
\|F(u)-F(0)\|_{H^s}\les\|u\|_{H^s}.
\end{equation*}
\item Let $F$ be a $C^\infty$ function with bounded derivatives. Then, we have for $s>0$,
\begin{equation*}
\|F(u)-F(v)\|_{H^s}\les\|u-v\|_{H^s}+\|u-v\|_{L^\infty}(\|u\|_{H^s}+\|v\|_{H^s}).
\end{equation*}
\item For $s_0,s_1,s_2$ such that $s_0+s_1+s_2\geq\frac{3}{2}$, $s_0+s_1+s_2\geq\max\{s_0,s_1,s_2\}$ with at least one of the inequalities being strict, we have 
\begin{equation*}
\|uv\|_{H^{-s_0}}\les\|u\|_{H^{s_1}}\|v\|_{H^{s_2}}.
\end{equation*}
\end{enumerate}
\end{lem}
\begin{proof}
See Section 2.8 in \cite{BCD}.
\end{proof}
\begin{thm}\label{MOT1.3}
Let $s>\frac{3}{2}$, there exist $\eps_c=\eps_c(s)>0$ and $M_c=M_c(s)>0$ such that the following holds. Let $(\go,\ko,\Eo,\Bo)\in\XX^s(\At_{16})$ be a solution to the EM constraint equations \eqref{EMconstraint} such that
\begin{equation}\label{gokoperturbation}
    \left\|(\go-e,\ko,\Eo,\Bo)\right\|_{\XX^s(\At_{16})}\leq\eps.
\end{equation}
Define $\Qo\in\RRR^{10}$ by $\Qo=\Q_{ADM}[(\go,\ko);\AA_{16}]$. Let $\QQ\subseteq\RRR^{10}$ be a bounded open set such that
\begin{equation}\label{Qoopen}
\Qo\in\QQ,\qquad\quad\BB_{M_c\eps^2}(\Qo)\subseteq\QQ,
\end{equation}
and consider an $\QQ$--admissible family of annular initial data sets $\{(g_Q,k_Q,E_Q,B_Q)\}_{Q\in\QQ}$ on $\At_{16}$, as in Definition \ref{MOT1.2}, with fixed electromagnetic charges $(Q_E,Q_B)$ on $
\AA_{16}\subset \At_{16}$ and Sobolev regularity $s$ such that, for all $Q,Q'\in\QQ$,
\begin{align}
\|(g_Q-e,k_Q,E_Q,B_Q)\|_{\XX^s(\At_{16})}&\leq\eps, \nonumber\\
\|(g_Q-g_{Q'},k_Q-k_{Q'},E_Q-E_{Q'},B_Q-B_{Q'})\|_{\XX^s(\At_{16})}&\leq K|Q-Q'|.\label{gQgQ'kQkQ'}
\end{align}
We also assume that the electromagnetic charges of $(\go,\ko,\Eo,\Bo)$ coincide with those of $(g_Q,k_Q,E_Q,B_Q)$ on $\AA_{16}$, i.e.
\begin{align}\label{chargeconservation}
\begin{split}
Q_E=\Q_E[(g_Q,E_Q);\AA_{16}]&=\Q_E[(\go,\Eo);\AA_{16}],\\
Q_B=\Q_B[(g_Q,B_Q);\AA_{16}]&=\Q_B[(\go,\Bo);\AA_{16}].
\end{split}
\end{align}
Assume that $(1+K)\eps<\eps_c$. Then, the following holds:
\begin{enumerate}
\item {\bf Gluing.} There exist $(g,k,E,B)\in\XX^s(\At_{16})$ solving EM constraint equations \eqref{EMconstraint} and $Q\in\QQ$ such that $\|(g-e,k,E,B)\|_{\XX^s(\At_{16})}\les\eps$, $|Q-\Qo|\leq M_c\eps^2$, and
    $$
    (g,k,E,B)=\left\{\begin{aligned}
    (g_Q,k_Q,E_Q,B_Q)\qquad &\mbox{ in }\;\AA_{8},\\
    (\go,\ko,\Eo,\Bo)\qquad &\mbox{ in }\;\AA_{32}.
    \end{aligned}\right.
    $$
    In particular, electromagnetic charges are conserved:
    \begin{align*}
        \Q_E[(g,E);\AA_{16}]=Q_E,\qquad \Q_B[(g,B);\AA_{16}]=Q_B.
    \end{align*}
\item {\bf Lipschitz continuity.} The map $(\go,\ko,\Eo,\Bo)\mapsto(g,k,E,B,Q)$ is Lipschitz as a map from the subset of $\XX^s$ restricted by \eqref{gokoperturbation} and \eqref{Qoopen} into $\XX^s\times\RRR^{10}$.
\item {\bf Persistence of regularity.} Assume that $(\go,\ko,\Eo,\Bo)\in\XX^{s+m}(\At_{16})$ and for each $Q\in\QQ$, $(g_Q,k_Q,E_Q,B_Q)$ is of Sobolev regularity $s+m$ for $m\in\mathbb{N}$, then $(g,k,E,B)\in\XX^{s+m}(\At_{16})$ and
    \begin{align*}
    \|(g-e,k,E,B)\|_{\XX^{s+m}(\At_{16})}&\les\|(\go-e,\ko,\Eo,\Bo)\|_{\XX^{s+m}(\At_{16})}\\
    &+\|(g_Q-e,k_Q,E_Q,B_Q)\|_{\XX^{s+m}(\At_{16})}.
    \end{align*}
\end{enumerate}
\end{thm}
\begin{proof}
The proof of Theorem \ref{MOT1.3} is divided into 3 steps.

\paragraph{Step 1. Construction of the preliminary gluing profile $\Xb_Q$ for each $Q$.}
Let $\chi(x)$ be a smooth radial cutoff equal to $0$ for $|x|<16$ and $1$ for $|x|>32$. 
We form the following interpolating ansatz:
\begin{align}
\begin{split}\label{defhbpib}
\Xb_Q:=(1-\chi)X_Q+\chi \Xo,
\end{split}
\end{align}
where $X_{Q}=(g_Q,k_Q,E_Q,B_Q)$ and $\Xo=(\go,\ko,\Eo,\Bo)$ are defined in \eqref{defX}. 
Certainly, $\Xb_Q$ would not solve \eqref{EMschem}.
The goal is, therefore, to find a correction $\Xt_Q$, supported in $\AA_{16}$, and a suitable choice of $Q\in\QQ$ such that
\[
X=\Xb_Q+\Xt_Q
\]
becomes a genuine solution of \eqref{EMschem}. Plugging this into the constraint yields
\begin{align}
\begin{split}\label{htiQpitQerror}
\vec{P}\Xt_Q=-\vec{P}\,\Xb_Q+\vec{N}_{EM}(X),
\end{split}
\end{align}
where $\vec{P}$ is the principal (linear) part as in \eqref{defvecPvecN} and $\vec{N}_{EM}$ is the nonlinear part of the EM constraint, as in Corollary \ref{cor:schem-EM}. Applying the Bogovskii-type operator $\vec{S}$ from \eqref{defvecS} \emph{formally} leads to the fixed-point map
\begin{equation}\label{fixedpointhtiQpitQ}
\Xt_Q=\vec{S}\big(\vec{N}_Q(\Xt_Q)\big)=\vec{S}\big(-\vec{P}\Xb_Q+\vec{N}_{EM}(\Xb_Q+\Xt_Q)\big),
\end{equation}
where $\vec{N}_Q:=(\Mt_Q,\Nt_Q,\Lt_Q,\Kt_Q)^\top$ denotes the right-hand side of \eqref{htiQpitQerror}. 
By Lemma~\ref{lem:annulus-bog}, however, $\vec{S}$ is a right-inverse only on the subspace of data satisfying the finite-dimensional compatibility (orthogonality) conditions. 
Thus \eqref{fixedpointhtiQpitQ} should be viewed as a \emph{relaxed} fixed-point problem: for a given $Q$ we can solve it precisely when $\vec{N}_Q(\Xt_Q)$ meets those conditions. 
In Step 3 we will choose $Q$ so that the compatibility conditions hold, thereby making \eqref{fixedpointhtiQpitQ} solvable and turning $\Xt_Q$ into a genuine correction.

\paragraph{Step 2. Finding the candidate $\Xt_Q$ for each $Q$.}

For each $Q\in\QQ$, let $\Xt_Q\in\XX$ with
\begin{align*}
\XX:=\left\{\Xt_Q\Big/\,\|\Xt_Q\|_{\XX^s(\AA_{16})}\leq M_c\ep\right\}.
\end{align*}
We denote
\begin{align*}
\vec{F}_Q:=-\vec{P}\,\Xb_Q+\vec{N}_{EM}(\Xb_Q).
\end{align*}
As an immediate consequence of \eqref{defhbpib} and Lemma \ref{MOT3.1}, we have
\begin{align*}
\supp(\vec{F}_Q)\subseteq \AA_{16},\qquad\quad\|\vec{F}_Q\|_{H^{s-2}(\AA_{16})}\les\ep.
\end{align*}
Noticing that
\begin{align*}
    \vec{N}_Q-\vec{F}_Q=\vec{N}_{EM}(\Xb_Q+\Xt_Q)-\vec{N}_{EM}(\Xb_Q),
\end{align*}
we have from \eqref{defhbpib} and Lemma \ref{MOT3.1}
\begin{align*}
\supp(\vec{N}_Q-\vec{F}_Q)\subseteq\AA_{16},\qquad\|\vec{N}_Q-\vec{F}_Q\|_{H^{s-2}(\AA_{16})}\les M_c\ep^2.
\end{align*}
Moreover, we have for $Q,Q'\in\QQ$
\begin{align*}
\|\vec{N}_Q-\vec{N}_{Q'}\|_{H^{s-2}(\AA_{16})}\les\ep\|\Xt_Q-\Xt_{Q'}\|_{\XX^s(\AA_{16})}.
\end{align*}
Applying Lemma \ref{lem:annulus-bog}, we deduce that the operator $\vec{S}$ in \eqref{fixedpointhtiQpitQ} is a contraction on $\XX$. Hence, by Banach's fixed point theorem, there exists a unique $\Xt_Q\in\XX$ solving \eqref{fixedpointhtiQpitQ}. 
At this stage $\Xt_Q$ is only a \emph{candidate} correction, since the compatibility conditions are not yet enforced. These will be used later to determine the correct $Q$. 
Finally, by varying $Q$ and repeating the same argument, one verifies that the map $Q\mapsto\Xt_Q\in\XX^s(\AA_{16})$ is locally Lipschitz.
\paragraph{Step 3. Finding the correct $Q$.}
In order for $\Xt_Q$ to solve the EM constraint (and not merely the relaxed fixed point equation), the compatibility conditions from Lemma~\ref{lem:annulus-bog} must be satisfied. These conditions pick out the correct value of $Q$. Thus, it remains to find $Q\in\QQ$ such that
\begin{align}
\begin{split}\label{orthogonalityMQNQ}
\int\Mt_Q\c (1,x_1,x_2,x_3)^\top dx&=0,\\
\int\Nt_Q\c(\ev_1,\ev_2,\ev_3,\vec{Y}_1,\vec{Y}_2,\vec{Y}_3)^\top dx&=0,\\
\int(\Lt_Q,\Kt_Q)^\top dx&=0,
\end{split}
\end{align}
which ensures that $\Xt_Q$ satisfies Statement 2 of Lemma \ref{lem:annulus-bog}. Applying Lemma \ref{Lemconservationlaw}, we obtain for $l=1,2,3$:
\begin{align*}
    \int\frac{1}{2}\Mt_Q\begin{pmatrix}
        1\\ x_l
    \end{pmatrix} dx&=-\int_{\At_1}\chi_{8,32}(r)\frac{1}{2}\pr_i\pr_j\hb_Q^{ij}\begin{pmatrix}
        1\\ x_l
    \end{pmatrix}dx+\int_{\At_1}\chi_{8,32}(r)\frac{1}{2}M_{EM}(h,\pi,E,B)\begin{pmatrix}
        1\\ x_l
    \end{pmatrix}dx \\
    &=\begin{pmatrix}
    \E\\ \C_l
    \end{pmatrix}[(g_Q,k_Q);\AA_{8}]-\begin{pmatrix}
    \E\\ \C_l
    \end{pmatrix}[(\go,\ko);\AA_{32}]\\
    &+\int_{\At_1}\chi_{8,32}(r)\frac{1}{2}M_{EM}(h,\pi,E,B)\begin{pmatrix}
    1\\ x_l
    \end{pmatrix}dx.
\end{align*}
Next, we recall from \eqref{EMschem}
\begin{align*}
    \pr_i\pr_j h^{ij}_{Q}=M_{EM}(h_Q,\pi_Q,E_Q,B_Q).
\end{align*}
Multiplying it by $\chi_{8,16}(r)$ and integrating it by parts, we obtain from Lemma \ref{Lemconservationlaw}
\begin{align*}
    \begin{pmatrix}\E\\ \C_l\end{pmatrix}[(g_Q,k_Q);\AA_{8}]&=\begin{pmatrix}\E\\ \C_l\end{pmatrix}[(g_Q,k_Q);\AA_{16}]+\int_{\At_{16}}\chi_{8,16}(r)\frac{1}{2}M_{EM}(h_Q,\pi_Q,E_Q,B_Q)\begin{pmatrix}
    1\\ x_l\end{pmatrix}dx.
\end{align*}
Similarly, we have
\begin{align*}
    \begin{pmatrix}\E\\ \C_l\end{pmatrix}[(\go,\ko);\AA_{16}]&=\begin{pmatrix}\E\\ \C_l\end{pmatrix}[(\go,\ko);\AA_{32}]+\int_{\At_{16}}\chi_{16,32}(r)\frac{1}{2}M_{EM}(\ho,\pio,\Eo,\Bo)\begin{pmatrix}
    1\\ x_l
    \end{pmatrix}dx.
\end{align*}
Combining the above identities, we deduce
\begin{align*}
    \int \frac{1}{2}\Mt_Q\begin{pmatrix}
        1 \\ x_l
    \end{pmatrix}dx&=\begin{pmatrix}
        \E \\ \C_l
    \end{pmatrix}[(g_Q,k_Q);\AA_{16}]-\begin{pmatrix}
        \E \\ \C_l
    \end{pmatrix}[(\go,\ko);\AA_{16}]+\begin{pmatrix}
        n[Q]_\E \\ n[Q]_{\C_l}
    \end{pmatrix},
\end{align*}
where $n[Q]$ denotes the nonlinear terms. Similarly, we have for $k,l=1,2,3$
\begin{align*}
    \int\Nt_Q\begin{pmatrix}
        \ev_k\\ \vec{Y}_l
    \end{pmatrix}dx&=\begin{pmatrix}
        \P_k \\ \J_l
    \end{pmatrix}[(g_Q,k_Q);\AA_{16}]-\begin{pmatrix}
        \P_k \\ \J_l
    \end{pmatrix}[(\go,\ko);\AA_{16}]+\begin{pmatrix}
        n[Q]_{\P_k}\\ n[Q]_{\J_l}
    \end{pmatrix}.
\end{align*}
We also have
\begin{align*}
    \int\begin{pmatrix}
        \Lt_Q \\ \Kt_Q
    \end{pmatrix}dx&=\begin{pmatrix}
        \Q_E \\ \Q_B
    \end{pmatrix}[(g_Q,E_Q,B_Q);\AA_{16}]-\begin{pmatrix}
        \Q_E \\ \Q_B
    \end{pmatrix}[(\go,\Eo,\Bo);\AA_{16}].
\end{align*}
Thus, we have from \eqref{chargeconservation} that 
\begin{align*}
    \int (\Lt_Q,\Kt_Q)^\top dx=0.
\end{align*}
Moreover, the first two identities in \eqref{orthogonalityMQNQ} reduce to the following fixed-point problem:
\begin{equation}\label{eqfixed}
Q=\Qo+n[Q].
\end{equation}
It follows from Banach's fixed point theorem that there exists a unique $Q\in\QQ$ satisfying \eqref{eqfixed}. The Lipschitz dependence and persistence of regularity properties can be deduced similarly to Theorem 1.3 in \cite{MOT}. This concludes the proof of Theorem \ref{MOT1.3}.
\end{proof}
\subsection{Extension procedures}
We prepare the data for gluing through two extension steps. The \emph{outer extension} propagates the exterior EM data (the Brill-Lindquist family from Section~\ref{secBL-EM}) inward to the annulus $\At_{16}=\BB_{64}\setminus\ov{\BB_{8}}$. The \emph{inner extension} propagates the near-star annulus from Theorem~\ref{interiorsolution} outward up to $\BB_{8}$, with a conic correction $\BB_{8}\cup C_\theta(\bom_0)$, while keeping the Klein-Gordon field compactly supported in $\BB_2$. The actual gluing will then take place in the intermediate annulus $\AA_{16}=\BB_{32}\setminus\ov{\BB_{16}}$.
\subsubsection{Extension of outer solutions}
\begin{lem}\label{outexten}
Let $(g,k,E,B)_{out}\in\XX^s(\AA_{32})$ solve \eqref{EMconstraint}. If $\|(g-e,k,E,B)_{out}\|_{\XX^s(\AA_{32})}$ is sufficiently small, there exists $(\gt,\kt,\Et,\Bt)_{out}\in\XX^s(\At_{16})$ that solves \eqref{EMconstraint} and satisfies
\begin{align*}
(\gt,\kt,\Et,\Bt)_{out}&=(g,k,E,B)_{out}\quad\mbox{ in }\;\AA_{32},\\
\|(\gt-e,\kt,\Et,\Bt)_{out}\|_{\XX^s(\At_{16})}&\les\|(g-e,k,E,B)_{out}\|_{\XX^s(\AA_{32})},\\
\Q_E[(\gt,\Et)_{out};\AA_{16}]&=\Q_E[(g,E)_{out};\AA_{32}],\\
\Q_B[(\gt,\Bt)_{out};\AA_{16}]&=\Q_B[(g,B)_{out};\AA_{32}].
\end{align*}
\end{lem}
\begin{proof}
We first \emph{freely extend} $(g,k,E,B)_{out}$ to $\XX^s(\BB_{64})$ with control over the norm. Let $\vec{S}_{in}:=(S,T,P,P)$ be the Bogovskii type operator on $\BB_{32}$, which satisfies the properties in Lemma \ref{lem:annulus-bog}. We then consider the following fixed point problem:
\begin{align*}
\Xt=\vec{S}_{in}\big(\vec{N}_{EM}(X_{out}+\Xt)-\vec{P}(X_{out})\big).
\end{align*}
By the Banach fixed point theorem, there exists a unique solution $\Xt\in\XX^s(\BB_{32})$. Defining
\begin{align*}
\Xt_{out}:=X_{out}+\Xt,
\end{align*}
and let $(\gt,\kt,\Et,\Bt)_{out}$ be the solution corresponding to $\Xt_{out}=(\hti,\pit,\Egt,\Bgt)_{out}$. Moreover, the conservation of electric and magnetic charges follows directly from the facts that $\sdiv_gE=0$ and $\sdiv_gB=0$. This concludes the proof of Lemma \ref{outexten}.
\end{proof}
\subsubsection{Extension of inner solutions}
\begin{lem}\label{conicSc}
Let $\ell\in[-1,-\frac{1}{2}]$, $\th>0$ and $\bom_0\in\SSS^2$, there exists a solution operator $\vec{S}_{out}$ for $\vec{P}$ defined in \eqref{defvecPvecN}, i.e. $\vec{P}\vec{S}_{out}=I$, with the following mapping property:
\begin{align*}
\vec{S}_{out}:\big(H_b^{s-2,\ell+2}\big)_0(\AA_4\cup C_\th(\bom_0))\rightarrow\big(\XX_b^{s,\ell}\big)_0(\AA_4\cup C_\th(\bom_0)).
\end{align*}
\end{lem}
\begin{proof}
The proof is conducted by applying the conic operators defined in \eqref{lem:conic}, which is largely analogous to Step 1 in the proof of Lemma 5.8 in \cite{MOT}.
\end{proof}
\begin{lem}\label{inexten}
Let $(g,k,E,B)_{in}$ be an initial data set on $\AA_2$ solving \eqref{EMconstraint} and satisfying
\begin{equation*}
    \|(g-e,k,E,B)_{in}\|_{\XX^s(\AA_2)}\leq\eps,
\end{equation*}
where $\eps$ is a constant small enough. Let $\de\in[-1,-\frac{1}{2}]$, $\th>0$ and let $\bom_0\in\SSS^2$. Then, there exists $(\gt,\kt,\Et,\Bt)_{in}$ that solves \eqref{EMconstraint} and satisfies the following properties:
\begin{enumerate}
\item We have
\begin{align*}
(\gt,\kt,\Et,\Bt)_{in}&=(g,k,E,B)_{in}\quad \mbox{ on }\;\AA_2,\\
\supp(\gt-e,\kt,\Et,\Bt)_{in}&\subseteq\BB_8\cup C_\th(\bom_0).
\end{align*}
\item The following estimate holds:
\begin{align}
\begin{split}\label{innerconclusion}
\|(\gt-e,\kt,\Et,\Bt)_{in}\|_{\XX_b^{s,\de}}\les\eps.
\end{split}
\end{align}
\item The charges of $(\gt,\kt,\Et,\Bt)_{in}$ satisfy for $r\in(4,8)$:
\begin{align}
\begin{split}\label{innerchargeconclusion}
\big|\Q_{ADM}[(\gt,\kt)_{in};\pr\BB_r]-\Q_{ADM}[(g,k)_{in};\pr\BB_4]\big|&\les\eps^2,\\
\Q_E[(\gt,\Et)_{in};\pr\BB_r]-\Q_E[(g,E)_{in};\pr\BB_4]&=0,\\
\Q_B[(\gt,\Bt)_{in};\pr\BB_r]-\Q_B[(g,B)_{in};\pr\BB_4]&=0. 
\end{split}
\end{align}
\end{enumerate}
\end{lem}
\begin{proof}
We \emph{freely extend} $(g,k,E,B)_{in}$ to $\AA_4$ so that
\begin{align*}
\|(g-e,k,E,B)_{in}\|_{\XX_0^s(\AA_4)}\les\eps.
\end{align*}
we consider the following fixed point problem:
\begin{align*}
\Xt=\vec{S}_{out}\big(\vec{N}_{EM}(X_{in}+\Xt)-\vec{P}(X_{in})\big)\qquad\mbox{ on }\;\AA_4\cup C_\th(\bom_0),
\end{align*}
where $\vec{S}_{out}$ denotes the solution operator in Lemma \ref{conicSc}. By Banach's fixed point theorem, there exists a unique solution $\Xt\in\XX_b^{s,\ell}(\AA_4\cup C_\th(\bom_0))$. Defining $\Xt_{in}:=X_{in}+\Xt$ and letting $(\gt,\kt,\Et,\Bt)_{in}$ be the solution corresponding to $\Xt_{in}=(\hti,\pit,\Egt,\Bgt)_{in}$, we deduce immediately \eqref{innerconclusion}. Finally, \eqref{innerchargeconclusion} follows directly from Lemma \ref{Lemconservationlaw} and the fact that $\sdiv_gE=\sdiv_gB=0$. This concludes the proof of Lemma \ref{inexten}.
\end{proof}
\subsection{Main gluing theorem}\label{secobsfree}
Combining the reductions of Sections \ref{ssecinner} and \ref{secEMlinearglue}, we obtain the abstract gluing result for the EM constraints. In the electrovacuum region produced by the localization of the Klein-Gordon field, the Maxwell sector glues linearly, with the electromagnetic charges as the only compatibility conditions, while the Einstein sector glues obstruction-free as in~\cite{MOT}. This theorem provides the EM component of the full EMKG gluing scheme developed in this section, tailored to the boson-star setting where the Klein-Gordon field is compactly supported and the exterior is electrovacuum. 

We first recall the multi-bump construction of \cite{MOT} with prescribed ADM charges.
\begin{prop}\label{sixbump}
Let $\de\in[-1,-\frac{1}{2}]$ and $\th\in(0,\frac{\pi}{8})$, we define
\begin{align*}
C_\th^{(6)}=C_\th(\ev_1)\cup C_\th(-\ev_1)\cup C_\th(\ev_2)\cup C_\th(-\ev_2)\cup C_\th(\ev_3)\cup C_\th(-\ev_3).
\end{align*}
Given $s>\frac{3}{2}$ and $\Ga\geq 1$, there exist $\eps_b:=\eps_b(s,\Ga)>0$ and $\mu_b=\mu_b(s,\Ga)$ such that the following holds. Consider
\begin{align*}
\UU_\Ga=\left\{\Q\in\RRR^{10}\Big/\;\E>|\P|,\quad \frac{\E}{\sqrt{\E^2-|\P|^2}}<2\Ga,\quad \E<\eps_b^2,\quad |\C|+|\J|<\mu_b\E\right\}.
\end{align*}
For each $Q\in\UU_\Ga$, there exists $(g_Q^{bump;\Ga},k_Q^{bump;\Ga})\in C^\infty(\RRR^3)$ solving vacuum constraint equations such that
\begin{align}
\begin{split}\label{gQbumpGa}
\supp(g_Q^{bump;\Ga}-e,k_Q^{bump;\Ga})&\subseteq C_\th^{(6)}\cap\BB_8^c,\\
\Q_{ADM}[(g_Q^{bump;\Ga},k_Q^{bump;\Ga});\AA_{16}]&=Q.
\end{split}
\end{align}
Moreover, we have
\begin{align*}
\|(g_Q^{bump;\Ga}-e,k_Q^{bump;\Ga})\|_{H_b^{s,\de}\times H_b^{s-1,\de+1}}&\les\sqrt{\E},\\
\|\pr_Q(g_Q^{bump;\Ga},k_Q^{bump;\Ga})\|_{H_b^{s,\de}\times H_b^{s-1,\de+1}}&\les\frac{1}{\sqrt{\E}},\\
\|(g_Q^{bump;\Ga}-e,k_Q^{bump;\Ga})\|_{H_b^{s,\de}\times H_b^{s-1,\de+1}(\BB_{16}^c)}&\les\E,\\
\|\pr_Q(g_Q^{bump;\Ga},k_Q^{bump;\Ga})\|_{H_b^{s,\de}\times H_b^{s-1,\de+1}(\BB_{16}^c)}&\les 1.
\end{align*}
\end{prop}
\begin{proof}
See Proposition 5.7 in \cite{MOT}.
\end{proof}
We are now ready to state and prove the following gluing Theorem, which extends Theorem 1.7 of \cite{MOT} to the EM constraints \eqref{EMconstraint}. 
\begin{thm}\label{MOT1.7}
Given $s>\frac{3}{2}$ and $\Ga>1$, there exist $\eps_o=\eps_o(s,\Ga)>0$, $\mu_o=\mu_o(s,\Ga)>0$, and $C_{o} = C_o(s,\Ga)>0$ such that the following holds. Let $(g,k,E,B)_{in}\in\XX^s(\AA_2)$ and $(g,k,E,B)_{out}\in\XX^s(\AA_{32})$ be solutions to the EM constraint equations \eqref{EMconstraint}. We define $\De Q = (\De\E,\De\P,\De\C,\De\J)\in\RRR^{10}$ by
\begin{equation}\label{eq:Delta-Q}
\De Q=\Q_{ADM}[(g,k)_{out};\AA_{32}]-\Q_{ADM}[(g,k)_{in};\AA_2],
\end{equation}
and assume that
\begin{align}
\De\E&>|\De\P|,\label{eq:obs-free-unit:EP}\\
\frac{\De\E}{\sqrt{(\De\E)^{2}-|\De\P|^{2}}}&<\Ga,\label{eq:obs-free-unit:Gamma} \\
\De\E &<\eps_{o}^{2},\label{eq:obs-free-unit:ep} \\
|\De\C|+|\De\J|&<\mu_o\De\E,\label{eq:obs-free-unit:CJ}\\
\Q_E[(g,E)_{out};\AA_{32}]&=\Q_E[(g,E)_{in};\AA_{2}],\label{eq:QEconservation}\\
\Q_B[(g,B)_{out};\AA_{32}]&=\Q_B[(g,B)_{in};\AA_{2}],\label{eq:QBconservation}
\end{align}
and
\begin{equation}
\|(g-e,k,E,B)_{in}\|_{\XX^s(\AA_2)}^{2}+\|(g-e,k,E,B)_{out}\|_{\XX^s(\AA_{32})}^{2}<\mu_o\De\E\label{eq:obs-free-unit:data}.
\end{equation}
Then, there exists $(g,k,E,B)\in\XX^s(\BB_{64}\setminus\ov{\BB_{2}})$ solving \eqref{EMconstraint} such that
\begin{align*}
(g,k,E,B)&= (g,k,E,B)_{in} \quad\,\,\mbox{ on }\;\AA_2,\\
(g,k,E,B)&=(g,k,E,B)_{out} \quad\mbox{ on }\;\AA_{32}, 
\end{align*}
and
\begin{equation}\label{eq:obs-free-conc}
\|(g-e,k,E,B)\|_{\XX^s(\BB_{64}\setminus\ov{\BB_{2}})}^{2}<C_o\De\E.
\end{equation}
\end{thm}
\begin{proof}
Fix $-1<\de<-\frac{1}{2}$ and $\th\in(0,\frac{\pi}{8})$. Let
\begin{align*}
\UU_\Ga(\De\E)=\left\{\De Q'\in\UU_\Ga\Big/\,\frac{1}{2}\De\E\leq\E(\De Q')\leq 2\De\E\right\}.
\end{align*}
Given $\De Q'\in\UU_\Ga(\De\E)$, by Proposition \ref{sixbump}, there exist six-bump initial data $(g_Q^{bump;\Ga}-e,k_Q^{bump;\Ga})$ to the vacuum constraint equations, supported in $C_\th^{(6)}\setminus\ov{\BB_8}$. By Lemma \ref{inexten}, given any $\bom_0\in\SSS^2$, there exists $(\gt,\kt,\Et,\Bt)_{in}\in\XX_b^{s,\de}$ extending $(g,k,E,B)_{in}$ and satisfying
\begin{equation*}
(\gt-e,\kt,\Et,\Bt)_{in}\subseteq\BB_8\cup C_\th(\bom_0).
\end{equation*}
Taking $\bom_0\in\SSS^2$ such that
$$
\big(C_\th^{(6)}\setminus\ov{\BB_8}\big)\bigcap\left(C_\th(\bom_0)\cup\BB_8\right)=\emptyset,
$$
we then define, for any $\De Q'\in\UU_\Ga(\De\E)$
\begin{align*}
(g-e,k,E,B)_{\De Q'}:=(\gt-e,\kt,\Et,\Bt)_{in}+(g-e,k,0,0)^{bump;\Ga}_{\De Q'},
\end{align*}
which extends $(g,k,E,B)_{in}$ and solves \eqref{EMconstraint}. We have from Lemma \ref{Lemconservationlaw}, \eqref{gQbumpGa} and \eqref{eq:obs-free-unit:data}
\begin{align}
\begin{split}\label{DeQ'A16}
\Q_{ADM}[(g,k)_{\De Q'};\AA_{16}]&=\Q_{ADM}[(g,k)_{\De Q'}^{bump;\Ga};\AA_{16}]+\Q_{ADM}[(\gt,\kt)_{in};\AA_{16}]\\
&=\De Q'+\Q_{ADM}[(g,k)_{in};\AA_2]+\|(g-e,k,E,B)_{in}\|^2_{\XX^s(\BB_{32})} \\
&=\De Q'+\Q_{ADM}[(g,k)_{in};\AA_2]+O_\Ga(\mu_o\De\E).
\end{split}
\end{align}
We also have from \eqref{gQbumpGa}
\begin{align*}
\Q_{ADM}[(g,k)_{\De Q'};\AA_{16}]-\Q_{ADM}[(g,k)_{\De Q''};\AA_{16}]=\De Q'-\De Q''.
\end{align*}
In particular, the following map is bi-Lipschitz:
\begin{align*}
T:\UU_\Ga(\De\E)&\to\RRR^{10},\\
\De Q'&\mapsto\Q_{ADM}[(g,k)_{\De Q'};\AA_{16}].
\end{align*}
Next, by Lemma \ref{outexten}, there exists $(\gt,\kt,\Et,\Bt)_{out}$ on $\At_{16}$ extending $(g,k,E,B)_{out}$ and solving \eqref{EMconstraint}. Again, we have from Lemma \ref{Lemconservationlaw}
\begin{align}
\begin{split}\label{outdiff}
\quad\;\Q_{ADM}[(\gt,\kt)_{out};\AA_{16}]&=\Q_{ADM}[(g,k)_{out};\AA_{32}]+\|(g-e,k,E,B)_{out}\|^2_{\XX^s(\At_{16})}\\
&=\Q_{ADM}[(g,k)_{out};\AA_{32}]+O_\Ga(\mu_o\De\E).
\end{split}
\end{align}
Thus, we infer from \eqref{DeQ'A16}, \eqref{eq:Delta-Q} and \eqref{outdiff}
\begin{align*}
T(\De Q')&=\Q_{ADM}[(g,k)_{\De Q'};\AA_{16}]\\
&=\De Q'+\Q_{ADM}[(g,k)_{in};\AA_2]+O_\Ga(\mu_o\De\E)\\
&=\De Q'-\De Q+\Q_{ADM}[(g,k)_{out};\AA_{32}]+O_\Ga(\mu_o\De\E)\\
&=\Q_{ADM}[(\gt,\kt)_{out};\AA_{16}]+\De Q'-\De Q+O_\Ga(\mu_o\De\E).
\end{align*}
Noticing that $T(\De Q)=\Q_{ADM}[(\gt,\kt)_{out};\AA_{16}]+O_\Ga(\mu_o\De\E)$, choosing $\mu_o$ even smaller depending on $\mu_b$ and $\Ga$, we have, from the inverse function theorem, that $T(\UU_\Ga)$ covers a ball of radius $O_\Ga(\mu_o\De\E)$ around $\Q_{ADM}[(\gt,\kt)_{out};\AA_{16}]$. By inverting $T$, we obtain a $\QQ$--admissible family of annular initial data sets on $\At_{16}$ satisfying \eqref{Qoopen}--\eqref{gQgQ'kQkQ'} with
\begin{align*}
\eps=O_\Ga\big(\mu_o^\frac{1}{2}(\De\E)^\frac{1}{2}\big),\qquad\quad K=O_\Ga(1).
\end{align*}
Recalling that $\sdiv_gE=\sdiv_gB=0$ holds on $\BB_2^c$, we have
\begin{align*}
\Q_{E,B}[(g,E,B)_{\De Q'};\AA_{16}]&=\Q_{E,B}[(g,E,B)_{in};\AA_2]\\
&=\Q_{E,B}[(g,E,B)_{out};\AA_{32}]\\
&=\Q_{E,B}[(\gt,\,\Et,\Bt)_{out};\AA_{16}].
\end{align*}
For $\eps_{o}$ small enough, applying Theorem \ref{MOT1.3} with $(\go,\ko,\Eo,\Bo)=(\gt,\kt,\Et,\Bt)_{out}$, we obtain that there exists a solution of \eqref{EMconstraint} such that
\begin{align*}
(g,k,E,B)=\left\{\begin{aligned}
(g,k,E,B)_{\De Q'}\quad &\mbox{ in }\;\AA_8,\\
(\gt,\kt,\Et,\Bt)_{out}\quad\, &\mbox{ in }\;\AA_{32}.
\end{aligned}\right.
\end{align*}
Recalling that $(g,k,E,B)_{\De Q'}$ is an extension of $(g,k,E,B)_{in}$, while $(\gt,\kt,\Et,\Bt)_{out}$ is an extension of $(g,k,E,B)_{out}$, this concludes the proof of Theorem \ref{MOT1.7}.
\end{proof}
\section{Construction of Cauchy initial data}\label{secconstruction}
The goal of this section is to prove the following theorem, which establishes the existence of the desired Cauchy initial data.
\begin{thm}\label{mainCauchy}
Let $N\in\NNN$ and $s\geq 3$. Let $(g_{BL},0,E_{BL},0)$ be a charged Brill-Lindquist metric defined by
\begin{align*}
g_{BL}&=\left(1+\sum_{I=1}^{N}\frac{M_I+Q_I}{2|\x-\cb_I|}\right)^2\left(1+\sum_{I=1}^{N}\frac{M_I-Q_I}{2|\x-\cb_I|}\right)^2 e,\qquad\, k_{BL}=0,\\
E_{BL}&=-\nab_g\ln\left(\frac{1+\sum_{I=1}^N\frac{M_I+Q_I}{2|\x-\cb_I|}}{1+\sum_{I=1}^N\frac{M_I-Q_I}{2|\x-\cb_I|}}\right),\qquad\qquad\qquad\qquad\quad B_{BL}=0,
\end{align*}
where $\{M_I\}_{I=1}^N$ and $\{\cb_I\}_{I=1}^N$ satisfy the small-mass and large-separation condition \eqref{dfMd-EM}. Then, there exists a suitable choice of $\{Q_I\}_{I=1}^N$ and a $2N$--parameter $\{(\de_I,a_I)\}_{I=1}^N$ family of initial data sets $(\Si,g,k,E,B,\psi,\phi,A,\Phi)$, which are solutions of \eqref{EMKGconstraint} and satisfy the following properties:
\begin{enumerate}
    \item For any $I\in\{1,2,\dots,N\}$, we have
    \begin{align*}
    (g,k,E,B)&=(g_{BL},0,E_{BL},0)\qquad\quad\qquad\quad\,\mbox{ on }\;\BB_{32}^c(\cb_I),\\
    (g,k,E,B)&=(g_{\de_I,a_I},k_{\de_I,a_I},E_{\de_I,a_I},B_{\de_I,a_I})\quad\,\mbox{ on }\;\BB_\frac{3}{2}(\cb_I)\setminus\ov{\BB_1(\cb_I)},\\
    (g,k,E,B)&=(e,0,0,0)\qquad\qquad\quad\qquad\quad\,\,\,\;\mbox{ on }\;\BB_{1-2\de_I}(\cb_I),
    \end{align*}
    where $(\Si,g,k,E,B)_{\de_I,a_I}$ are the spacelike initial data constructed in Theorem \ref{interiorsolution} for parameter $(\de_I,a_I)$.
    \item The Klein-Gordon fields are supported in the star region:
    \begin{align*}
        \supp(\psi,\phi)\subseteq\bigcup_{I=1}^N\BB_2(\cb_I).
    \end{align*}
    \item For any $I\in\{1,2,\dots,N\}$, the following estimate holds:
    \begin{equation}\label{64-1control}
    \|(g-e,k,E,B)\|^2_{\XX^s(\BB_{64}(\cb_I)\setminus\ov{\BB_{1}(\cb_I)})}\les M_I.
    \end{equation}
    \item For any $I\in\{1,2,\dots,N\}$, a trapped surface will form in $D^+(\BB_1(\cb_I))$, the future domain of dependence of $\BB_1(\cb_I)$.
    \item The Cauchy data $(\Si,g,k)$ are free of trapped surfaces.
    \item For any $I\in\{1,2,\dots,N\}$, the coordinate ball $\BB_{64}(\cb_I)$ is free of MOTS.
    \end{enumerate}
    See Figure \ref{IntroID} for a geometric illustration of the Cauchy data in the particular case $N=3$.
\end{thm}
\begin{proof}
The idea of the proof is to perform surgery for the charged Brill-Lindquist metric $(g_{BL},0,E_{BL},0)$ near $\{\cb_I\}_{I=1}^N$. More precisely, we apply Theorem \ref{MOT1.7} to patch the constant-time slices $\Si(\de_I,a_I)$ constructed in Theorem \ref{interiorsolution} near the poles $\{\cb_I\}_{I=1}^N$. The proof is divided into 3 steps.
\paragraph{Step 1. Gluing construction of Cauchy data.}
Fix $I\in\{1,2,\cdots,N\}$, we denote
\begin{align*}
(g,k,E,B)_{out}:=(g_{BL},0,E_{BL},0).
\end{align*}
Let $(x^1,x^2,x^3)$ be a coordinate system centered at $\cb_I$. We have from Corollary \ref{cor:BL-annulus-EM}
\begin{align}
\begin{split}\label{QgkEBout}
\E[(g,k)_{out};\AA_{32}(\cb_I)]&=8\pi M_I+O(M_IM+Md_I^{-1}),\\
\P_l[(g,k)_{out};\AA_{32}(\cb_I)]&=0,\\
\C_l[(g,k)_{out};\AA_{32}(\cb_I)]&=O(M_IM+Md_I^{-1}),\\
\J_l[(g,k)_{out};\AA_{32}(\cb_I)]&=0,\\
\Q_E[(g,E)_{out};\AA_{32}(\cb_I)]&=4\pi Q_I+O(Q_I M + M d_I^{-1}),\\
\Q_B[(g,B)_{out};\AA_{32}(\cb_I)]&=0.
\end{split}
\end{align}
Next, let $\Si(\de_I,a_I)$ be the constant-time slice constructed in Theorem \ref{interiorsolution}. We denote 
$$
(g,k,E,B,\psi,\phi,A,\Phi)_{in}:=(\gt,\kt,\Et,\Bt,\psit,\phit,\AAt,\Phit)
$$
constructed from $\Si(\de_I,a_I)$ as in Theorem \ref{constructioninnerannulus}. As an immediate consequence of \eqref{innerannQ}, we have
\begin{align*}
\Q_{ADM}[(g,k)_{in};\AA_2(\cb_I)]&=O(a_I^{-1}),\\
\Q_E[(g,E)_{in};\AA_2(\cb_I)]&=O(a_I^{-1}),\\
\Q_B[(g,B)_{in};\AA_2(\cb_I)]&=0.
\end{align*}
Thus, we obtain
\begin{align}\label{DeEestimate}
\begin{split}
\De\E&=\E[(g,k)_{out};\AA_{32}(\cb_I)]-\E[(g,k)_{in};\AA_2(\cb_I)]\\
&=8\pi M_I+O(M_IM+Md_I^{-1}+a_I^{-1})\simeq M_I,
\end{split}
\end{align}
where we used the facts that $M\ll 1$, $d_I^{-1}\ll M_I$ and $a_I^{-1}\ll M_I$. We also have
\begin{align}\label{DePCJestimate}
(\De\P,\De\C,\De\J)=O(M_IM+Md_I^{-1}+a_I^{-1}).
\end{align}
As an immediate consequence of \eqref{DeEestimate} and \eqref{DePCJestimate}, we infer
\begin{align*}
    \De\E\gg |\De\P|,\qquad\quad\De\E\ll 1,\qquad\quad|\De\C|+|\De\J|\ll\De\E.
\end{align*}
These imply that the conditions \eqref{eq:obs-free-unit:EP}--\eqref{eq:obs-free-unit:CJ} in Theorem \ref{MOT1.7} hold. Next, we have from \eqref{QgkEBout}
\begin{align*}
    \Q_E\big[(g,E)_{out};\AA_{32}(\cb_I)\big]&=4\pi Q_I+F_{out}^{(I)}(Q_1,\cdots Q_N),\\
    \Q_E\big[(g,E)_{in};\pr\BB_4(\cb_I)\big]&=F_{in}^{(I)}(Q_1,\cdots Q_N),
\end{align*}
with $F_{out}$ and $F_{in}$ are functions of $\{Q_1\}_{I=1}^N$ that satisfy:
\begin{equation*}
    F_{out}^{(I)}(Q_1,\cdots Q_N)=O(Q_IM+Md_I^{-1}),\qquad F_{in}^{(I)}(Q_1,\cdots Q_N)=O(a_I^{-1}).
\end{equation*}
By the implicit function theorem, there exists a suitable choice of $\{Q_I\}_{I=1}^N$ such that
\begin{align*}
    \Q_E\big[(g,E)_{out};\AA_{32}(\cb_I)\big]=\Q_E\big[(g,E)_{in};\AA_2(\cb_I)\big].
\end{align*}
Moreover, we have trivially
\begin{align*}
\Q_B\big[(g,B)_{out};\AA_{32}(\cb_I)\big]=0=\Q_B\big[(g,B)_{in};\AA_2(\cb_I)\big].
\end{align*}
We also have from \eqref{innerannSobolev} and \eqref{eq:Sobolev-g-EM}
\begin{align*}
    \|(g-e,k,E,B)_{in}\|_{\XX^s(\AA_2(\cb_I))}\les a_I^{-1},\qquad \|(g-e,k,E,B)_{out}\|_{\XX^s(\AA_{32}(\cb_I))}\les M_I.
\end{align*}
Thus, we obtain for $a_I^{-1}\ll M_I\ll 1$
\begin{align*}
    \|(g-e,k,E,B)_{in}\|_{\XX^s(\AA_2(\cb_I))}^2+\|(g-e,k,E,B)_{out}\|_{\XX^s(\AA_{32}(\cb_I))}^2\les a_I^{-2}+M_I^2\ll M_I\simeq \De\E,
\end{align*}
which implies that \eqref{eq:obs-free-conc} in Theorem \ref{MOT1.7} holds. Thus, for any $I\in\{1,2,\cdots,N\}$, all the conditions in Theorem \ref{MOT1.7} hold near $\cb_I$, we have from Theorem \ref{MOT1.7} that there exists a solution of \eqref{EMKGconstraint} satisfies the properties 1--3 of Theorem \ref{mainCauchy}.
\paragraph{Step 2. Formation of multiple trapped surfaces.}
For any $I\in\{1,2,\dots,N\}$, the ball $\BB_1(\cb_I)$ is obtained from the evolution of the characteristic initial data of EMKG system in Theorem \ref{thmtrapped}. Noticing that the future domain of dependence $D^+(\BB_1(\cb_I))$ is causally independent of $\BB_1^c(\cb_I)$, we have from Theorem \ref{thmtrapped} that a trapped surface will form in $D^+(\BB_1(\cb_I))$. Hence, $N$ trapped surfaces will form, respectively, in $D^+(\BB_1(\cb_I))$ for $I\in\{1,2,\dots,N\}$, which yields property 4 in Theorem \ref{mainCauchy}.
\paragraph{Step 3. Free of trapped surfaces.}
Fix $I\in\{1,\dots,N\}$. By the Sobolev control \eqref{64-1control} in annulus $\BB_{64}(\cb_I)\setminus\ov{\BB_{1}(\cb_I)}$,
\[
\|(g-e,k)\|^2_{H^s\times H^{s-1}(\BB_{64}(\cb_I)\setminus\ov{\BB_{1}(\cb_I)})}\les M_I,
\]
hence for $s\geq 3$, along the radial foliation $r=|\x|$,
\[
|k|\les M^{1/2},\qquad \Big|\tr_\slg\th-\frac{2}{r}\Big|\les M^{1/2}
\quad\text{ on }\;\pr\BB_r(\cb_I),\quad \forall\; 1<r<64.
\]
For $M\ll 1$, this implies
\begin{align}\label{poscurvature-barrier}
   \tr_\slg\th>0.03>0\quad\text{ on }\;\pr\BB_r(\cb_I),\quad\forall\;1<r<64.
\end{align}
In the short-pulse region and Minkowski region, we have by \eqref{trchctrchbcnotrapping} and $a_I \gg 1$,
\begin{align}\label{poscurvature-inner}
    \tr_\slg\th = \trch - \trchb > 1 \quad\text{ on }\;\pr\BB_r(\cb_I),\quad \forall\; 0<r\le1.
\end{align}
Combining \eqref{poscurvature-barrier} and \eqref{poscurvature-inner},
\begin{align}\label{trgth003}
    \tr_\slg\th>0.03>0\quad\text{on }\pr\BB_r(\cb_I),\;0<r<64.
\end{align}
Let $S\subset\BB_{64}(\cb_I)$ be a compact, embedded smooth $2$--surface, and denote by $\BB_{r_S}(\cb_I)$ the innermost ball centered at $\cb_I$ containing $S$. Then $S$ and $S_{r_S}:=\pr\BB_{r_S}(\cb_I)$ are tangent at some $p$, and by a standard mean curvature comparison,
\begin{align*}
    \tr_\slg\th_{S}(p)\geq\tr_\slg\th_{S_{r_S}}(p).
\end{align*}
Using \eqref{trgth003} and $|k|\les M^{1/2}\ll1$,
\begin{align*}
    \tr_\slg(\th_S-k_S)(p)>0.02>0,
\end{align*}
so $S$ is neither trapped nor marginally trapped. It remains to rule out trapped surfaces in the exterior region $\Si_{ext}:=\bigcup_{I=1}^N\BB_{32}^c(\cb_I)$. If
$S\subset\Si$ satisfies $S\cap\Si_{ext}\neq\emptyset$, then for some $p\in S\cap\Si_{ext}$ we have $k=0$, and hence
\begin{align*}
    \tr_\slg(-\th-k)\tr_\slg(\th-k)(p)
    =-(\tr_\slg\th)^2(p)\le0,
\end{align*}
which contradicts \eqref{dftrapped}. Thus, no trapped surface intersects $\Si_{ext}$. This proves Properties 5 and 6 of Theorem \ref{mainCauchy}. Combined with Steps 1--3, this completes the proof of Theorem \ref{mainCauchy}.
\end{proof}
\appendix
\section{Double null foliation}\label{doublenullfoliation}
We consider a double null foliation of $\MM$ by two future-increasing optical functions $u,\ub$ satisfying
\[
    \g(\grad u,\grad u)=\g(\grad\ub,\grad\ub)=0.
\]
The level sets $H_u:=\{u=\mathrm{const}\}$ and $\Hb_{\ub}:=\{\ub=\mathrm{const}\}$ are outgoing and ingoing null hypersurfaces, and their intersections
\begin{equation}\label{eq:def-S-u-ub}
    S_{u,\ub}:=H_u\cap\Hb_{\ub}
\end{equation}
are spacelike $2$--spheres. Set $L:=-\grad u$, $\Lb:=-\grad\ub$,
and define the lapse $\Om>0$ by
\begin{equation*}
    \g(L,\Lb)=-2\Om^{-2}.
\end{equation*}
The normalized null pairs are
\[
e_4:=\Om L,\qquad e_3:=\Om\Lb,\qquad  N:=\Om e_4,\qquad \Nb:=\Om e_3.
\]
On each $S_{u,\ub}$, choose a local orthonormal frame $(e_1,e_2)$ so that $(e_1,e_2,e_3,e_4)$ is a null frame (with $e_A\in\{e_1,e_2\}$). The induced Riemannian metric on $S_{u,\ub}$ is $\slg$, with Levi--Civita connection $\nabs$. Using $(u,\ub)$ together with local coordinates $x^A$ on $S_{u,\ub}$ satisfying $e_4(x^A)=0$, the metric takes the double null form
\begin{equation*}
\g=-2\Om^2(d\ub\otimes du+du\otimes d\ub)+\slg_{AB}(dx^A-\bbb^A du)\otimes(dx^B-\bbb^B du),
\end{equation*}
with
\[
\Nb=\pr_u+\bbb,\qquad N=\pr_{\ub},\qquad \bbb:=\bbb^A\pr_{x^A}.
\]
\subsection{Ricci coefficients and Weyl components}
We recall the null decomposition of the Ricci coefficients and Weyl components of the null frame $(e_1,e_2,e_3,e_4)$ as follows:
\begin{align*}
\chi_{AB}&=\g(\bnab_A e_4, e_B),\qquad \chib_{AB}=\g(\bnab_A e_3, e_B),\qquad\quad \xi_A=\frac 1 2 \g(\bnab_4 e_4,e_A),\\ 
\xib_A&=\frac 1 2 \g(\bnab_3 e_3, e_A),\qquad\;\;\,\om=\frac 1 4 \g(\bnab_4 e_4, e_3),\qquad\quad\, \omb=\frac 1 4 \g(\bnab_3e_3 ,e_4),\\
\eta_A&=\frac 1 2 \g(\bnab_3 e_4, e_A),\qquad\; \etab_A=\frac 1 2 \g(\bnab_4 e_3, e_A),\qquad\;\; \ze_A=\frac 1 2 \g(\bnab_{A}e_4, e_3),
\end{align*}
and
\begin{align*}
\a_{AB} &= \W(e_A, e_4, e_B, e_4), 
\qquad 
\b_A = \frac12 \W(e_A, e_4, e_3, e_4), 
\qquad 
\rho = \frac14 \W(e_3, e_4, e_3, e_4), \\
\aa_{AB} &= \W(e_A, e_3, e_B, e_3), 
\qquad 
\bb_A = \frac12 \W(e_A, e_3, e_3, e_4), 
\qquad 
\si = \frac14 {^* \W}(e_3, e_4, e_3, e_4),
\end{align*}
where $^*\W$ denotes the Hodge dual of the Weyl tensor $\W$. The null second fundamental forms $\chi, \chib$ are further decomposed into their traces $\trch$ and $\trchb$, and their traceless parts $\hch$ and $\hchb$:
\begin{align*}
\trch:=\de^{AB}\chi_{AB},\quad\hch_{AB}:=\chi_{AB}-\frac{1}{2}\de_{AB}\trch,\qquad
\trchb:=\de^{AB}\chib_{AB},\quad \hchb_{AB}:=\chib_{AB}-\frac{1}{2}\de_{AB}\trchb.
\end{align*}
We define the horizontal covariant operator $\nabs$ as follows:
\begin{equation*}
\nabs_X Y:=\bnab_X Y-\frac{1}{2}\chib(X,Y)e_4-\frac{1}{2}\chi(X,Y)e_3.
\end{equation*}
We also define $\nabs_4 X$ and $\nabs_3 X$ to be the horizontal projections:
\begin{align*}
\nabs_4 X&:=\bnab_4 X-\frac{1}{2} \g(X,\bnab_4e_3)e_4-\frac{1}{2} \g(X,\bnab_4e_4)e_3,\\
\nabs_3 X&:=\bnab_3 X-\frac{1}{2} \g(X,\bnab_3e_3)e_3-\frac{1}{2} \g(X,\bnab_3e_4)e_4.
\end{align*}
A tensor field $\phi$ defined on $\MM$ is called tangent to $S$ if it is a priori defined on the spacetime $\M$ and all possible contractions of $\phi$ with either $e_3$ or $e_4$ are zero. We use $\nabs_3 \phi$ and $\nabs_4 \phi$ to denote the projection to $S_{u,\ub}$ of the usual derivatives $\bnab_3\phi$ and $\bnab_4\phi$. 

The following identities hold for a double null foliation:
\begin{align}
\begin{split}\label{nullidentities}
    \nabs\log\Om&=\frac{1}{2}(\eta+\etab),\qquad\;\;\;\,\om=-\frac{1}{2}\nabs_4(\log\Om), \qquad\;\;\;\omb=-\frac{1}{2}\nabs_3(\log\Om),\\ 
    \eta&=\ze+\nabs\log\Om,\qquad \etab=-\zeta+\nabs\log\Om,\qquad\quad \xi=\xib=0.
\end{split}
\end{align}
see for example (6) in \cite{kr}.
\subsection{Null components of matter fields}\label{secmatterfields}
Let $\F$ be the \emph{Faraday tensor}. We define the following null decomposition of $\F$:
\begin{equation}\label{dfF-components}
 \bF_{A}:=\F(e_A,e_4),\quad \bbF_{A}:=\F(e_A,e_3),\quad\rhoF:=\frac{1}{2}\F(e_3,e_4),\quad \siF:=\frac{1}{2}\dual\F(e_3,e_4).
\end{equation}
In particular, $\F(e_A,e_B)=-\slep_{AB}\siF$. Moreover, there exists an electromagnetic potential $\A$ that satisfies $\F=d\A$. In terms of null components, this reads
\begin{align*}
    \bnab_\mu\A_\nu-\bnab_\nu\A_\mu=\F_{\mu\nu}.
\end{align*}
We denote the components of the electromagnetic potential $\A$ as follows:
\begin{equation}\label{dfUUbAsl}
    U:=\A_4=\A(e_4),\qquad\Ub:=\A_3=\A(e_3),\qquad\Asl_B:=\A_B=\A(e_B).
\end{equation}
Moreover, we define $\Psi$ as the gauge covariant derivative $\D\psi$ of the complex scalar field $\psi$ and the following null components of $\Psi$:
\begin{equation*}
    \Psi_4=e_3\psi+i\ef\Ub\psi,\qquad \Psi_3=e_4\psi+i\ef U\psi,\qquad \Psisl_B=\nabs_B\psi+i\ef\Asl_B\psi.
\end{equation*}
\subsection{Ricci curvature and Schouten tensor}
We define the following \emph{Schouten tensor} $\S$ and \emph{Cotton tensor} $J$:
\begin{align}
\begin{split}\label{dfSchouten}
    \S_{\mu\nu}:=\Ric_{\mu\nu}-\frac{1}{6}\g_{\mu\nu}\R,\qquad\quad J_{\la\mu\nu}:=\frac{1}{2}(\bnab_\mu\S_{\la\nu}-\bnab_\nu\S_{\la\mu}).
\end{split}
\end{align}
\begin{prop}\label{Ricciexpression}
We have the following expressions for the null components of Ricci curvature and scalar curvature:
\begin{align*}
\Ric_{AB}&=2\Re(\Psisl_A\Psisl_B^\dag)-2\left(\bF\hot\bbF\right)_{AB}+\slg_{AB}\left(\rhoF^2+\siF^2+V(|\psi|^2)\right),\\
\Ric_{A3}&=2\Re(\Psisl_A\Psi_3^\dag)-2\rhoF\bbF_A+2\siF{}^*\bbF_A,\\
\Ric_{A4}&=2\Re(\Psisl_A\Psi_4^\dag)+2\rhoF\bF_A+2\siF{}^*\bF_A,\\
\Ric_{33}&=2|\Psi_3|^2+2|\bbF|^2,\\
\Ric_{34}&=2\Re(\Psi_3\Psi_4^\dag)+2\rhoF^2+2\siF^2-2V(|\psi|^2),\\
\Ric_{44}&=2|\Psi_4|^2+2|\bF|^2,\\
\R&=-2\Re(\Psi_3\Psi_4^\dag)+2|\Psisl|^2+4V(|\psi|^2).
\end{align*}
The null components of the Schouten tensor $\S_{\mu\nu}$ can be computed similarly.
\end{prop}
\begin{proof}
See Lemmas 4.12 and 4.13 in \cite{ShenWan}.
\end{proof}
\subsection{Commutation identities}
We recall the following commutation formulae.
\begin{lem}\label{comm}
Let $U_{A_1...A_k}$ be an $S$-tangent $k$-covariant tensor on $(\M,\g)$. Then
\begin{align*}
    [\Om\nabs_4,\nabs_B]U_{A_1...A_k}&=-\Om\chi_{BC}\nabs_CU_{A_1...A_k}+\sum_{i=1}^k \Om(\chi_{A_iB}\,\etab_C-\chi_{BC}\,\etab_{A_i}+\ins_{A_iC}{^*\b}_B)U_{A_1...C...A_k},\\
    [\Om\nabs_3,\nabs_B]U_{A_1...A_k}&=-\Om\chib_{BC}\nabs_C U_{A_1...A_k}+\sum_{i=1}^k\Om(\chib_{A_iB}\,\eta_C-\chib_{BC}\,\eta_{A_i}+\ins_{A_iC}{^*\bb}_B)U_{A_1...C...A_k},\\
    [\Om\nabs_3,\Om\nabs_4]U_{A_1...A_k}&=4\Om^2\ze_B\nabs_B U_{A_1...A_k}+2\Om^2\sum_{i=1}^k(\etab_{A_i}\,\eta_C-\etab_{A_i}\,\eta_C+\ins_{A_iC}\si)U_{A_1...C...A_k}.
\end{align*}
\end{lem}
\begin{proof}
It is a direct consequence of Lemma 7.3.3 in \cite{Ch-Kl} and \eqref{nullidentities}.
\end{proof}
\subsection{Einstein-Maxwell-Klein-Gordon equations}
We collect here the system of equations for the Ricci coefficients, electromagnetic components, and curvature components obtained as a consequence of the EMKG system \eqref{EMKG-SP}.
\subsubsection{Null structure equations}\label{sec-nullstr}
\begin{prop}\label{nulles}
We have the following null structure equations:
\begin{align*}
\nabs_4\eta&=-\chi\c(\eta-\etab)-\b-\frac{1}{2}\Ssl_4,\\
\nabs_3\etab&=-\chib\c(\etab-\eta)+\bb-\frac{1}{2}\Ssl_3,\\
\nabs_4\hch+(\trch)\hch&=-2\om\hch-\a,\\
\nabs_4\trch+\frac{1}{2}(\trch)^2&=-|\hch|^2-2\om\trch-\S_{44},\\
\nabs_3\hchb+(\trchb)\hchb&=-2\omb\hchb-\aa,\\
\nabs_3\trchb+\frac{1}{2}(\trchb)^2&=-|\hchb|^2-2\omb\trchb-\S_{33},\\
\nabs_4\hchb+\frac{1}{2}(\trch)\hchb&=\nabs\hot\etab+2\om\hchb-\frac{1}{2}\trchb\,\hch+\etab\hot\etab+\Sslh,\\
\nabs_3\hch+\frac{1}{2}(\trchb)\hch&=\nabs\hot\eta+2\omb\hch-\frac{1}{2}\trch\,\hchb+\eta\hot\eta+\Sslh,\\
\nabs_4\trchb+\frac{1}{2}(\trch)\trchb&=2\om\trchb+2\rho-\hch\c\hchb+2\sdivs\etab+2|\etab|^2+\Tr\S,\\
\nabs_3\trch+\frac{1}{2}(\trchb)\trch&=2\omb\trch+2\rho-\hch\c\hchb+2\sdivs\eta+2|\eta|^2+\Tr\S,
\end{align*}
where we denote
\begin{align*}
(\Ssl_3)_A:=\S_{3A},\qquad (\Ssl_4)_A:=\S_{4A},\qquad \Ssl_{AB}:=\S_{AB},\qquad \Tr\S:=\g^{\mu\nu}\S_{\mu\nu}.
\end{align*}
and $\Sslh$ denotes the traceless part of $\Ssl$. We also have the Codazzi equations:
\begin{align}\label{codazzi}
\begin{split}
\sdivs\hch&=\frac{1}{2}\nabs\trch-\ze\c\left(\hch-\frac{1}{2}\trch\right)-\b+\frac{1}{2}\Ssl_4,\\ 
\sdivs\hchb&=\frac{1}{2}\nabs\trchb+\ze\c\left(\hchb-\frac{1}{2}\trchb\right)+\bb+\frac{1}{2}\Ssl_3,
\end{split}
\end{align}
the torsion equation:
\begin{equation*}
\curls\eta=-\curls\etab=\si-\frac{1}{2}\hch\wedge\hchb,
\end{equation*}
and the Gauss equation:
\begin{equation*}
\K=-\frac{1}{4}\trch\trchb+\frac{1}{2}\hch\c\hchb-\rho+\frac{1}{2}\tr\Ssl.
\end{equation*}
Moreover, we have
\begin{align*}
\nabs_4\omb&=2\om\omb+\frac{3}{4}|\eta-\etab|^2-\frac{1}{4}(\eta-\etab)\c(\eta+\etab)-\frac{1}{8}|\eta+\etab|^2+\frac{1}{2}\rho+\frac{1}{4}(\tr\Ssl-\Tr\S),\\
\nabs_3\om&=2\om\omb+\frac{3}{4}|\eta-\etab|^2+\frac{1}{4}(\eta-\etab)\c(\eta+\etab)-\frac{1}{8}|\eta+\etab|^2+\frac{1}{2}\rho+\frac{1}{4}(\tr\Ssl-\Tr\S).
\end{align*}
\end{prop}
\begin{proof}
See Proposition 2.13 in \cite{ShenWan}.
\end{proof}
\subsubsection{Bianchi equations}\label{sec-Bianchi}
\begin{prop}\label{bianchiequations}
We have the following Bianchi equations:
\begin{align*}
\nabs_3\a+\frac{1}{2}\trchb\,\a&=\nabs\hot\b+4\omb\a-3(\hch\rho+{^*\hch}\si)+(\ze+4\eta)\hot\b+\slJ_4,\\
\nabs_4\b+2\trch\,\b&=\sdivs\a-2\om\b+(2\ze+\etab)\c\a-J_{4\bu4},\\
\nabs_3\b+\trchb\,\b&=\nabs\rho+{^*\nabs}\si+2\omb\b+2\hch\c\bb+3(\eta\rho+{^*\eta}\si)+J_{3\bu4},\\
\nabs_4\rho+\frac{3}{2}\trch\,\rho&=\sdivs\b-\frac{1}{2}\hchb\c\a+\ze\c\b+2\etab\c\b-\frac{1}{2}J_{434},\\
\nabs_4\si+\frac{3}{2}\trch\,\si&=-\curls\b+\frac{1}{2}\hchb\c{^*\a}-\ze\c{^*\b}-2\etab\c{^*\b}-\frac{1}{2}{}^*J_{434},\\
\nabs_3\rho+\frac{3}{2}\trchb\,\rho&=-\sdivs\bb-\frac{1}{2}\hch\c\aa+\ze\c\bb-2\eta\c\bb-\frac{1}{2}J_{343},\\
\nabs_3\si+\frac{3}{2}\trchb\,\si&=-\curls\bb+\frac{1}{2}\hch\c{^*\aa}-\ze\c{^*\bb}-2\eta\c{^*\bb}+\frac{1}{2}{}^*J_{343},\\
\nabs_4\bb+\trch\,\bb&=-\nabs\rho+{^*\nabs\si}+2\om\bb+2\hchb\c\b-3(\etab\rho-{^*\etab}\si)-J_{4\bu3},\\
\nabs_3\bb+2\trchb\,\bb&=-\sdivs\aa-2\omb\bb+(2\ze-\eta)\c\aa+J_{3\bu3},\\
\nabs_4\aa+\frac{1}{2}\trch\,\aa&=-\nabs\hot\bb+4\om\aa-3(\hchb\rho-{^*\hchb}\si)+(\ze-4\etab)\hot\bb+\slJ_3,
\end{align*}
where we denote
\begin{align*}
(\slJ_4)_{AB}&:=J_{AB4}+J_{BA4}-\frac{1}{2}J_{434}\,\slg_{AB},\qquad\;\; (\slJ_3)_{AB}:=J_{AB3}+J_{BA3}-\frac{1}{2}J_{343}\,\slg_{AB},\\
\dual J_{434}&:=(J_{AB4}-J_{BA4})\ins^{AB},\qquad\qquad\qquad\,\,\dual J_{343}:=(J_{AB3}-J_{BA3})\ins^{AB},\\
(J_{\la\bu\nu})_A&:=J_{\la A\nu},
\end{align*}
with $J_{\la\mu\nu}$ the Cotton tensor defined by \eqref{dfSchouten}.
\end{prop}
\begin{proof}
See Section 2.3.4 in \cite{Giorgi}.
\end{proof}
\subsubsection{Maxwell equations and Electromagnetic potential equations}
Recall $\F=d\A$ for the electromagnetic potential $\A$, with null components given in \eqref{dfF-components} and \eqref{dfUUbAsl}.
\begin{prop}\label{Maxwellequations}
The Maxwell equations take the following form:
\begin{align*}
\nabs_3\bF+\frac{1}{2}\trchb\bF&=-\sld_1^*(\rhoF,\siF)+2\omb\bF+2\rhoF\eta-2\siF\dual\eta\\
&+\hch\c\bbF-2\ef\Im\big(\psi\Psisl^\dag\big),\\
\nabs_4(\rhoF,\siF)+\trch(\rhoF,\siF)&=\sld_1\bF+(\etab+\ze)\c\big(\bF,{}^*\bF\big)+2\ef\big(\Im\big(\psi\Psi_4^\dag\big),0\big),\\
\nabs_3(\rhoF,-\siF)+\trchb(\rhoF,-\siF)&=-\sld_1\bbF-(\eta-\ze)\c\big(\bbF,{}^*\bbF\big)-2\ef\big(\Im\big(\psi\Psit_3^\dag\big),0\big),\\
\nabs_4\bbF+\frac{1}{2}\trch\bbF&=\sld_1^*(\rhoF,-\siF)+2\om\bbF-2\rhoF\etab-2\siF\dual\etab\\
&+\hchb\c\bF-2\ef\Im\big(\psi\Psisl^\dag\big).
\end{align*}
\end{prop}
\begin{proof}
See Proposition 2.21 in \cite{ShenWan}.
\end{proof}
\begin{prop}\label{propdA=F}
Let $U$, $\Ub$ and $\slA$ be defined in \eqref{dfUUbAsl}. Under the gauge choice $U=0$, the null components of $\F$ can be expressed in terms of the null components of $\A$ as follows:
\begin{align*}
    \rhoF&=-\frac{1}{2}\nabs_4\Ub+\om\Ub+(\etab-\eta)\c\Asl,\qquad\quad\siF=-\curls\Asl,\\
    \bF&=-\nabs_4\Asl-\frac{1}{2}\trch\Asl-\hch\c\Asl,\quad \bbF=\nabs \Ub-\nabs_3\Asl-\frac{1}{2}\trchb\Asl-\hchb\c\Asl+\Ub(-\ze+\eta).
\end{align*}
\end{prop}
\begin{proof}
See Corollary 2.24 in \cite{ShenWan}.
\end{proof}
\subsubsection{Klein-Gordon equations}
\begin{prop}\label{waveequation}
Under the gauge choice $U=0$, we have the following equations:
\begin{align*}
\nabs_3(\Psi_4)+\frac{1}{2}\trchb\Psi_4-2\omb\Psi_4+i\ef \Ub\Psi_4&=\sdivs\Psisl+2\eta\c\Psisl+i\ef \Asl\c\Psisl-\frac{1}{2}\trch\Psi_3\\
&+i\ef\rhoF\psi-V'(|\psi|^2)\psi,\\
\nabs_4\Psisl+\frac{1}{2}\trch \Psisl+\hch\c\Psisl&=\nabs(\Psi_4)+\frac{\eta+\etab}{2}\Psi_4+i\ef\Asl\Psi_4-i\ef\bF\psi,\\
\nabs_3\Psisl+\frac{1}{2}\trchb \Psisl+\hchb\c\Psisl+i\ef \Ub\Psisl&=\nabs(\Psi_3)+\frac{\eta+\etab}{2}\Psi_3+i\ef\Asl\Psi_3-i\ef\bbF\psi,\\
\nabs_4(\Psi_3)+\frac{1}{2}\trch\Psi_3-2\om\Psi_3&=\sdivs\Psisl+2\etab\c\Psisl+i\ef\Asl\c\Psisl-\frac{1}{2}\trchb\Psi_4\\
&-i\ef\rhoF\psi-V'(|\psi|^2)\psi,\\
\curls\Psisl&=-i\ef\Asl\wedge\Psisl-i\ef\siF\psi.
\end{align*}
We also have
\begin{align*}
    \nabs_3\Psisl+\trchb \Psisl+\hchb\c\Psisl+i\ef \Ub\Psisl&=\nabs(\Psit_3)+\frac{\eta+\etab}{2}\Psit_3+\frac{1}{2}\trchbt\Psisl+i\ef \Asl \Psit_3-i\ef\bbF\psi,\\
    \nabs_4(\Psit_3)+\frac{1}{2}\trch\Psit_3-2\om\Psit_3&=\sdivs\Psisl+2\etab\c\Psisl+i\ef\Asl\c\Psisl-\frac{1}{2}\trchbt\Psi_4-\frac{\trch}{2\Om|u|}\psi\\
    &-i\ef\rhoF\psi-V'(|\psi|^2)\psi,
\end{align*}
where we denote
\begin{equation}\label{dfPsit3}
    \Psit_3:=|u|^{-1}e_3(|u|\psi)+i\ef\Ub\psi.
\end{equation}
\end{prop}
\begin{proof}
The derivation is identical to Proposition 2.25 in \cite{ShenWan}, except that the potential contribution $V'(|\psi|^2)\psi$ is kept throughout the computation.
\end{proof}

\end{document}